\documentclass[11pt]{article}
\usepackage{asherspector}

\usepackage[margin=0.75in]{geometry}
\usepackage{color}
\usepackage{graphicx}
\usepackage{subcaption}
\usepackage[font=small,skip=0pt]{caption}

\usepackage{rotating}
\usepackage{verbatim}
\usepackage{hyperref}
\usepackage{bm}
\usepackage[normalem]{ulem}
\usepackage[authoryear]{natbib}
\usepackage{bbm}
\usepackage{array}
\usepackage{float}
\usepackage{authblk}% delete this for AOS template
\usepackage{algorithm}
\usepackage[noend]{algpseudocode}
\renewcommand{\algorithmicensure}{\textbf{Input:}}

\usepackage[flushleft]{threeparttable}
\usepackage{tablefootnote}

\usepackage{footmisc}

\newcommand{\tX}{\widetilde{X}}%\mathbf{X}^{(1)}}
\newcommand{\tbX}{\widetilde{\mathbf{X}}}
\newcommand{\tbx}{\widetilde{\mathbf{x}}}
\newcommand{\bigY}{\mathbf{Y}}
\newcommand{\smally}{\mathbf{y}}

\newcommand{\swap}[1]{{\text{swap}(#1)}}

\newcommand{\LCD}{\mathrm{LCD}}
\newcommand{\LSM}{\mathrm{LSM}}
\newcommand{\mlr}{\mathrm{mlr}}
\newcommand{\MLR}{\mathrm{MLR}}
\newcommand{\amlr}{\mathrm{amlr}}
\newcommand{\AMLR}{\mathrm{AMLR}}
\newcommand{\uppi}{^{\pi}}
\newcommand{\littley}{\mathbf{y}}
\newcommand{\tbeta}{\tilde{\beta}}

\newcommand{\Power}{\Gamma}
\newcommand{\oracle}{^{\mathrm{oracle}}}
\newcommand{\bayes}{^{\pi}}
\newcommand{\uptheta}{^{(\theta)}}

\newcommand{\sorted}{\mathrm{sorted}}

\newcommand{\upj}{^{(j)}}
\newcommand{\negupj}{^{(\mathrm{-}j)}}
\newcommand{\upnegj}{\negupj}

\newcommand{\br}{\mathbf{r}}
\newcommand{\sample}{\mathrm{sample}}

\newcommand{\obvious}{_{\mathrm{obvious}}}
\newcommand{\nonobvious}{_{\mathrm{non-obvious}}}

\newcommand{\numnonnull}{s_n}

\usepackage[dvipsnames]{xcolor}

\numberwithin{equation}{section}

\newcommand{\Powervtwo}{\texttt{ENDisc}}
\newcommand{\TP}{\mathrm{TP}}

\title{Asymptotically Optimal Knockoff Statistics via the Masked Likelihood Ratio}

\author[1]{Asher Spector}
\author[2]{William Fithian}
\affil[1]{Department of Statistics, Stanford University, USA}
\affil[2]{Department of Statistics, UC Berkeley, USA}

\begin{document}
\maketitle

\begin{abstract}
In feature selection problems, knockoffs are synthetic controls for the original features. Employing knockoffs allows analysts to use nearly any variable importance measure or ``feature statistic" to select features while rigorously controlling false positives. However, it is not clear which statistic maximizes power. In this paper, we argue that state-of-the-art lasso-based feature statistics often prioritize features that are unlikely to be discovered, leading to low power in real applications. Instead, we introduce masked likelihood ratio (MLR) statistics, which prioritize features according to one's ability to distinguish each feature from its knockoff.  Although no single feature statistic is uniformly most powerful in all situations, we show that MLR statistics asymptotically maximize the number of discoveries under a user-specified Bayesian model of the data. (Like all feature statistics, MLR statistics always provide frequentist error control.) This result places no restrictions on the problem dimensions and makes no parametric assumptions; instead, we require a ``local dependence" condition that depends only on known quantities. In simulations and three real applications, MLR statistics outperform state-of-the-art feature statistics, including in settings where the Bayesian model is highly misspecified. We implement MLR statistics in the python package \texttt{knockpy}; our implementation is often faster than computing a cross-validated lasso.
\end{abstract}

\section{Introduction}\label{sec::intro}

Given a design matrix $\bX = (\bX_1, \dots, \bX_p) \in \R^{n \times p}$ and a response vector $\bigY \in \R^n$, the task of \textit{controlled feature selection} is, informally, to discover features that influence $\bigY$ while controlling the false discovery rate (FDR). In this context, knockoffs \citep{fxknock, mxknockoffs2018} are fake variables $\tbX \in \R^{n \times p}$ which act as negative controls for the features $\bX$. Remarkably, employing knockoffs allows analysts to use nearly any machine learning model or test statistic, often known interchangeably as a ``feature statistic" or ``knockoff statistic," to select features while exactly controlling the FDR. As a result, knockoffs has become popular in the analysis of genetic studies, financial data, clinical trials, and more \citep{genehunting2017, knockoffzoom2019, challet2021, sechidis2021predictiveknockoffs}.

The flexibility of knockoffs has inspired the development of a variety of feature statistics based on penalized regression coefficients, sparse Bayesian models, random forests, neural networks, and more (see, e.g., \cite{fxknock, mxknockoffs2018, knockoffsmass2018, deeppink2018}). These feature statistics not only reflect different modeling assumptions, but more fundamentally, they estimate different quantities, including coefficient sizes, Bayesian posterior inclusion probabilities, and various other measures of variable importance. Yet there has been relatively little theoretical comparison of these methods, in large part because analyzing the power of knockoffs can be very challenging; see Section \ref{subsec::literature}. In this work, we develop a principled approach and concrete methods for designing knockoff statistics that maximize power.

\subsection{Review of model-X and fixed-X knockoffs}\label{subsec::knockreview}

This section reviews the key elements of Model-X (MX) and Fixed-X (FX) knockoff methods. 

\textbf{Model-X (MX) knockoffs} \citep{mxknockoffs2018} is a method to test the hypotheses $H_j : \bX_j \Perp \bigY \mid \bX_{-j}$, where $\bX_{-j} \defeq \{\bX_{\ell}\}_{\ell \ne j}$ denotes all features except $\bX_j$, assuming that the law of $\bX$ is known.\footnote{Note that this assumption can be relaxed to having a well-specified parametric model for $\bX$ \citep{condknock2019}, and knockoffs are known to be robust to misspecification of the law of $\bX$ \citep{robustness2020}.} Applying MX knockoffs requires three steps.

\textit{1. Constructing knockoffs.} Valid MX knockoffs must satisfy two properties.  First, the columns of $\bX$ must be \textit{pairwise exchangeable} with the corresponding columns of $\tbX$, i.e. $[\bX_{-j}, \tbX_{-j}, \bX_j, \tbX_j] \disteq [\bX_{-j}, \tbX_{-j}, \tbX_j, \bX_j]$ must hold for all $j \in [p]$.
%I.e., if $[\bX, \tbX]_{\swap{j}}$ denotes the matrix $[\bX, \tbX]$ but with the columns $\bX_j$ and $\tbX_j$ swapped, then $[\bX, \tbX]_{\swap{j}} \disteq [\bX, \tbX]$ must hold for $j \in [p]$. 
Second, we require that $\tbX \Perp \bigY \mid \bX$, which holds if one constructs $\tbX$ without looking at $\bigY$. Informally, these constraints guarantee that $\bX_j, \tbX_j$ are ``indistinguishable" under $H_j$. Sampling knockoffs can be challenging, but this problem is well studied \citep[e.g.,][]{metro2019}.

\textit{2. Fitting feature statistics.} Next, use any machine learning (ML) algorithm to fit feature importances $Z = z([\bX, \tbX], \bigY) \in \R^{2p}$, where $Z_j$ and $Z_{j+p}$ heuristically measure the ``importance" of $\bX_j$ and $\tbX_j$ in predicting $\bigY$. The only restriction is that swapping $\bX_j$ and $\tbX_j$ must also swap the feature importances $Z_j$ and $Z_{j+p}$ without changing any of the other feature importances $\{Z_\ell\}_{\ell \ne j}, \{Z_{\ell+p}\}_{\ell \ne j}$. This restriction is satisfied by most ML algorithms, such as the lasso or various neural networks \citep{deeppink2018}.

Given $Z$, we define the \textit{feature statistics} $W = w([\bX, \tbX], \bigY) \in \R^p$ via $W_j = f(Z_j, Z_{j+p})$ where $f(x,y) = - f(y,x)$ is any antisymmetric function. E.g., the lasso coefficient difference (LCD) statistic sets $W_j = |Z_j| - |Z_{j+p}|$, where $Z_j$ and $Z_{j+p}$ are coefficients from a lasso fit on $[\bX, \tbX], \bigY$. Intuitively, when $W_j$ is positive, this suggests that $\bX_j$ is more important than $\tbX_j$ and thus is evidence against the null. Indeed, Steps 1-2 guarantee that the signs of the null $\{W_j\}_{j=1}^p$ are i.i.d. random signs.

\textit{3.Make rejections.} Define the data-dependent threshold $T \defeq \inf\left\{t > 0: \frac{\#\{j : W_j \le -t \} + 1}{\#\{W_j \ge t \}} \le q \right\}$, where $\inf \emptyset \defeq \infty$. 
Then, reject $S \defeq \{j : W_j \ge T\}$, which guarantees finite-sample FDR control at level $q \in (0,1)$. Note this result does not require any assumptions about the law of $\bigY \mid \bX$.

\begin{theorem}[\cite{mxknockoffs2018}]\label{thm::fdr} Let $P\opt$ denote the unknown joint law of 
$(\bX, \bigY)$, and suppose the law of $\bX \sim P_X\opt$ is known, allowing one to construct valid knockoffs $\tilde{\bX}$. \footnote{Typically, one assumes that the observations are i.i.d. to construct valid knockoffs, but the i.i.d. assumption is not necessary as long as $\tbX$ are valid knockoffs.} Then for any feature statistic $w$, 
\begin{equation*}
    \mathrm{FDR} \defeq \E_{P\opt}\left[\frac{|S \cap \mcH_0|}{1 \vee |S|}\right] \le q,
\end{equation*}
where $\mcH_0 = \{j \in [p] : H_j \text{ is true}\}$ is the set of nulls under $P\opt$.
\end{theorem}

\textbf{Fixed-X (FX) knockoffs} \citep{fxknock} treats $\bX$ as nonrandom and yields exact FDR control under the Gaussian linear model $\bigY \mid \bX \sim \mcN(\bX \beta, \sigma^2 I_n)$. Fitting FX knockoffs is identical to the steps above with two exceptions:
\begin{enumerate}[noitemsep, topsep=0pt]
    \item FX knockoffs need not satisfy the constraints in Step 1: instead, $\tbX$ must satisfy (i) $\tbX^T \tbX = \bX^T \bX$ and (ii) $\tbX^T \bX = \bX^T \bX - \Delta$, for some diagonal matrix $\Delta$ satisfying $2 \bX^T \bX \succ \Delta$.
    \item The feature importances $Z$ can only depend on $\bigY$ through $[\bX, \tbX]^T \bigY$, which permits the use of many test statistics, but not all (for example, this prohibits the use of cross-validation). 
\end{enumerate}
Our theory applies to both MX and FX knockoffs,but  oftenfocus on the MX context for brevity.

\subsection{Theoretical problem statement}\label{subsec::problemstatement}

This section defines two types of optimal knockoff statistics: {\em oracle statistics}, which maximize the expected number of discoveries (\Powervtwo) for the true (unknown) data distribution $P\opt$, and {\em Bayes-optimal statistics}, which maximize \Powervtwo\, with respect to a prior distribution over $P\opt$. We focus on the expected number of discoveries since it greatly simplifies the analysis and all feature statistics provably control the frequentist FDR anyway. However, Section \ref{subsec::amlr} extends our analysis to consider the expected number of true discoveries.

Let $S_w \subset [p]$ denote the discovery set using feature statistic $w$ on data $([\bX, \tbX], \bigY)$. 
An \textit{oracle statistic} maximizes the expected number of discoveries under $P\opt$ as defined below:
\begin{equation}\label{eq::oracledef}
    \Powervtwo\opt(w) \defeq
    \,\,\E_{P\opt}\left[|S_w|\right].
\end{equation}
Next, let $\mcP = \{P^{(\theta)} : \theta \in \Theta\}$ denote a model class of potential distributions for $(\bX, \bigY)$ and let  $\pi : \Theta \to \R_{\ge 0}$ denote a prior density over $\mcP$.\footnote{We implicitly assume all elements of $\mcP$ are consistent with the core assumptions of MX/FX knockoffs (see Section \ref{subsec::knockreview}). For example, when employing MX knockoffs, all $P\uptheta \in \mcP$ should specify the correct marginal law for $X$. Although this is not necessary for our theoretical results, it is necessary for the computational techniques in Section \ref{subsec::computation}.} A {\em Bayes-optimal statistic} maximizes the average-case expected number of discoveries with respect to $\pi$:
\begin{equation}\label{eq::problemstatement}
    \Powervtwo\uppi(w) \defeq \int_{\Theta} \E_{P^{(\theta)}}[|S_w|] \pi(\theta) d\theta \defeq \E_{P\bayes}[|S_w|],
\end{equation}
where above, $P\bayes$ denotes the mixture distribution which first samples a parameter $\theta\opt \in \Theta$ according to the prior $\pi$ and then samples $(\bX, \bigY) \mid \theta\opt \sim P^{(\theta\opt)}$. We refer to $P\bayes$ as a ``Bayesian model," and we give a default choice of $P\bayes$ (based on sparse generalized additive models) in Section \ref{subsec::gams}. Our paper introduces statistics $\mlr\oracle$ and $\mlr\uppi$ that asymptotically maximize $\Powervtwo\opt$ and $\Powervtwo\uppi$, respectively. 

While introducing a prior distribution may seem a strong assumption to some readers, Bayesian models are routinely used in applications where knockoffs are commonly applied, such as genetic fine-mapping \citep[e.g.,][]{guanstephens2011, polyfun2019}. Furthermore, in simulations and real applications, our approach yields significant power gains over pre-existing approaches even when the prior is highly misspecified. Lastly, we emphasize that Theorem \ref{thm::fdr} guarantees that using $\mlr\uppi$ as a feature statistic guarantees frequentist FDR control under the \textit{true} law $P\opt$ of the data, even if $P\opt$ is not a member of $\mcP$---and in this Type I error result, the conditional independence null hypotheses $H_1,\ldots,H_p$ are defined nonparametrically with respect to the unknown $P\opt$, not with respect to $P\bayes$.%%%COULDCUT last clause

\subsection{Contribution and overview of results}\label{subsec::contribution}

This paper develops \textit{masked likelihood ratio} (MLR) statistics, a class of feature statistics that are asymptotically optimal, computationally efficient, and powerful in applications. To derive these statistics, we reformulate MX knockoffs as a guessing game on \textit{masked data} $D = (\bigY, \{\bX_j, \tbX_j\}_{j=1}^p),$ where the notation $\{\bX_j, \tbX_j\}$ denotes an unordered set.\footnote{For brevity, this section only presents results for MX knockoffs. See Section \ref{sec::mlr} for analogous results for the FX case.} After observing $D$, the analyst must do as follows:
\begin{itemize}[itemsep=0.5pt, topsep=0pt, leftmargin=*]
    \item For each $j \in [p]$, the analyst must produce a guess $\widehat{\bX}_j \in \R^n$ of the value of the feature $\bX_j$ based on $D$. Note that given $D$, $\widehat{\bX}_j \in \{\bX_j, \tbX_j\}$ takes one of two values.
    \item The analyst must then assign an order to their $p$ guesses, ideally from most to least promising.
    \item The analyst may make $k$ discoveries if roughly $k$ of their first $(1+q)k$ guesses are correct (according to the order they specify), where $q$ is the FDR level. Here, guess $j$ is ``correct" if $\widehat{\bX}_j = \bX_j$.
\end{itemize}

We show that to maximize the expected number of discoveries, an asymptotically optimal strategy is:
\begin{itemize}[topsep=0pt, itemsep=0.5pt, leftmargin=*]
    \item For each $j \in [p]$, guess the value $\widehat{\bX}_j \in \{\bX_j, \tbX_j\}$ which is conditionally more likely given $D$ (see below).
    \item Order the guesses in descending order of the probability that each guess is correct.
\end{itemize}
In the traditional language of knockoffs, this corresponds to using the masked data to compute the log-likelihood ratio between the two possible values of $\bX_j$ (namely $\{\bX_j, \tbX_j\})$ given $D$. Precisely, let $P\bayes$ denote a Bayesian model as defined in Section \ref{subsec::problemstatement}. Then for any value $\bd = (\littley, \{\bx_j, \tbx_j\}_{j=1}^p)$ in the support of $D$ and any fixed $\bx \in \{\bx_j, \tbx_j\}$, let $P\bayes_{j}(\bx \mid \bd) = P\bayes(\bX_j = \bx \mid D = \bd)$ denote the conditional law of $\bX_j$ given $D$. The \textit{masked likelihood ratio} (MLR) statistic is defined as
\begin{equation}\label{eq::mlr}
    \MLR_j\uppi \defeq \mlr_j\uppi([\bX, \tbX], \bigY) \defeq \log\left(\frac{P\bayes_j(\bX_j \mid D)}{P\bayes_j(\tbX_j \mid D)}\right).
\end{equation}
In words, the numerator plugs the observed values of $\bX_j$ and $D$ into the conditional law of $\bX_j \mid D$ under $P\bayes$, and the denominator plugs $\tbX_j$ and $D$ into the same law. Since swapping $\bX_j$ and $\tbX_j$ flips the sign of $\MLR_j\uppi$ without changing the values of $\{\MLR_\ell\uppi\}_{\ell \ne j}$, this equation defines a valid knockoff statistic. See Section \ref{subsec::literature} for comparison to \cite{katsevichmx2020}'s (unmasked) likelihood ratio statistic.

This paper gives three arguments motivating the use of MLR statistics.

\textbf{1. Intuition: the right notion of variable importance.} Existing feature statistics measure many different proxies for variable importance, ranging from regression coefficients to posterior inclusion probabilities. However, Section \ref{subsec::hiv_motiv} shows that popular lasso-based methods incorrectly prioritize features $\bX_j$ that are predictive of $\bigY$ but are nearly indistinguishable from their knockoffs $\tbX_j$, leading to low power in real applications. In contrast to conventional variable importances, MLR statistics instead estimate whether a feature $\bX_j$ is distinguishable from its knockoff $\tbX_j$. %We show in Section \ref{sec::mlr} that this is the optimal way to prioritize the features.

\textbf{2. Theory: asymptotic Bayes-optimality.} Section \ref{subsec::avgopt} shows that MLR statistics asymptotically maximize the number of expected discoveries under $P\bayes$. 
Namely, under technical assumptions, Theorem \ref{thm::avgopt} shows that for any valid feature statistic $w$,
\begin{equation}
    \Powervtwo\uppi(\mlr\uppi) \ge \Powervtwo\uppi(w)
    + o(\text{\# of non-nulls}).
\end{equation}
Our result applies to arbitrarily high-dimensional asymptotic regimes and allows $P\bayes$ to take any form---we do not assume that $\bigY \mid \bX$ follows a linear model under $P\bayes$. Instead, we assume the signs of the MLR statistics satisfy a local dependency condition, similar to dependency conditions often assumed on p-values \citep{genovese2004, storey2004, ferreira2006, farcomeni2007}. Our condition does not involve unknown quantities, so it can be diagnosed in practice.

Despite the Bayesian nature of this optimality result, we emphasize that MLR statistics are valid knockoff statistics. Thus, if $S_{\mlr} \subset [p]$ are the discoveries made by the $\mlr\uppi$ feature statistic, Theorem \ref{thm::fdr} shows that the frequentist FDR is controlled in finite samples assuming only that $\tbX$ are valid knockoffs:
\begin{equation}
    \mathrm{FDR} \defeq \E_{P\opt}\left[\frac{|S_{\mlr} \cap \mcH_0|}{1 \vee |S_{\mlr}|}\right] \le q.
\end{equation}
%%%%%COULDCUT this equation and shorten the previous paragraph. I.e., Despite the Bayesian nature of this optimality result, we emphasize that MLR statistics are valid knockoff statistics. Thus, when employing them, the frequentist FDR is controlled in finite samples assuming only that $\tbX$ are valid knockoffs (see Theorem 1.1 and section 1.2 for discussion).
%This result holds even if $P\bayes$ and $P\opt$ are quite different.

\textbf{3. Empirical results.} We demonstrate via simulations and three real data analyses that MLR statistics are powerful in practice, even when the user-specified Bayesian model $P\bayes$ is highly misspecified.

\begin{itemize}[topsep=0pt, itemsep=0.5pt, leftmargin=*]
    \item We develop concrete instantiations of MLR statistics based on uninformative (sparse) priors for generalized additive models and binary GLMs. Our Python implementation is often faster than fitting a cross-validated lasso. 

    \item In simulations, MLR statistics outperform other state-of-the-art feature statistics, often by wide margins. Even when $P\bayes$ is highly misspecified, MLR statistics often nearly match the performance of the oracle which sets $P\bayes = P\opt$.
    Furthermore, when $\bigY$ has a highly nonlinear relationship with $\bX$, MLR statistics also outperform ``black-box" feature statistics based on neural networks and random forests.
    
    \item We replicate three knockoff-based analyses of drug resistance \citep{fxknock}, financial factor selection \citep{challet2021}, and RNA-seq data \citep{nodewiseknock}. We find that MLR statistics (with an uninformative prior) make one to ten times more discoveries than the original analyses.
\end{itemize}
Overall, our results suggest that MLR statistics can substantially increase the power of knockoffs. 
%%%COULDCUT this last line

\subsection{Related literature}\label{subsec::literature}

The literature contains many feature statistics, which can (roughly) be separated into three categories. First, perhaps the most common feature statistics are based on penalized regression coefficients, notably the lasso signed maximum (LSM) and lasso coefficient difference (LCD) statistics \citep{fxknock}. Indeed, these lasso-based statistics are often used in applied work \citep[e.g.,][]{knockoffzoom2019} and have received much theoretical attention \citep{weinstein2017, rank2017, ke2020, weinstein2020power, wang2021power}. Perhaps surprisingly, we argue that many of these statistics target the wrong notion of variable importance, leading to reduced power. Second, some works have introduced Bayesian knockoff statistics \citep[e.g.,][]{mxknockoffs2018, RenAdaptiveKnockoffs2020}. MLR statistics have a Bayesian flavor but take a different form than previous statistics. Furthermore, our motivation differs from those of previous works: the real innovation of MLR statistics is to estimate a \textit{masked} likelihood ratio, and we mainly use a Bayesian framework to quantify uncertainty about nuisance parameters (see Section \ref{subsec::mlr}). In contrast, previous works largely motivated Bayesian statistics as a way to incorporate prior information \citep{mxknockoffs2018, RenAdaptiveKnockoffs2020}. That said, an important special case of MLR statistics is similar to the ``BVS" statistics from \cite{mxknockoffs2018}, as discussed in Section \ref{subsec::computation}. Third, many feature statistics take advantage of ``black-box" ML to assign variable importances \cite[e.g.,][]{deeppink2018, knockoffsmass2018}. Empirically, our implementation of MLR statistics based on regression splines outperforms ``black-box" feature statistics in Section \ref{sec::sims}.

Previous power analyses for knockoffs have largely focused on showing the consistency of coefficient-difference feature statistics \citep{liurigollet2019, rank2017, mrcknock2022} or quantifying the power of coefficient-difference feature statistics assuming $\bX$ has i.i.d. Gaussian entries \citep{weinstein2017, weinstein2020power, wang2021power}. \cite{ke2020} also derive a phase diagram for LCD statistics assuming $\bX$ is blockwise orthogonal. Our goal is different: to show that MLR statistics are asymptotically optimal, with particular focus on settings where the asymptotic power lies strictly between $0$ and $1$. Furthermore, the works above exclusively focus on  Gaussian linear models, whereas our analysis places no explicit restrictions on the law of $\bigY \mid \bX$ or the dimensionality of the problem. Instead, we assume the signs of the MLR statistics satisfy a local dependency condition, similar to common dependency conditions on p-values \citep{genovese2004, storey2004, ferreira2006, farcomeni2007}. However, our proof technique is novel and specific to knockoffs. %%%COULDCUT last sentence

Our theory builds on \cite{whiteout2021}, who developed knockoff$\star$, a provably optimal oracle statistic for FX knockoffs---in fact, oracle MLR statistics are equivalent to knockoff$\star$ for FX knockoffs. Our work also builds on \cite{katsevichmx2020}, who showed that \textit{unmasked} likelihood statistics maximize $P\opt(W_j > 0)$. MLR statistics also have this property, although we show the stronger result that MLR statistics maximize the expected number of overall discoveries. Another key difference is that unmasked likelihood statistics are not jointly valid knockoff statistics (see Appendix \ref{appendix::unmasked_lr}). Thus, unmasked likelihood statistics do not yield provable FDR control, whereas MLR statistics do. Lastly, we note that the oracle procedures derived in these two works cannot be used in practice since they depend on unknown parameters. To our knowledge, MLR statistics are the first usable knockoff statistics with explicit optimality guarantees.

\subsection{Notation and outline}

\textbf{Notation}: Let $\bX \in \R^{n \times p}$ and $\bigY \in \R^{n}$ denote the design matrix and response vector in a feature selection problem with $n$ data points and $p$ features. We let the non-bold versions $X = (X_1, \dots, X_p) \in \R^p$ and $Y \in \R$ denote the features and response for any arbitrary observation. %%%COULDCUT do we use this?
For $k \in \N$, define $[k] \defeq \{1, \dots, k\}$. For any $M \in \R^{m \times k}$ and $J \subset [k]$, $M_J$ denotes the columns of $M$ corresponding to the indices in $J$. Similarly, $M_{-J}$ denotes the columns of $M$ which do not appear in $J$, and $M_{-j}$ denotes all columns except column $j \in [k]$. For matrices $M_1 \in \R^{n \times k_1}, M_2 \in \R^{n \times k_2}$, $[M_1, M_2] \in \R^{n \times (k_1 + k_2)}$ denotes the column-wise concatenation of $M_1, M_2$.
$I_n$ denotes the $n \times n$ identity.  Throughout, $P\opt$ denotes the true (unknown) law of $\bX, \bigY$, and $P\bayes$ denotes a user-specified Bayesian model of the law of $\bX, \bigY, \theta\opt$ as defined in Section \ref{subsec::problemstatement}.

\textbf{Outline}: Section \ref{subsec::hiv_motiv} gives intuition explaining why popular feature statistics may have low power, using an HIV resistance dataset as motivation. Section \ref{sec::mlr} introduces MLR statistics and presents our theoretical results. Section \ref{subsec::computation} discusses computation and suggests default choices of the Bayesian model $P\bayes$. Section \ref{sec::sims} compares MLR statistics to competitors via simulations. Section \ref{sec::data} applies MLR statistics to three real datasets. Section \ref{sec::discussion} discusses future directions.
%%%COULDCUT this whole section honestly

\section{Intuition and motivation from an HIV drug resistance dataset}\label{subsec::hiv_motiv}

\subsection{Intuition: what makes knockoffs powerful?}\label{subsec::intuition}

Given a vector of knockoff statistics $W \in \R^p$, the number of discoveries is determined as follows:
\begin{itemize}[itemsep=0.5pt, topsep=0pt, leftmargin=*]
    \item Step 1: Let $\sigma : [p] \to [p]$ denote the permutation such that $|W_{\sigma(1)}| \ge |W_{\sigma(2)}| \ge \dots \ge |W_{\sigma(p)}|$.
    \item Step 2: Let $k$ be the largest integer such that at least $\ceil{(k+1)/(1+q)}$ of the $k$ feature statistics $W_{\sigma(1)}, \dots, W_{\sigma(k)}$ with the largest absolute values have positive signs. Then the analyst may discover the features corresponding to the positive signs among $W_{\sigma(1)}, \dots, W_{\sigma(k)}$.
\end{itemize}
The procedure above is equivalent to using the ``data-dependent threshold" from Section \ref{subsec::knockreview} \citep[see][]{fxknock}. This characterization suggests that to make many discoveries, we should:
\begin{itemize}[itemsep=0.5pt, topsep=0pt, leftmargin=*]
    \item Goal 1: Maximize the probability that each coordinate $W_j$ has a positive sign. (Note that null coordinates are guaranteed to be symmetric.) %COULDCUT: part in parentheses
    \item Goal 2: Assign absolute values such that coordinates $W_j$ with larger absolute values also have higher probabilities of being positive. This ensures that for each $k$, the $k$ feature statistics with the highest absolute values contain as many positive signs as possible, thus maximizing the number of discoveries. Although it is not yet clear how to formalize this goal, intuitively, we would like $\{|W_j|\}_{j=1}^p$ to have the same order as $\{P\opt(W_j > 0)\}_{j=1}^p$. 
\end{itemize}

We emphasize that the second goal is crucial to make discoveries when $\{P\opt(W_j > 0)\}_{j=1}^p$ is heterogeneous, as illustrated in Section \ref{subsec::hiv}. See also Appendix \ref{appendix::wstat_synth} for a simpler simulated example.

\subsection{Motivation from the HIV drug resistance dataset}\label{subsec::hiv}

We now ask: do the most common choices of feature statistics used in the literature, LCD and LSM statistics, accomplish Goals 1-2? We argue \textit{no}, using the HIV drug resistance dataset from \cite{rhee2006data} as an illustrative example. This dataset has been used as a benchmark in several papers about knockoffs, e.g., \cite{fxknock, deepknockoffs2018}, and we perform a complete analysis of this dataset in Section \ref{sec::data}. For now, note that the design $\bX$ consists of genotype data from $n \approx 750$ HIV samples, the response $\bigY$ measures the resistance of each sample to a drug (in this case Indinavir), and we apply knockoffs to discover which genetic variants affect drug resistance---note our analysis exactly mimics that of \cite{fxknock}. As notation, let $(\hat \beta^{(\lambda)}, \tbeta^{(\lambda)}) \in \R^{2p}$ denote the estimated lasso coefficients fit on $[\bX, \tbX]$ and $\bigY$ with regularization parameter $\lambda$. Furthermore, let $\hat \lambda_j$ (resp. $\tilde{\lambda}_j)$ denote the smallest value of $\lambda$ such that $\hat \beta^{(\lambda)}_j \ne 0$ (resp. $\tbeta^{(\lambda)}_j \ne 0$).  Then the LCD and LSM statistics are defined as:
\begin{equation}\label{eq::lassodef}
    W_j^{\LCD} = |\hat\beta_j^{(\lambda)}| - |\tbeta_j^{(\lambda)}|, \,\,\,\,\,\,\,\,\,\, W_j^{\LSM} = \sign(\hat \lambda_j - \tilde{\lambda}_j) \max(\hat \lambda_j, \tilde{\lambda}_j). 
\end{equation}

\begin{figure}[!ht]
    \centering
    \includegraphics[width=\linewidth]{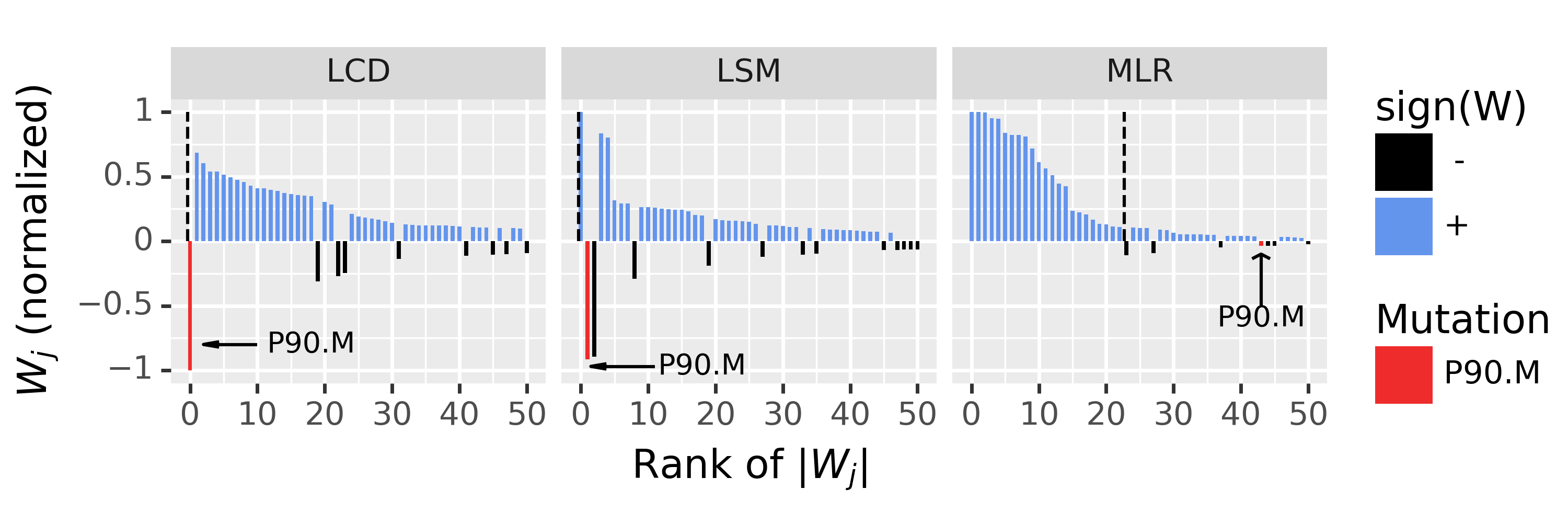}
    \caption{We plot the first $50$ LCD, LSM, and MLR feature statistics sorted in descending order of absolute value when applied to the HIV drug resistance dataset for the drug Indinavir (IDV). For FDR level $q=0.05$, all positive feature statistics to the left of the dotted black line are discoveries. This figure shows that when $\bX$ is correlated, LCD and LSM statistics make few discoveries because they occasionally yield highly negative $W$-statistics for highly predictive variables that have low-quality knockoffs, such as the ``P90.M" variant from Section \ref{subsec::hiv_motiv}. In contrast, MLR statistics (defined in Section \ref{sec::mlr}) deprioritize the P90.M variant; although they still do not discover P.90M, this deprioritization allows the discovery of $22$ other features. For visualization, we apply a monotone transformation to $\{W_j\}$ such that $|W_j| \le 1$, which (provably) does not change the performance of knockoffs. See Appendix \ref{appendix::realdata} for further details and corresponding plots for the other fifteen drugs in the dataset.}
    \label{fig::wstatplot}
\end{figure}

As intuition, imagine that a feature $\bX_j$ appears to influence $\bigY$: however, due to high correlations within $\bX$, we must create a knockoff $\tbX_j$ which is highly correlated with $\bX_j$. For example, the ``P90.M" variant in the HIV dataset is extremely predictive of resistance to Indinavir (IDV), as its OLS t-statistic is $\approx 8.95$. However, in the original analysis, P90.M is $> 99\%$ correlated with its knockoff, so the lasso may select $\tbX_j$ instead of $\bX_j$. Furthermore, since the lasso induces sparsity, it is unlikely to select \textit{both} $\bX_j$ and $\tbX_j$ as they are highly correlated. Thus, $W_j^\LCD$ and $W_j^\LSM$ will have large absolute values, since $\bX_j$ appears significant, and a reasonably high probability of being negative, since $\bX_j \approx \tbX_j$. Indeed, the LCD and LSM statistics for P90.M have, respectively, the largest and second-largest absolute values among all genetic variants, but both statistics are negative because the lasso selected the knockoff instead of the feature. 

Figure \ref{fig::wstatplot} shows that this misprioritization prevents the LCD and LSM statistics from making any discoveries when $q = 0.05$. Yet this problem is avoidable. If $\cor(X_j,\tX_j)$ is large and $W_j$ may be negative, we can ``deprioritize" $W_j$ by lowering its absolute value. As shown by Figure \ref{fig::wstatplot}, this is exactly what MLR statistics do for the P90.M variant. Although this does not allow us to discover P90.M, it does allow us to discover $22$ other features. 

\begin{remark}[Alternative solutions] 
To our knowledge, this problem with lasso statistics has not been previously discussed (see Section \ref{subsec::literature}). Once pointed out, there are many intuitive approaches that mitigate (but do not fully solve) this problem, such as studentizing the coefficients or adding a ridge penalty. These practical ideas may merit further exploration; however, we focus on obtaining optimal feature statistics. Furthermore, some may argue that the best solution is simply to ensure that P90.M is less correlated with its knockoff. We wholeheartedly agree that the SDP knockoff construction from \cite{fxknock} is sub-optimal here, and thus, we also use an alternative knockoff construction in Section \ref{sec::data}. Yet reasonably strong correlations between features and knockoffs are inevitable when the features are correlated \citep{daibarber2016}. However, knockoffs can still be powerful in these settings if the feature statistics properly account for the (known) dependencies among $[\bX, \tbX]$. MLR statistics are designed to do this.
\end{remark}

\section{Masked likelihood ratio statistics}\label{sec::mlr}

This section introduces and analyzes MLR statistics. First, Section \ref{subsec::maskeddata} introduces the notation needed to define MLR statistics. Then, Section \ref{subsec::mlr} defines MLR statistics, and Section \ref{subsec::avgopt} shows that MLR statistics asymptotically maximize the expected number of discoveries. Finally, Section \ref{subsec::amlr} introduces an adjusted MLR statistic that asymptotically maximizes the expected number of \textit{true} discoveries.

\subsection{Knockoffs as inference on masked data}\label{subsec::maskeddata}

Section \ref{subsec::hiv_motiv} argued that to maximize power, we should assign $W_j$ a large absolute value if and only if $P\opt(W_j > 0)$ is large. To do this, we must estimate $P\opt(W_j > 0)$ from the data, but we cannot use all the data for this purpose: e.g., we cannot directly adjust $|W|$ based on $\sign(W)$ without violating FDR control. To resolve this ambiguity, we reformulate knockoffs as inference on \textit{masked data}.

\begin{definition}\label{def::maskeddata} Suppose we observe data $\bX, \bigY$, knockoffs $\tbX$, and independent random noise $U$. ($U$ may be used to fit a randomized feature statistic.) The \emph{masked data} $D$ is defined as
\begin{equation}\label{eq::maskeddata}
    D = \begin{cases}
        (\bigY, \{\bX_j, \tbX_j\}_{j=1}^p, U) & \text{ for model-X knockoffs} \\
        (\bX, \tbX, \{\bX_j^T \bigY, \tbX_j^T \bigY\}_{j=1}^p, U) & \text{ for fixed-X knockoffs.}
    \end{cases}
\end{equation}
\end{definition}

As shown in Propositions \ref{prop::mxdistinguish}-\ref{prop::fxdistinguish}, the masked data $D$ is all the data we may use when assigning magnitudes to $W$, and knockoffs will be powerful when we can recover the full data from  $D$. 

\begin{proposition}\label{prop::mxdistinguish} Let $\tbX$ be model-X knockoffs such that $\bX_j \ne \tbX_j$ a.s. for $j \in [p]$. Then $W = w([\bX, \tbX], \bigY)$ is a valid feature statistic if and only if:
\begin{enumerate}[topsep=0pt, leftmargin=*]
    \setlength{\parskip}{0pt}
    \setlength{\itemsep}{0pt plus 1pt}
    \item $|W|$ is a function of the masked data $D$.
    \item For all $j \in [p]$, there exists a $D$-measurable random vector $\widehat{\bX}_j$ such that $W_j > 0$ if and only if $\widehat{\bX}_j = \bX_j$.
\end{enumerate}
\end{proposition}
%%COULDCUT make this in-line and not bullet points.

Proposition \ref{prop::mxdistinguish} reformulates knockoffs as a guessing game, where we produce a ``guess" $\widehat{\bX}_j \in \{\bX_j, \tbX_j\}$ of the value of $\bX_j$ based on $D = (\bigY, \{\bX_j, \tbX_j\}_{j=1}^p)$. If our guess is right, meaning $\widehat{\bX}_j = \bX_j$, then we are rewarded and $W_j > 0$: else $W_j < 0$. To avoid highly negative $W$-statistics, we should only assign $W_j$ a large absolute value when we are confident that our ``guess" $\widehat{\bX}_j$ is correct. We discuss more implications of this result in the next section: for now, we obtain an analogous result for fixed-X knockoffs (similar to a result from \cite{whiteout2021}) by substituting $\{\bX_j^T \bigY, \tbX_j^T \bigY\}$ for $\{\bX_j, \tbX_j\}$.

\begin{proposition}\label{prop::fxdistinguish} Let $\tbX$ be fixed-X knockoffs satisfying $\bX_j^T \bigY \ne \tbX_j^T \bigY$ a.s. for $j \in [p]$. Then $W = w([\bX, \tbX], \bigY)$ is a valid feature statistic if and only if:
\begin{enumerate}[topsep=0pt, leftmargin=*]
    \setlength{\parskip}{0pt}
    \setlength{\itemsep}{0pt plus 1pt}
    \item $|W|$ is a function of the masked data $D$.
    \item For $j \in [p]$, there exists a $D$-measurable random variable $R_j$ such that $W_j > 0$ if and only if $R_j = \bX_j^T \bigY$.
\end{enumerate}
\end{proposition}

\begin{remark} Propositions \ref{prop::mxdistinguish}-\ref{prop::fxdistinguish} hold for knockoffs as defined in \cite{fxknock, mxknockoffs2018}. However, in the fixed-X case, one can also augment $D$ to include $\hat \sigma^2 = \|(I_n - H) \bigY\|_2^2$, where $H$ is the OLS projection matrix of $[\bX, \tbX]$ while preserving validity \citep{altsign2017, whiteout2021}. Our theory also applies to this extension of the knockoffs framework.
\end{remark}
%%COULDCUT this remark

\subsection{Introducing masked likelihood ratio (MLR) statistics}\label{subsec::mlr}

We now introduce masked likelihood ratio (MLR) statistics in two steps. First, we introduce oracle MLR statistics, which depend on the unknown law $P\opt$ of the data. Then, we introduce Bayesian MLR statistics, which substitute $P\bayes$ for $P\opt$. Throughout, we focus on MX knockoffs, but analogous results for FX knockoffs merely replace $\{\bX_j, \tbX_j\}$ with $\{\bX_j^T \bigY, \tbX_j^T \bigY\}$ (see Definition \ref{def::mlr}).

\underline{Step 1: Oracle MLR statistics}. We now apply Proposition \ref{prop::mxdistinguish} to achieve the two intuitive optimality criteria from Section \ref{subsec::hiv_motiv}.

\begin{itemize}[itemsep=0.5pt, topsep=0pt, leftmargin=*]
    \item Goal 1 asks that we maximize $P\opt(W_j > 0)$. Proposition \ref{prop::mxdistinguish} shows that ensuring $W_j > 0$ is equivalent to correctly guessing the value of $\bX_j$ from the masked data  $D = (\bigY, \{\bX_j, \tbX_j\}_{j=1}^p)$. Thus, to maximize $P\opt(W_j > 0 \mid D) = P\opt(\widehat{\bX}_j = \bX_j \mid D)$, the analyst should guess the value $\widehat{\bX}_j = \argmax_{\bx \in \{\bX_j, \tbX_j\}} P_j\opt(\bx \mid D)$ which  maximizes the likelihood that the guess is correct.
    \item Goal 2 asks us to order the guesses in descending order of $P\opt(W_j > 0)$, i.e., in descending order of the likelihood that each guess $\widehat{\bX}_j$ is correct. %%COULDCUT second clause
\end{itemize}

Both goals are achieved by using the masked data to compute a log-likelihood ratio between the two possible values of $\bX_j$ (namely $\{\bX_j, \tbX_j\}$). This defines the oracle masked likelihood ratio: 
\begin{equation}
    \MLR_j\oracle = \log\left(\frac{P_{j}\opt(\bX_j \mid D)}{P_{j}\opt(\tbX_j \mid D)}\right), %\tag{\ref{eq::oraclemlr}}
\end{equation}
where $P\opt_{j} \defeq P\opt_{\bX_j \mid D}$ is the true (unknown) conditional law of $\bX_j$ given $D$. Soon, Proposition \ref{prop::intuit_opt} will verify that $\MLR\oracle$ achieves both goals above, and Section \ref{subsec::avgopt} shows that $\MLR_j\oracle$ asymptotically maximizes the expected number of discoveries under $P\opt$ (under regularity conditions).

\underline{Step 2: Bayesian MLR statistics}. $\MLR^{\mathrm{oracle}}$ cannot be used in practice since it depends on $P\opt$. A heuristic solution is to ``plug in" an estimate $\hat P$ for $P\opt$. For example, given some model class $\mcP = \{P\uptheta : \theta \in \Theta\}$ of the law of $(\bigY, \bX)$, one could estimate $\hat \theta$ using $D$ and replace $P\opt$ with $P^{(\hat\theta)}$. However, this ``plug-in" approach can perform poorly, since knockoffs are most popular in high-dimensional settings with significant uncertainty about the true value of any unknown parameters.
Thus, to account for uncertainty, we suggest replacing $P\opt$ with a Bayesian model $P\bayes$.

\begin{definition}[MLR statistics]\label{def::mlr} For any Bayesian model $P\bayes$ (see Section \ref{subsec::problemstatement}), we define the model-X masked likelihood ratio (MLR) statistic below:

\begin{equation}\label{eq::mlr_def}
    \MLR_j\uppi \defeq \mlr_j\uppi([\bX, \tbX], \bigY) \defeq \log\left(\frac{P\bayes_{j}(\bX_j \mid D)}{P\bayes_{j}(\tbX_j \mid D)}\right)
        \text{ for model-X knockoffs,}\footnote{In the edge case where the MLR is zero, we set $\MLR_j\uppi \simind \Unif(\{-\epsilon, \epsilon\})$ where $\epsilon$ satisfies $\epsilon < |\MLR_k\uppi|$ for each $k$ such that the $k$th feature statistic is nonzero, i.e., $|\MLR_k\uppi| > 0$.
        }
    %\tag{\ref{eq::mlr}}
\end{equation}
where $P\bayes_{j} \defeq P\bayes_{\bX_j \mid D}$ denotes the conditional law of $\bX_j \mid D$ under $P\bayes$. 

The fixed-X MLR statistic is analogous but replaces $\{\bX_j, \tbX_j\}$ with $\{\bX_j^T \bigY, \tbX_j^T \bigY\}$. In particular, if $P\bayes_{j,\mathrm{fx}}$ denotes the conditional law of $\bX_j^T \bigY \mid D$ under $P\bayes$, then
\begin{equation}\label{eq::fx_mlr}
    \MLR_j\uppi \defeq \mlr_j\uppi([\bX, \tbX], \bigY) \defeq  \log\left(\frac{P\bayes_{j,\mathrm{fx}} (\bX_j^T \bigY \mid D)}{P\bayes_{j,\mathrm{fx}}(\tbX_j^T \bigY \mid D)}\right) \text{ for fixed-X knockoffs.}
\end{equation}

\end{definition}
\begin{lemma}\label{lem::mlrvalid} Equations (\ref{eq::mlr_def}) and (\ref{eq::fx_mlr}) define valid MX and FX knockoff statistics, respectively.
\end{lemma}

To see how MLR statistics account for uncertainty about nuisance parameters, let $\pi(\theta \mid D)$ denote the posterior density of $\theta\opt \mid D$ under $P\bayes$. We can write, e.g., in the model-X case:

\begin{equation}
    \MLR_j\uppi = \log\left(\frac{\int_{\Theta} P\uptheta_j(\bX_j \mid D) \pi(\theta \mid D) d\theta}{\int_{\Theta} P\uptheta_j(\tbX_j \mid D) \pi(\theta \mid D) d\theta}\right).
\end{equation}
Unlike the ``plug-in" approach, $\MLR_j\uppi$ does not rely on a single estimate of $\theta$---instead, it takes the weighted average of the likelihoods $P\uptheta_j(\bX_j \mid D)$, weighted by the posterior law of $\theta \mid D$ under $P\bayes$.

We now verify that under $P\bayes$, MLR statistics achieve the intuitive criteria from Section \ref{subsec::hiv_motiv} (Goals 1-2). This result applies to oracle MLR statistics, since $\MLR_j\oracle = \MLR_j\uppi$ in the special case where $P\bayes = P\opt$.

\begin{proposition}\label{prop::intuit_opt} Let $\MLR\uppi$ be the MLR statistics with respect to a Bayesian model $P\bayes$. Then for any other feature statistic $W$,
\begin{equation}\label{eq::bestsigns}
    P\bayes(\MLR_j\uppi > 0 \mid D) \ge P\bayes(W_j > 0 \mid D).
\end{equation}
Furthermore, $\{|\MLR_j\uppi|\}_{j=1}^p$ has the same order as $\{P\bayes(\MLR_j\uppi > 0 \mid D)\}_{j=1}^p$. More precisely,
\begin{equation}\label{eq::bestmags}
    P\bayes(\MLR_j\uppi > 0 \mid D) = \frac{\exp(|\MLR_j\uppi|)}{1 + \exp(|\MLR_j\uppi|)}.
\end{equation}
\end{proposition}

Equation (\ref{eq::bestmags}) shows that the absolute values $|\MLR_j\uppi|$ have the same order as $P\bayes(\MLR_j\uppi > 0 \mid D)$, so under $P\bayes$, MLR statistics prioritize the hypotheses ``correctly." More generally, if $\bX_j$ is predictive of $\bigY$ but $\tbX_j$ is nearly indistinguishable from $\bX_j$, $|\MLR_j\uppi|$ should be small, since $\bX_j \approx \tbX_j$ suggests $P\bayes_j(\bX_j \mid D) \approx P\bayes_j(\tbX_j \mid D)$. Thus, MLR statistics should rarely be highly negative (see Figure \ref{fig::wstatplot}).

Lastly, we make two connections to the literature. First, one proposal in \cite{RenAdaptiveKnockoffs2020} also suggests ranking the hypotheses by $\P(W_j > 0 \mid |W_j|)$ (see their footnote 8). That said, \cite{RenAdaptiveKnockoffs2020} do not propose a feature statistic accomplishing this. Rather, they develop ``adaptive knockoffs," an extension of knockoffs that can be combined with any predefined feature statistic, including MLR or lasso statistics. Indeed, using better initial feature statistics should increase the power of adaptive knockoffs, so our contribution is both orthogonal and complementary to theirs (see Appendix \ref{appendix::adaknock} for more details). Second, \cite{katsevichmx2020} show that the \textit{unmasked} likelihood statistic maximizes $P\opt(W_j > 0)$; indeed, our work builds on theirs. However, there are two key differences. First, unlike MLR statistics, the unmasked likelihood statistic is not a valid knockoff statistic even though it is marginally symmetric under the null (see Appendix \ref{appendix::unmasked_lr}), so it does not provably control the FDR. Second, MLR statistics have additional guarantees on their \textit{magnitudes} (Eq. \ref{eq::bestmags}), allowing us to show much stronger theoretical results in Section \ref{subsec::avgopt}.

\begin{remark} Appendix \ref{appendix::groupmlr} extends this section's results to apply to \textit{group} knockoffs \citep{daibarber2016}. 
\end{remark}

\subsection{MLR statistics are asymptotically optimal}\label{subsec::avgopt}

We now show that MLR statistics asymptotically maximize $\Powervtwo\uppi$, the expected number of discoveries under $P\bayes$. Indeed, Proposition \ref{prop::intuit_opt} might make one hope that MLR statistics exactly maximize $\Powervtwo\uppi$, since MLR statistics exactly accomplish Goals 1-2 from Section \ref{subsec::hiv_motiv}. This intuition is correct under the conditional independence condition below (generalizing \cite{whiteout2021} Proposition 2).

\begin{proposition}\label{prop::exactopt} If $\{\I(\MLR_j\uppi > 0)\}_{j=1}^p$ are conditionally independent given $D$ under $P\bayes$, then \newline $\Powervtwo\uppi(\mlr\uppi) \ge \Powervtwo\uppi(w)$ for any valid feature statistic $w$.
%the MLR statistics with respect to $P\bayes$ exactly maximize the expected number of discoveries under $P\bayes$. 
\end{proposition}

Furthermore, in Gaussian linear models, oracle MLR statistics satisfy this conditional independence condition, making them finite-sample optimal.

\begin{proposition}\label{prop::gaussian_cond_ind}
Suppose that (i) $\bX$ are FX knockoffs or Gaussian conditional MX knockoffs \citep{condknock2019} and (ii) under $P\opt$, $\bigY \mid \bX \sim \mcN(\bX \beta, \sigma^2 I_n)$. Then under $P\opt$, $\{\I(\MLR_j\oracle > 0)\}_{j=1}^p \mid D$ are conditionally independent.
\end{proposition}

Absent this independence condition, it may be possible to exploit dependencies among the coordinates of $\sign(\MLR\uppi)$ to slightly improve power. Yet Appendix \ref{subsec::finitesampleopt} shows that to improve power even slightly seems to require pathological dependencies, making it hard to imagine that accounting for dependencies can substantially increase power in practice. Formally, we now show that MLR statistics are \textit{asymptotically} optimal under regularity conditions on the dependence of $\sign(\MLR\uppi) \mid D$.

To this end, consider any asymptotic regime where we observe $\bX^{(n)} \in \R^{n \times p_n}, \bigY^{(n)} \in \R^{n}$ and construct knockoffs $\tbX^{(n)}$. For each $n$, let $P\bayes_n$ denote a Bayesian model based on a model class $\mcP\upn = \{P\uptheta : \theta \in \Theta\upn\}$ and prior density $\pi\upn : \Theta\upn \to \R_{\ge 0}$. Let $D^{(n)}$ denote the masked data (Definition \ref{def::maskeddata}). For a sequence of feature statistics $W\upn = w_n([\bX\upn, \tbX\upn], \bigY\upn)$, let $S\upn(q)$ denote the rejection set of $W\upn$ when controlling the FDR at level $q$. So far, we have made no assumptions about the law of $\bigY\upn, \bX\upn$ under $P\bayes_n$, and we allow the dimension $p_n$ to grow arbitrarily with $n$. To analyze the asymptotic behavior of MLR statistics under $P\bayes_n$, we need two main assumptions.

\begin{assumption}[Sparsity]\label{assump::sparsity} 
%Let $\numnonnull = \E_{P\bayes_n}\left[\left|\{j\in [p_n] : X_j \not \Perp Y \mid X_{-j}, \theta\opt\}\right|\right]$ denote the expected number of non-nulls under $P_n\bayes$. 
For $\theta \in \Theta\upn$, let $\numnonnull\uptheta$ denote the number of non-nulls under $P\uptheta$ and $\numnonnull = \int_{\Theta} \numnonnull\uptheta \pi(\theta) d\theta$ denote the expected number of non-nulls under $P_n\bayes$.
We assume $\numnonnull \gg \log(p_n)^5$ as $n \to \infty$.
\end{assumption}

Assumption \ref{assump::sparsity} allows for many previously studied sparsity regimes, such as polynomial \citep{donohojin2004, ke2020} and linear \citep[e.g.,][]{weinstein2017} sparsity regimes.

\begin{assumption}[Local dependence]\label{assump::localdep} Under $P\bayes_n$, the conditional covariance matrix of $\sign(\MLR\uppi)$ given $D\upn$ decays exponentially off its diagonal. Formally, there exist constants $C\ge 0, \rho \in (0,1)$ such that
\begin{equation}\label{eq::expcovdecay}
    |\cov_{P\bayes_n}(\I(\MLR\uppi_i > 0), \I(\MLR\uppi_j > 0) \mid D\upn)| \le C \rho^{|i-j|}.
\end{equation}
\end{assumption}

%  We give a few justifications for this assumption:
% \begin{enumerate}[noitemsep, topsep=0pt, leftmargin=*]
%     \item This assumption can be diagnosed using the data. The covariances in Eq. (\ref{eq::expcovdecay}) are known quantities, since they depend only on $P\bayes$, which is specified by the analyst. In Section \ref{subsec::computation}, we develop an algorithm to compute $\cov_{P\bayes}(\sign(\MLR\uppi) \mid D)$, allowing analysts to assess the plausibility of Assumption \ref{assump::localdep}. 
%     \item In each of our simulations and three real data analyses, we find that $\cov_{P\bayes}(\sign(\MLR\uppi) \mid D)$ is nearly indistinguishable from a diagonal matrix, suggesting that Assumption \ref{assump::localdep} holds in practice (at least for the choices of $P\bayes$ we recommend). We give further intuition for this empirical finding in Section \ref{sec::sims}, but in brief, this is not too surprising because knockoffs are designed to ensure that $\sign(W)$ are conditionally independent for null variables, and this roughly extends to non-null variables as well.
% \end{enumerate}

Assumption \ref{assump::localdep} quantifies the requirement that $\sign(\MLR\uppi) \mid D\upn$ are not ``too" conditionally dependent. Similar local dependence conditions are common in the multiple testing literature \citep{genovese2004, storey2004, ferreira2006, farcomeni2007}, although previous assumptions are typically made about p-values. We justify this assumption below.
\begin{enumerate}[noitemsep, topsep=0pt, leftmargin=*]
    \item This assumption is intuitively plausible because knockoffs guarantee that the null coordinates of $\sign(W)$ are independent given $D$ under $P\opt$, regardless of the correlations among $\bX$ \citep{fxknock}. This independence also holds for non-null coordinates in Gaussian linear models (see Prop. \ref{prop::gaussian_cond_ind}). Appendix \ref{appendix::dependencediscussion} gives additional informal intuition explaining why this result often holds approximately under $P_n\bayes$ for both null and non-null variables.
    \item  Empirically, $\cov_{P\bayes}(\sign(\MLR\uppi) \mid D)$ is nearly indistinguishable from a diagonal matrix in all of our simulations and three real analyses. This suggests that Assumption \ref{assump::localdep} holds in practice.
    \item This assumption can also be diagnosed in real applications, since it depends only on $P\bayes$, which is specified by the analyst. To this end, Section \ref{subsec::computation} shows how to compute $\cov_{P\bayes}(\sign(\MLR\uppi) \mid D)$.% and diagnose Assumption \ref{assump::localdep} in real applications.
    \item Explicit analysis of the covariances in Eq. (\ref{eq::expcovdecay}) is known to be challenging. Nonetheless, in Appendix \ref{subsec::unhappy}, we prove that Assumption \ref{assump::localdep} holds if the design matrix $\bX$ is blockwise orthogonal, which is an important (if not entirely realistic) special case studied by \cite{ke2020}.
    \item Assumption \ref{assump::localdep} can also be substantially relaxed (see Appendix \ref{subsec::assumptions}). All we require is that the partial sums of $\{\sign(\MLR_j\uppi)\}_{j=1}^p$ obey a strong law of large numbers conditional on $D$.
\end{enumerate}

With these two assumptions, we show that MLR statistics asymptotically maximize $\Power_q(w_n)$, the expected number of discoveries normalized by the expected number of non-nulls:
\begin{equation}\label{eq::powerdef}
    \Power_q(w_n) \defeq \frac{\E_{(\bX\upn, \bigY\upn) \sim P\bayes_n}[|S\upn(q)|]}{\numnonnull}.
\end{equation}

% \begin{remark}\label{rem::numdisc_not_power} $\Power_q(w_n)$ counts the (normalized) expected number of discoveries instead of the number of \textit{true} discoveries. Intuitively, these quantities should be similar since knockoffs provably control the FDR. However, Section \ref{subsec::amlr}  
% extends our analysis to analyze the expected number of true discoveries.
% \end{remark}

\begin{theorem}\label{thm::avgopt} Consider any high-dimensional asymptotic regime where we observe data $\bX^{(n)} \in \R^{n \times p_n}, \bigY^{(n)} \in \R^{n}$ and knockoffs $\tbX^{(n)}$ with $D^{(n)}$ denoting the masked data. Let $P\bayes_n$ be a sequence of Bayesian models of the data satisfying Assumptions \ref{assump::sparsity}-\ref{assump::localdep}, and let $\mlr_n\uppi([\bX^{(n)}, \tbX^{(n)}], \bigY^{(n)})$ denote the MLR statistics with respect $P\bayes_n$. Let $w_n([\bX\upn, \tbX\upn], \bigY\upn)$ denote any other sequence of feature statistics. 

Then, if the limits $\lim_{n \to \infty} \Power_q(w_n)$ and $\lim_{n \to \infty} \Power_q(\mlr_n\uppi)$ exist for $q \in (0,1)$, we have that
\begin{equation}
    \lim_{n \to \infty} \Power_q(\mlr_n\uppi) \ge \lim_{n \to \infty} \Power_q(w_n)
\end{equation}
holds for all but countably many $q$.
\end{theorem}

Theorem \ref{thm::avgopt} shows that MLR statistics asymptotically maximize the (normalized) number of expected discoveries without any explicit assumptions on the relationship between $\bigY$ and $\bX$ or the dimensionality. Besides Assumptions \ref{assump::sparsity}-\ref{assump::localdep}, we also assume that the quantities we aim to study actually exist, i.e., $\lim_{n \to \infty} \Power_q(w_n)$ and $\lim_{n \to \infty} \Power_q(\mlr_n\uppi)$ exist---however, even this assumption can be relaxed (see Appendix \ref{subsec::assumptions}). Yet the weakest aspect of Theorem \ref{thm::avgopt} is that MLR statistics are only provably optimal under $P\bayes$. If $P_j\opt$ and $P_j\bayes$ are quite different, MLR statistics may not perform well. For this reason, Section \ref{subsec::gams} suggests practical choices of $P\bayes$ that performed well empirically, even under misspecification.
%Indeed, the first two assumptions of Theorem \ref{thm::avgopt} are quite weak: the first assumption merely guarantees that the limiting powers we aim to study actually exist, and the second is only a mild restriction on the sparsity regime. Indeed, our setting allows for many previously studied sparsity regimes, such as a polynomial sparsity model \citep{donohojin2004, ke2020} and the linear sparsity regime \citep{weinstein2017}. Nonetheless, these assumptions can be substantially weakened at the cost of a more technical theorem statement, as discussed in Appendix \ref{subsec::assumptions}.

\subsection{Maximizing the expected number of true discoveries}\label{subsec::amlr}

We now introduce \textit{adjusted} MLR (AMLR) statistics, which asymptotically maximize the number of expected \emph{true} discoveries under $P\bayes$. Empirically, AMLR and MLR statistics perform similarly, 
but AMLR statistics are less intuitive and depend somewhat counterintuitively on the FDR level. (This is why our paper focuses mostly on MLR statistics.) Thus, for brevity, this section gives only a little intuition and a slightly informal theorem statement. Please see Appendix \ref{appendix::tpr} for a rigorous theorem statement.

We begin with notation. For $\theta \in \Theta$, $\mcH_1(\theta) \subset [p]$ denotes the set of non-nulls under $P\uptheta$ and $\mcH_1(\theta\opt) \subset [p]$ denotes the random set of non-nulls under $P\bayes$. Then, $P\bayes(\MLR_j\uppi > 0, j \in \mcH_1(\theta\opt) \mid D)$ is the conditional probability that $\MLR_j\uppi$ is positive and the $j$th feature is non-null given the masked data. Finally, define the following ratio $\nu_j$:
\begin{equation}
    \nu_j = \frac{P\bayes(\MLR_j\uppi > 0, j \in \mcH_1(\theta\opt) \mid D)}{(1+q)^{-1} - P\bayes(\MLR_j\uppi > 0 \mid D)}.
\end{equation}
\begin{definition}\label{def::amlr} With this notation, we now define AMLR statistics $\{\AMLR_j\}_{j=1}^p$ in two cases.
\begin{itemize}[itemsep=0.5pt, topsep=0pt, leftmargin=*]
    \item Case 1: $\AMLR_j\bayes = \MLR_j\bayes$ if $P\bayes(\MLR_j\bayes > 0 \mid D) \ge (1+q)^{-1}$.
    \item Case 2: Otherwise, with $\logit(x) \defeq \log(x/(1-x))$, we define
\begin{equation}\label{eq::amlr}
    \AMLR_j\bayes = \sign(\MLR_j\bayes) \cdot \logit\left((1+q)^{-1}\right) \cdot \logit^{-1}\left(\nu_j \right).
\end{equation}
By construction, all AMLR statistics in Case 2 have smaller absolute values than all statistics in Case 1. Note that Appendix \ref{subappendix::computing_amlr} shows how to compute AMLR statistics. 
\end{itemize}
\end{definition}

\begin{corollary}\label{cor::amlr_valid} AMLR statistics from Definition \ref{def::amlr} are valid knockoff statistics.
\end{corollary}

MLR and AMLR statistics have the same signs but different absolute values. To understand why, Appendix \ref{appendix::tpr} argues that maximizing the expected number of true discoveries can be formulated as a simple linear program where the ``benefit" of prioritizing a feature is $b_j \defeq P\uppi(\MLR_j\uppi > 0, j \in \mcH_1(\theta\opt) \mid D)$---the probability that $\MLR_j\uppi$ is positive and $j$ is non-null---and the ``cost" is $c_j \defeq (1+q)^{-1} - P\uppi(\MLR_j\uppi > 0 \mid D)$. The intuition is that to make $k$ discoveries, $\approx (1+q)^{-1} k$ of the $k$ feature statistics with the largest absolute values must have positive signs. Thus, $c_j$ measures the difference between $(1+q)^{-1}$ and the (conditional) probability that $\MLR_j\uppi$ is positive. Feature $j$ has a \textit{negative} cost $c_j < 0$ if it produces a ``surplus" of $\ge (1+q)^{-1}$ positive signs in expectation.

The optimal solution to this problem is to (a) maximally prioritize all features with negative costs by giving them the highest absolute values---i.e., the features in Case 1 above---and (b) prioritize all other features in descending order of the benefit-cost ratio $\nu_j = b_j / c_j$. This is accomplished by the AMLR formulas in Definition \ref{def::amlr}. In contrast, MLR statistic magnitudes are a decreasing function of only the costs $c_j$. By incorporating the benefit $b_j$, AMLR statistics reduce the expected number of discoveries while increasing the expected number of \textit{true} discoveries. See Appendix \ref{appendix::tpr} for further details.

AMLR and MLR statistics are different but not \textit{too} different, since typically, $\{P\bayes(\MLR_j\bayes > 0, j \in \mcH_1(\theta\opt) \mid D)\}_{j=1}^p$ has a similar order as $\{P\bayes(\MLR_j\bayes > 0 \mid D)\}_{j=1}^p$. (When the orders are the same, AMLR and MLR statistics yield identical rejection sets.) Indeed, Figure \ref{fig::amlr} shows in a simple simulation that the power of AMLR and MLR statistics is nearly identical.

We now show that AMLR statistics asymptotically maximize power under $P\bayes$ (see Appendix \ref{appendix::tpr} for a formal statement and proof). For any statistic $w([\bX, \tbX], \bigY)$ with discovery set $S_w \subset [p]$, let $\TP\bayes(w)$ denote the expected number of true positives under $P\bayes$:
\begin{equation}\label{eq::tpdef}
    \TP\bayes(w) \,\defeq\, \int_{\Theta} \E_{P\uptheta}\big[|S_w \cap \mcH_1(\theta)|\big] \, \pi(\theta) \, d\theta \,=\, \E_{P\bayes}[|S_w \cap \mcH_1(\theta\opt)|].
\end{equation}

\begin{theorem}[Informal]\label{thm::amlr_power_informal} Suppose the conditions of Theorem \ref{thm::avgopt} hold. Furthermore, suppose that (i) the local dependence condition in Assumption \ref{assump::localdep} holds when replacing $\I(\MLR_j\uppi > 0)$ with $\I(\MLR_j\uppi > 0, j \in \mcH_1(\theta\opt))$ and (ii) the coefficient of variation of the number of non-nulls $|\mcH_1(\theta\opt)|$ is bounded as $n \to \infty$. Then for any sequence of feature statistics $\{w_n\}_{n \in \N}$,
\begin{equation}
    \TP\bayes(\amlr) \ge \TP\bayes(w_n) + o(s_n),
\end{equation}
where $s_n$ is the expected number of non-nulls under $P_n\bayes$, as defined in Assumption \ref{assump::sparsity}.
\end{theorem}

\begin{figure}
    \centering
    \includegraphics[width=\linewidth]{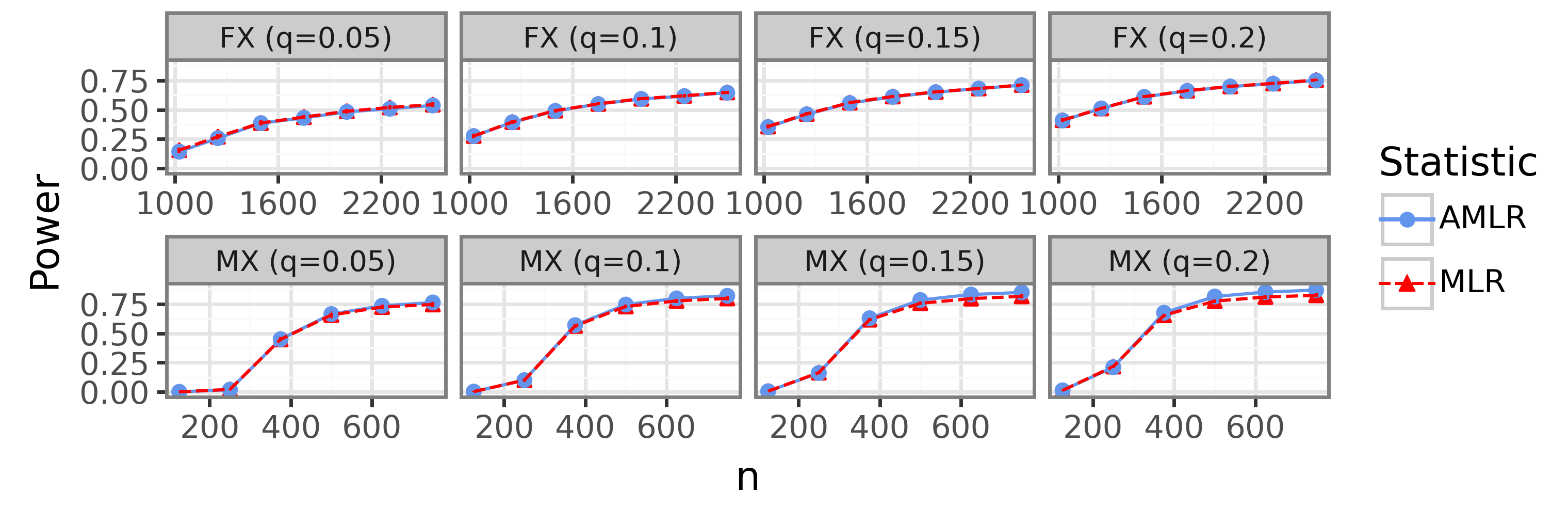}
    \caption{In the AR(1) simulation setting from Section \ref{subsec::linear_sims}, this figure plots the power of MLR and AMLR statistics (using MVR fixed-X or model-X knockoffs) for different nominal FDR levels $q$. It shows that both statistics have essentially the same power. See Section \ref{subsec::linear_sims} for details.
    }\label{fig::amlr}
\end{figure}

\section{Computing MLR statistics}\label{subsec::computation}

\subsection{General strategy}

We now show how to compute $P\bayes_j(\bX_j \mid D)$ and $P\bayes_j(\tbX_j \mid D)$ by Gibbs sampling from the law of $\bX \mid D$ under $P\bayes$. For brevity, we focus on the MX setting---Appendix \ref{appendix::gibbs} discusses the FX case. 

The key idea is that
\textit{conditional} on $\bX_{-j}$ and the latent parameter $\theta\opt$, sampling from the law of $\bX_j \mid \bX_{-j}, \theta\opt, D$ is easy. In particular, for any fixed $\bd = (\mathbf{y}, \{\bx_j, \tbx_j\}_{j=1}^p)$, observing $D = \bd$ implies that $\bX_j$ must lie in $\{\bx_j, \tbx_j\}$. Lemma \ref{lem::gibbs_sampler_result} shows that the conditional likelihood ratio equals:
\begin{align}
        \frac{P\bayes(\bX_j = \bx_j \mid \bX_{-j}, \theta\opt = \theta, D = \bd)}{P\bayes(\bX_j = \tbx_j \mid \bX_{-j}, \theta\opt = \theta, D = \bd)} 
    &=
        \frac{P_{\bigY \mid \bX}\uptheta(\smally \mid \bX_j = \bx_j, \bX_{-j})}{P_{\bigY \mid \bX}\uptheta(\smally \mid \bX_j = \tbx_j, \bX_{-j})} \label{eq::resamplexj}.
\end{align}
The right-hand side of Eq. (\ref{eq::resamplexj}) is easy to compute for most parametric models $\mcP$, since it only involves computing the likelihood of $\bigY$ given $\bX$. Thus, we can easily sample from the law of $\bX_j \mid D, \theta\opt, \bX_{-j}$. 

To sample from the law of $\bX \mid D$, Algorithm \ref{alg::mlrgibbs} describes a Gibbs sampler which (i) for $j \in [p]$, resamples from $\bX_j \mid \bigY, \bX_{-j}, \theta\opt$ and (ii) resamples from the posterior of $\theta\opt \mid \bigY, \bX$. Step (ii) can be done using any off-the-shelf Bayesian sampler  \citep{mcmcbook2011}, since this step is identical to a typical Bayesian regression. Lemma \ref{lem::gibbs} shows that Algorithm \ref{alg::mlrgibbs} correctly computes the MLR statistics as $n_\sample \to \infty$ under standard regularity conditions \citep{robertcasella2004mcmc}. These mild conditions are satisfied by our default choices (Example \ref{ex::sparsegam}), but they can also be relaxed further (see Appendix \ref{subsec::gibbsproof}).

\begin{algorithm}[h!]
\caption{Gibbs sampling meta-algorithm to compute MLR statistics.}\label{alg::mlrgibbs}
\algorithmicensure\,  $\bigY, \bX, \tbX$,  
a model class $\{P\uptheta : \theta \in \Theta\}$ and prior $\pi : \Theta \to \R_{\ge 0}$.

\begin{algorithmic}[1]
        \State Initialize $\theta^{(0)} \sim \pi$ and 
        $\bX_j^{(0)} \simind \Unif(\{\bX_j, \tbX_j\})$ for $j \in [p]$. \smallskip \Comment{Initialization}
        % and $(\bX_j^{(0)}, \tbX_j^{(0)}) \simind \Unif\left(\left\{(\bX_j, \tbX_j), (\tbX_j, \bX_j)\right\}\right)$ for $j \in [p]$.
        \For {$i=1,2,\dots,n_\sample$}:
            \State Initialize $\bX^{(i)} = \bX^{(i-1)} \in \R^{n \times p}$. \medskip
            %and $\theta^{(i)} = \theta^{(i-1)} \in \R^{k}$.
            \For {$j=1,\dots,p$}:  \Comment{Resample $\bX^{(i)}$}
                \State Set $\eta_j^{(i)} = \log\left(P_{\bigY \mid \bX}^{(\theta_i)}(\bigY \mid [\bX_{-j}^{(i)}, \bX_j] \right) - \log\left(P_{\bigY \mid \bX}^{(\theta_i)}(\bigY \mid [\bX_{-j}^{(i)}, \tbX_j])\right)$.
                % \item Set $p_j^{(i)} = \exp\left(\frac{\exp(\eta_j^{(i)}}{1 + \exp(\eta_j^{(i))}}\right)$
                \State Define $p_j^{(i)} = \logit^{-1}(\eta_j^{(i)})$.
                \State Set $\bX_j^{(i)} = \bX_j$ with probability $p_j^{(i)}$.  Else set $\bX_j^{(i)} = \tbX_j$.
            \EndFor
            \medskip 

            %\For {$k=1,\dots,K$}:
            \State Sample $\theta^{(i)}$ from the law of $\theta\opt \mid \bigY, \bX = \bX^{(i)}$ under $P\bayes$.  \Comment{Resample $\theta^{(i)}$}
                %$\theta^{(i)}_k \sim P\bayes(\theta_k\opt = \cdot \mid \bX = \bX^{(i)}, \bigY = \bigY, \theta\opt_{-k} = \theta_{-k}^{(i)})$.
            %\EndFor
        \EndFor
    \smallskip 
    \State Return $\MLR_j\uppi = \log\left(\sum_{i=1}^{n_\sample} p_j^{(i)}\right) - \log\left(\sum_{i=1}^{n_\sample} 1 - p_j^{(i)}\right)$, for $j \in [p]$.
	\end{algorithmic} 
\end{algorithm}

\begin{lemma}\label{lem::gibbs} Let $p_j^{(i)}$ be defined as in Algorithm \ref{alg::mlrgibbs}. Suppose that under $P\bayes$, (i) $p_j^{(i)} \in (0,1)$ a.s. for $j \in [p]$ and (ii) the support of $\theta\opt \mid \bX, \bigY$ equals $\Theta$. Then as $n_\sample \to \infty$,
\begin{equation*}
        \log\left(\sum_{i=1}^{n_\sample} p_j^{(i)}\right) 
        - \log\left(\sum_{i=1}^{n_\sample} 1 - p_j^{(i)}\right)
    \toprob
        \MLR_j\uppi
    \defeq 
        \log\left(\frac{P_j\bayes(\bX_j \mid D)}{P_j\bayes(\tbX_j \mid D)}\right).
\end{equation*}
\end{lemma}

\begin{remark}\label{rem::covs} Algorithm \ref{alg::mlrgibbs} also allows us to diagnose Assumption \ref{assump::localdep}. Prop. \ref{prop::mxdistinguish} yields that if $\widehat{\bX}_j = \argmax_{\bx \in \{\bX_j, \tbX_j\}} P_j\bayes(\bx \mid D)$, then
\begin{align}
    \cov_{P\bayes}(\I(\MLR_k\uppi > 0), \I(\MLR_j\uppi > 0) \mid D)
    =
    \cov_{P\bayes}(\I(\widehat{\bX}_k = \bX_k), \I(\widehat{\bX}_j = \bX_j) \mid D).
\end{align}
Thus, we can approximate the covariance above with the empirical covariance of $\{\I(\bX_j^{(i)} = \widehat{\bX}_j)\}_{i=1}^{n_\sample}$ and $\{\I(\bX_k^{(i)} = \widehat{\bX}_k)\}_{i=1}^{n_\sample}$, where $\{\bX^{(i)}\}_{i=1}^{n_\sample}$ are the samples from Algorithm \ref{alg::mlrgibbs}.
\end{remark}

\begin{remark}
In the special case of Gaussian linear models with a sparse prior on the coefficients $\beta$, Algorithm \ref{alg::mlrgibbs} is similar in flavor to the ``Bayesian Variable Selection" (BVS) feature statistic from \cite{mxknockoffs2018}, although there are differences in the Gibbs sampler and the final estimand. Broadly, we see our work as complementary to theirs. Yet aside from technical details, a main difference is that \cite{mxknockoffs2018} seemed to argue that the advantage of BVS was to incorporate accurate prior information. In contrast, we argue that MLR statistics can improve power even without prior information  (see Section \ref{sec::sims}) by estimating the right notion of variable importance. 
\end{remark}

\subsection{A default choice of Bayesian model}\label{subsec::gams}

Below, we describe a class of Bayesian models that is computationally efficient and can flexibly model both linear and nonlinear relationships. Note that to specify $\mcP$, it suffices to model the law of $\bigY \mid \bX$, since the law of $\bX$ is assumed known in the MX case and $\bX$ is fixed in the FX case.

\begin{example}[Sparse generalized additive model]\label{ex::sparsegam} For linear coefficients $\beta\upj \in \R^{k_j}$ and noise variance $\sigma^2 \in \R$, let $\theta = (\beta^{(1)}, \dots, \beta^{(p)}, \sigma^2) \in \Theta \defeq \R^K \times \R_{\ge 0}$. For a prespecified set of basis functions $\phi_j : \R \to \R^{k_j}$, we consider the model class $\mcP = \{P\uptheta : \theta \in \Theta\}$ where
\begin{equation}
\bigY_i \mid \bX \simind \mcN\left(\sum_{j=1}^p \phi_j(\bX_{ij})^T \beta\upj, \sigma^2\right) \text{ for } i = 1, \dots, n \text{ under } P\uptheta.
\end{equation}
By default, we take $\phi_j$ to be the identity function, which reduces to a Gaussian linear model. However, if $\bigY$ and $\bX$ may have nonlinear relationships, we suggest taking $\phi_j(\cdot)$ to be the basis representation of regression splines (see \cite{esl2001} for review), as we do in Section \ref{subsec::nonlinear_sims}. 
For the prior, we let $\pi$ denote the law of $\theta\opt$ after sampling from the following process:
\begin{itemize}[noitemsep, topsep=0pt]
    \item Sample hyperparameters $p_0 \sim \Beta(a_0, b_0)$ (sparsity), $\tau^2 \sim \invGamma(a_{\tau}, b_{\tau})$ (signal size), and $\sigma^2 \sim \invGamma(a_{\sigma}, b_{\sigma})$ (noise variance). By default, we take $a = b = b_{\tau} = b_{\sigma} = 1, a_{\tau} = a_{\sigma} = 2$.
    \item Sample $\beta\upj = B_j Z_j$ for $Z_j \simind \mcN(0, \tau^2 I_{k_j})$ and $B_j \iid \Bern(1-p_0).$
\end{itemize}
This group-sparse prior is effectively a ``two-groups" model, as $\bX_j$ is null if and only if $\beta\upj = 0$. As shown in Section \ref{sec::sims}, using these hyperpriors allows us to adaptively estimate the sparsity level. 
\end{example}

Standard techniques for ``spike-and-slab" models \citep{mcculloch1997} allow us to compute the MLR statistics from Ex. \ref{ex::sparsegam} in $O(n_\sample n p)$ operations (assuming  $\sum_{j=1}^p k_j = O(p)$)---see Appendix \ref{appendix::gibbs} for review. This cost is cheaper than computing Gaussian MX or FX knockoffs, which requires $O(n p^2  + p^3)$ operations. Fitting the LASSO has a comparable cost, which is $O(n_\sample n p)$ using coordinate descent or $O(n p^2)$ using the LARS algorithm \citep{efron2004}.

Lastly, we can easily extend this algorithm to binary responses. In particular, using techniques from \cite{albertchib1993}, we can compute Gibbs updates in the same computational complexity when $P\bayes(Y = 1 \mid X) = \Phi\left(\sum_{j=1}^p \phi_j(X_j)^T \beta\upj \right)$, where $\Phi$ is the Gaussian CDF (see Appendix \ref{appendix::gibbs} for details).

\section{Simulations}\label{sec::sims}

We now analyze the power of MLR statistics in simulations. Throughout, MLR statistics do \textit{not} have accurate prior information: we use exactly the same choice of Bayesian model $P\bayes$ (the default from Section \ref{subsec::gams}) to compute MLR statistics in every plot. Also, we let $\bX$ be highly correlated to test whether MLR statistics perform well even when Assumption \ref{assump::localdep} may fail. Nonetheless, MLR statistics uniformly outperform existing competitors.%, suggesting that they are robust to misspecification of the Bayesian model and highly correlated features.

The FDR level is $q=0.05$. All plots have two standard deviation error bars, although the bars may be too small to be visible. In each plot, knockoffs provably control the frequentist FDR, so we only plot power. All code is available at \url{https://github.com/amspector100/mlr_knockoff_paper}.

\subsection{Gaussian linear models}\label{subsec::linear_sims}

\begin{figure}
    \centering
    \includegraphics[width=\linewidth]{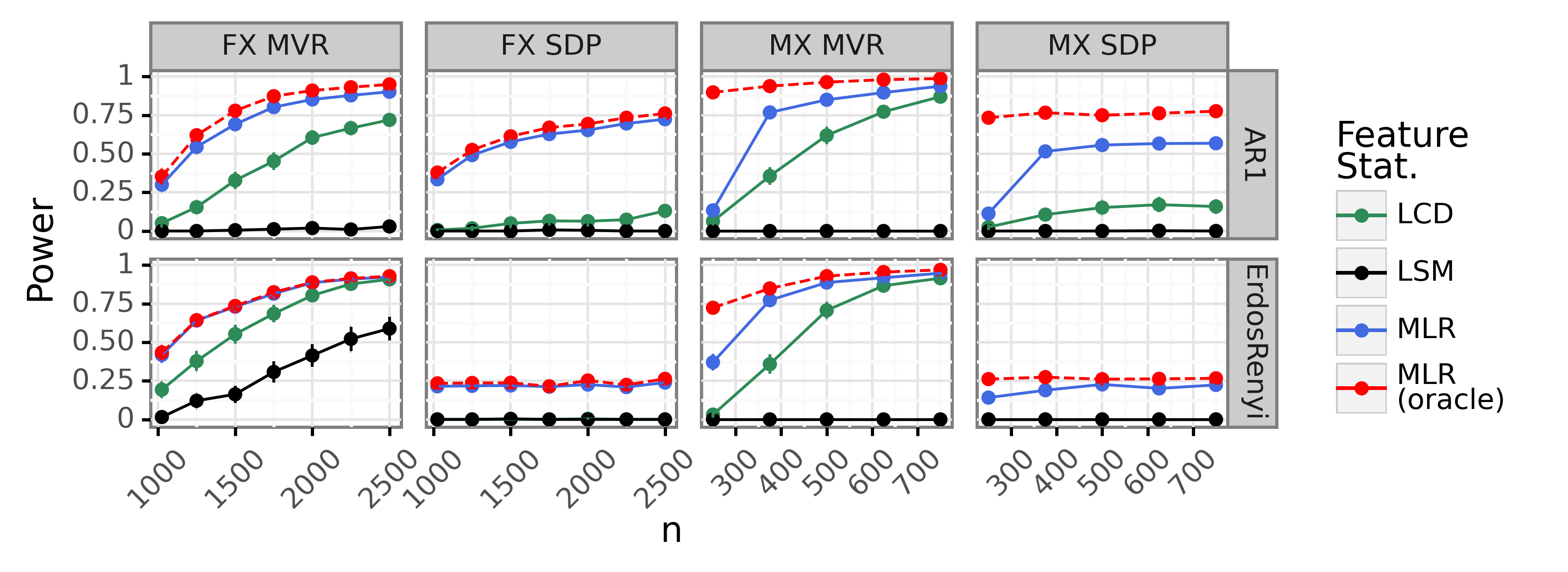}
    \caption{Power of MLR, LCD, and LSM statistics in a sparse Gaussian linear model with $p=500$ and $50$ non-nulls. For FX knockoffs, MLR statistics almost exactly match the power of the oracle procedure which provably upper bounds the power of any feature statistic. For MX knockoffs, MLR statistics are slightly less powerful than the oracle, although they are still very powerful compared to the lasso-based statistics. Note that the power of knockoffs can be roughly constant in $n$ in the ``SDP" setting: this is because SDP knockoffs sometimes have identifiability issues \citep{mrcknock2022}. Also, note that the ``LCD" and ``LSM" curves completely overlap in two of the bottom panels, where both methods have zero power. See Appendix \ref{appendix::simdetails} for precise simulation details.}
    \label{fig::linear_model}
\end{figure}

In this section, we sample $\bigY \mid \bX \sim \mcN(\bX \beta, I_n)$ for sparse $\beta$. We draw $X \simiid \mcN(0, \Sigma)$ for two choices of $\Sigma$. By default, $\Sigma$ corresponds to a highly correlated nonstationary AR(1) process, inspired by real genetic design matrices. However, we also analyze an ``ErdosRenyi" covariance matrix where $\Sigma$ is $80\%$ sparse with the nonzero entries drawn uniformly at random. We compute both ``SDP" and ``MVR" knockoffs \citep{mxknockoffs2018, mrcknock2022} to show that MLR statistics perform well in both cases. See Appendix \ref{appendix::simdetails} for further simulation details.

We compare four feature statistics. First, we compute MLR statistics using the default Bayesian model from Section \ref{subsec::computation}---in plots, ``MLR" refers to this version of MLR statistics. Second, we compute LCD and LSM statistics as described in Section \ref{subsec::hiv_motiv}. Lastly, we compute the oracle MLR statistics which have full knowledge of the true value of $\beta$. Figure \ref{fig::linear_model} shows the results while varying $n$ in low dimensions (using FX knockoffs) and high dimensions (using MX knockoffs). It shows that MLR statistics are substantially more powerful than the lasso-based statistics and, in the FX case, MLR statistics almost perfectly match the power of the oracle. Indeed, this result holds even for the ``ErdosRenyi" covariance matrix, where $\bX$ exhibits strong non-local dependencies (in contrast to Assumption \ref{assump::localdep}). Furthermore, Figure \ref{fig::comptime} shows that MLR statistics are computationally efficient, often faster than a cross-validated lasso and comparable to the cost of computing FX knockoffs.

\begin{figure}[!ht]
    \centering
    \includegraphics[width=\linewidth]{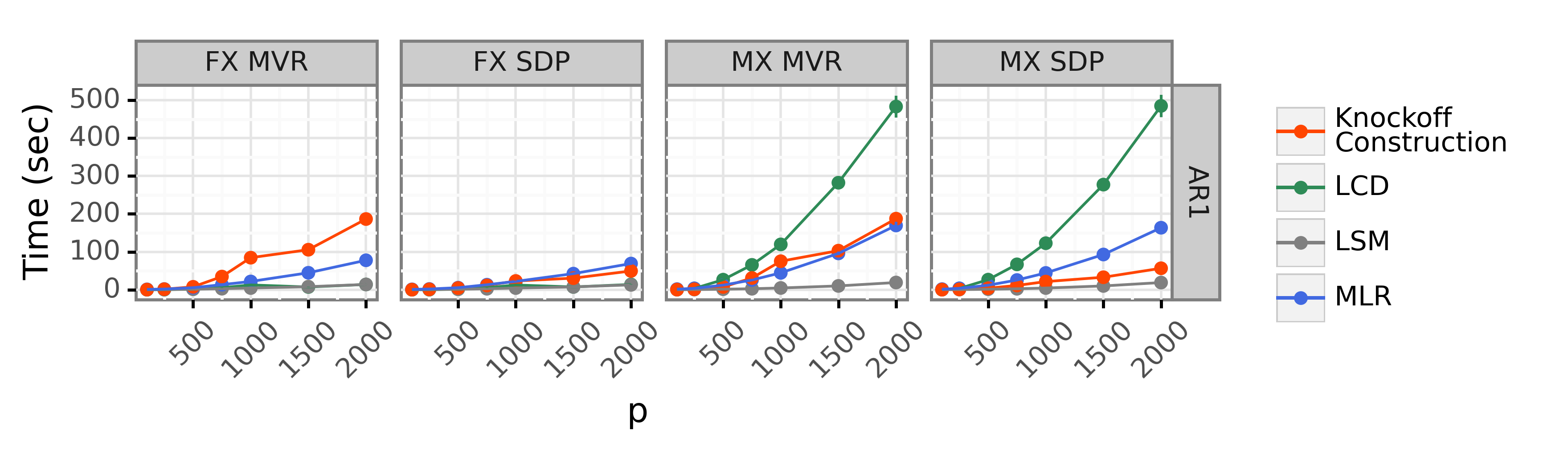}
    \caption{This figure shows the computation time for various feature statistics in the same setting as Figure \ref{fig::linear_model}, as well as the cost of computing knockoffs. It shows that MLR statistics are competitive with state-of-the-art feature statistics (in the model-X case) or comparable to the cost of computing knockoffs (in the fixed-X case).}
    \label{fig::comptime}
\end{figure}

Next, we analyze the performance of MLR statistics when the prior is misspecified. In Figure \ref{fig::misspec}, we vary the sparsity (proportion of non-nulls) between $5\%$ and $40\%$, and we draw the non-null coefficients as (i) heavy-tailed i.i.d. Laplace variables and (ii) ``light-tailed" i.i.d. $\Unif([-1/2, -1/4] \cup [1/4, 1/2])$ variables. In all cases, the MLR prior assumes the non-null coefficients are i.i.d. $\mcN(0, \tau^2)$ with sparsity $p_0 \sim \Beta(1, 1)$. Nonetheless, MLR statistics consistently outperform the lasso-based statistics and nearly match the performance of the oracle.

\begin{figure}[!ht]
    \centering
    \includegraphics[width=\linewidth]{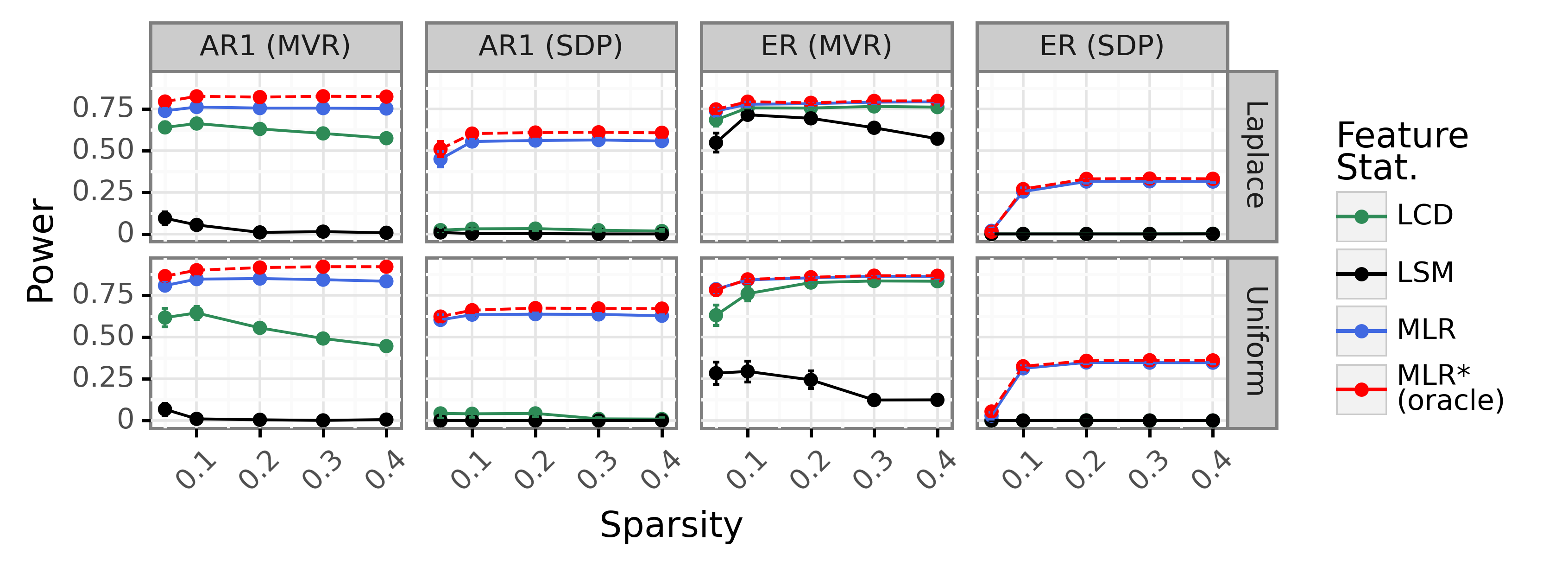}
    \caption{This figure shows the power of MLR, LCD, and LSM statistics when varying the sparsity level and drawing the non-null coefficients from a heavy-tailed (Laplace) and light-tailed (Uniform) distribution, with $p=500$ and $n=1250$. The setting is otherwise identical to the AR1 setting from Figure \ref{fig::linear_model}. It shows that the MLR statistics perform well despite using the same (misspecified) prior in every setting.}
    \label{fig::misspec}
\end{figure}

Lastly, we verify that the local dependence condition assumed in Theorem \ref{thm::avgopt} holds empirically. We consider the AR(1) setting but modify the parameters so that $\bX$ is extremely highly correlated, with adjacent correlations drawn as i.i.d. $\Beta(50, 1)$ variables. We also consider a setting where $\bX$ is equicorrelated with correlation $95\%$. In both cases, Figure \ref{fig::dependence} shows that $\cov_{P\bayes}(\I(\MLR\uppi > 0) \mid D)$ has entries which decay off the main diagonal---in fact, the maximum off-diagonal covariance across both examples is $0.07$. Please see Section \ref{subsec::avgopt} and Appendix \ref{appendix::dependencediscussion} for intuition behind this result, although we cannot perfectly explain it.

\begin{figure}[!ht]
    \centering
    \includegraphics[width=\linewidth]{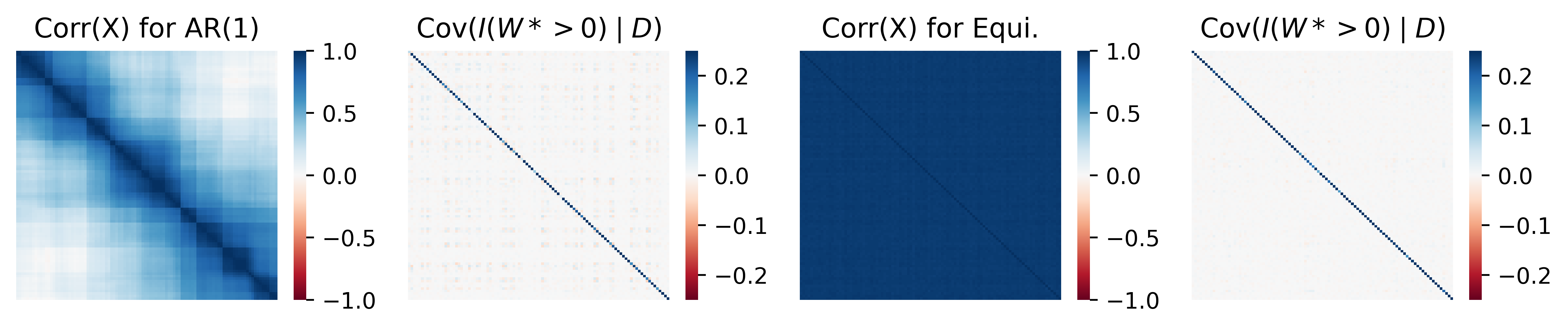}
    \caption{In the AR(1) and equicorrelated settings, this plot shows both the correlation matrix of $\bX$ as well as the conditional covariance of signs of the MLR statistic $\MLR\uppi$, computed as per Remark \ref{rem::covs}. It shows that even when $\bX$ is very highly correlated, the signs of $\MLR\uppi$ are only locally dependent. Note in this plot, every feature is non-null and the power of knockoffs is $53\%$ and $10\%$ for the AR(1) and equicorrelated settings, respectively.}
    \label{fig::dependence}
\end{figure}

\subsection{Generalized additive models}\label{subsec::nonlinear_sims}

We now sample $\bigY \mid \bX \sim \mcN(h(\bX) \beta, I_n)$ for some non-linear function $h : \R \to \R$ applied element-wise to $\bX \in \R^{n \times p}$. We consider the AR(1) setting from Section \ref{subsec::linear_sims} with four choices of $h$: $h(x) = \sin(x), h(x) = \cos(x), h(x) = x^2, h(x) = x^3$. We compare six feature statistics: linear MLR statistics, MLR based on cubic regression splines with one knot, the LCD, a random forest with swap importances as in \cite{knockoffsmass2018}, and DeepPINK \citep{deeppink2018}, which is based on a feedforward neural network. This setting is more challenging than the linear setting, since the feature statistics must learn (or approximate) the function $h$. Thus, our simulations in this section are low-dimensional with $n > p$, and we should not expect any feature statistic to match the performance of the oracle MLR statistics.
%%COULDCUT last clause?

\begin{figure}[!ht]
    \centering
    \includegraphics[width=\linewidth]{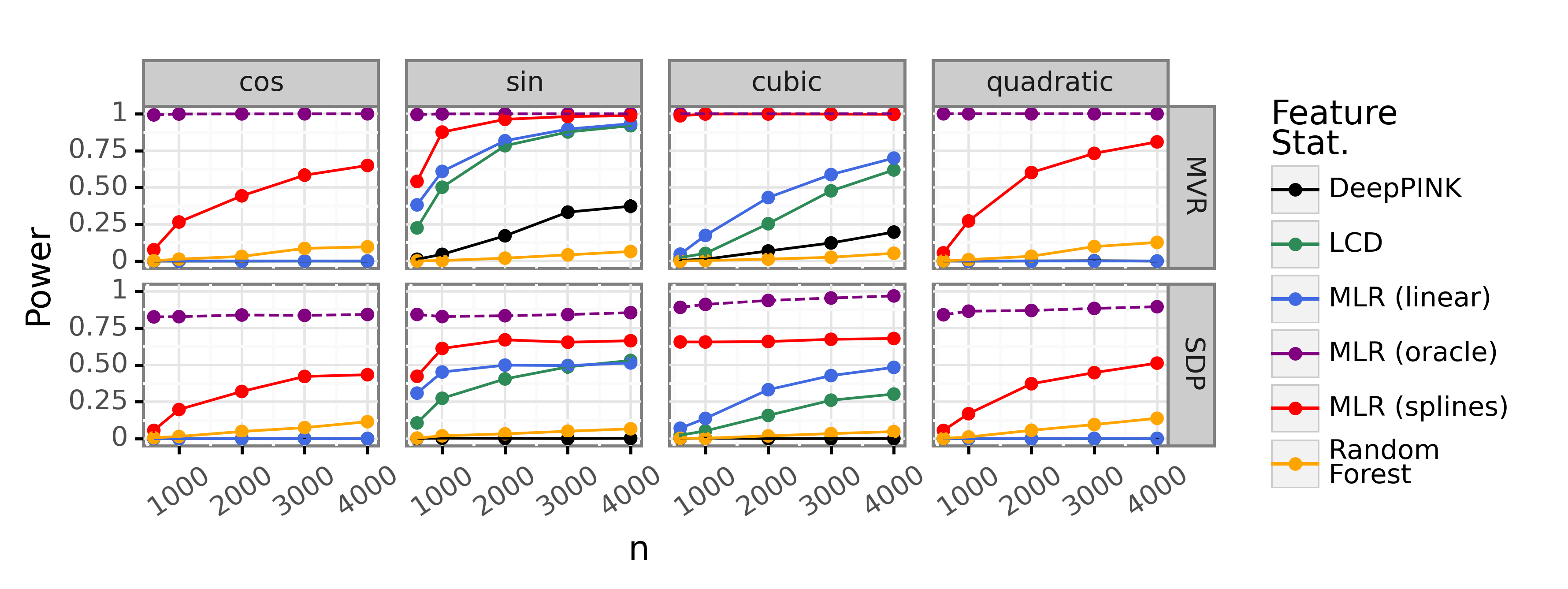}
    \caption{This figure plots power in generalized additive models where $\bigY \mid \bX \sim \mcN(h(\bX) \beta, I_n)$, for $h : \R \to \R$ applied elementwise to $\bX$. The x-facets show the choice of $h$, and the y-facets show results for both MVR and SDP MX knockoffs; note $p=200$ and there are $60$ non-nulls. For this plot only, we choose $q=0.1$ because several of the competitors made zero discoveries at $q=0.05$. See Appendix \ref{appendix::simdetails} for the plot with $q=0.05$. Note that in the top ``cubic" panel, the MLR (splines) statistic has $100\%$ power and overlaps with the oracle.}
    \label{fig::nonlin}
\end{figure}

Figure \ref{fig::nonlin} shows that ``MLR (splines)" uniformly outperforms every other feature statistic, often by wide margins. Linear MLR and LCD statistics are powerless in the $\cos$ and quadratic settings, where $h$ is an even function and thus the non-null features have no linear relationship with the response. However, in the $\sin$ and cubic settings, linear MLR statistics outperform the LCD, suggesting that linear MLR statistics can be powerful under misspecification as long as there is some linear effect.

\subsection{Logistic regression}\label{subsec::logit_sims}

Lastly, we now consider the setting of logistic regression, so we sample $Y \mid X \sim \Bern(s(X^{\top} \beta))$ where $s$ is the sigmoid function. We run the same simulation setting as Figure \ref{fig::linear_model}, except that now $Y$ is binary and we consider low-dimensional settings, since inference in logistic regression is generally more challenging than in linear regression. The results are shown by Figure \ref{fig::logistic}, which shows that MLR statistics outperform the LCD, although there is a substantial gap between the performances of the MLR and oracle MLR statistics. %MLR statistics do take $\approx 3$ times longer to compute than the LCD in this setting due to the Gibbs sampler used for binary regression (see Appendix \ref{appendix::gibbs}). %It may be possible to speed up the Gibbs sampler, but we leave this possibility to future work.

\begin{figure}[!ht]
    \centering
    \includegraphics[width=\linewidth]{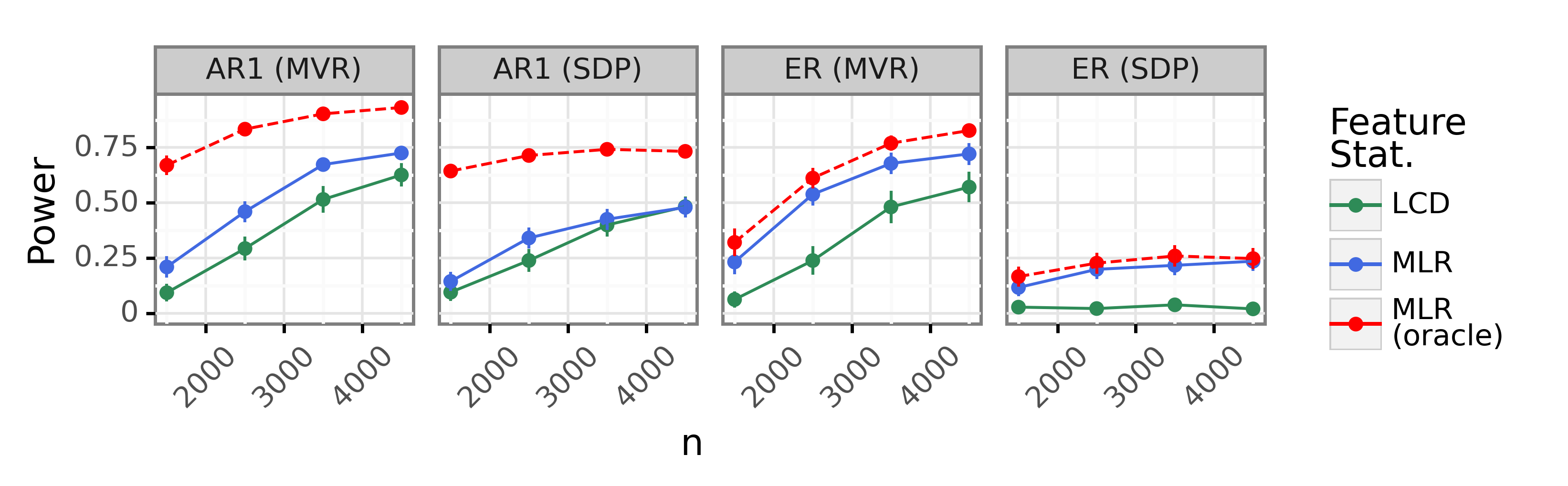}
    \caption{This plot shows the power of MLR statistics compared to the cross-validated LCD in logistic regression, with $p=500$, $50$ non-nulls, and $n$ varied between $1500$ and $4500$. The setting is otherwise identical to Figure \ref{fig::linear_model}.}
    \label{fig::logistic}
\end{figure}

\section{Real applications}\label{sec::data}

In this section, we apply MLR statistics to three real datasets which have been previously analyzed using knockoffs. We use the same default choice of MLR statistics from our simulations in all three applications. In each case, MLR statistics have comparable or higher power than competitor statistics. All code and data are available at \url{https://github.com/amspector100/mlr_knockoff_paper}.

\subsection{HIV drug resistance}

We begin with the HIV drug resistance dataset from \cite{rhee2006data}, which (e.g.) \cite{fxknock} previously analyzed using knockoffs. The dataset consists of genotype data from $n \approx 750$ HIV samples as well as drug resistance measurements for $16$ different drugs, and the goal is to discover genetic variants that affect drug resistance for each of the drugs. Furthermore, \cite{rhee2005corroborate} published treatment-selected mutation panels for this setting, so we can check whether any discoveries made by knockoffs are corroborated by this separate analysis.

We preprocess and model the data following \cite{fxknock}. Then, we apply FX knockoffs with LCD, LSM, and MLR statistics and FDR level $q=0.05$. For both MVR and SDP knockoffs, Figure \ref{fig::hiv} shows the total number of discoveries made by each statistic, stratified by whether each discovery is corroborated by \cite{rhee2005corroborate}. For SDP knockoffs, the MLR statistics make nearly an order of magnitude more discoveries than the competitor methods with a comparable corroboration rate. For MVR knockoffs, MLR and LCD statistics perform roughly equally well, although MLR statistics make $\approx 5\%$ more discoveries with a slightly higher corroboration rate. Overall, in this setting, MLR statistics are competitive with and sometimes substantially outperform the lasso-based statistics. See Appendix \ref{appendix::realdata} for specific results for each drug.

\begin{figure}
    \centering
    \includegraphics[width=0.66\linewidth]{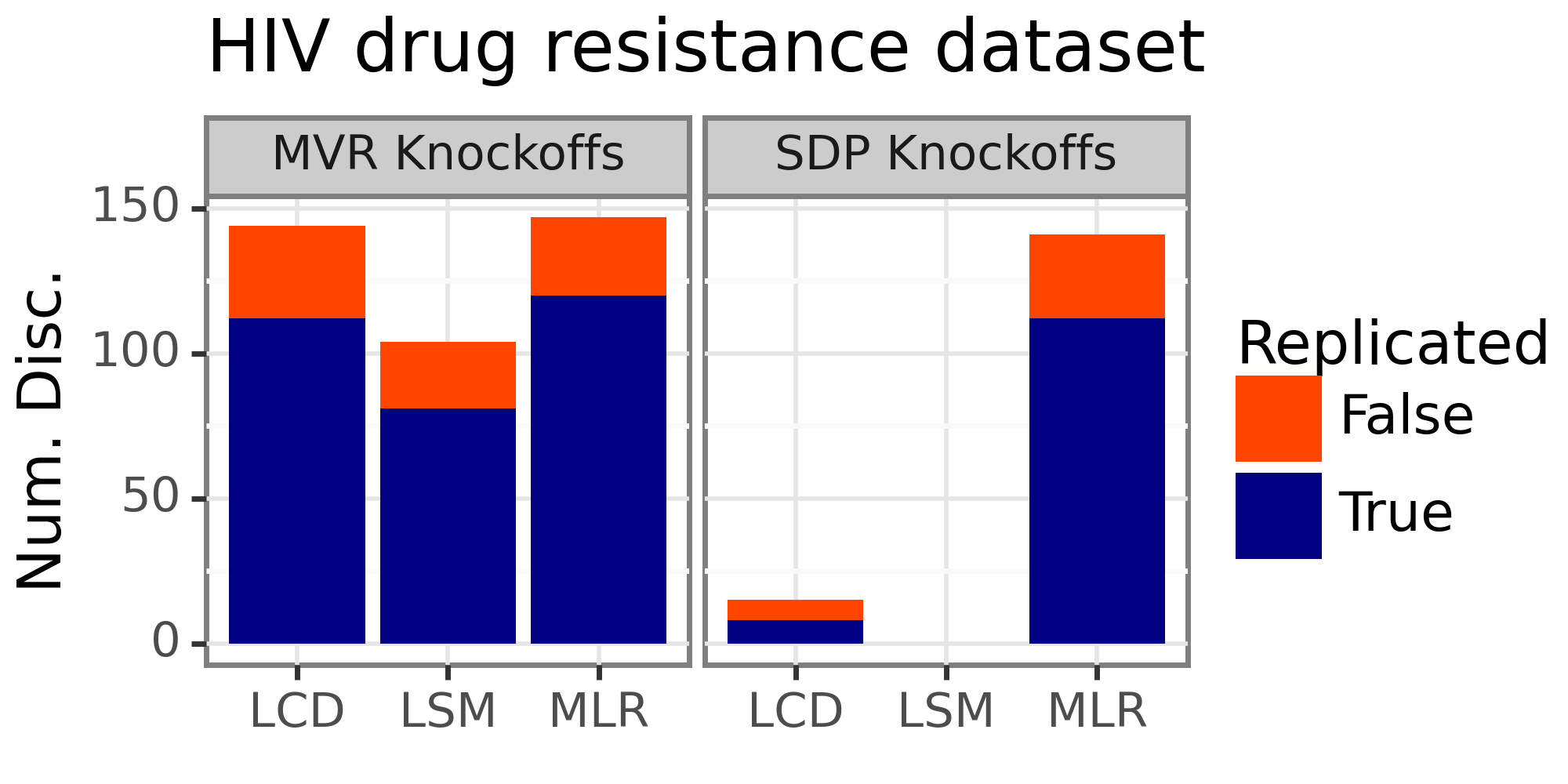}
    \caption{This figure shows the total number of discoveries made by the LCD, LSM, and MLR feature statistics in the HIV drug resistance dataset from \cite{rhee2006data}, summed across all $16$ drugs.}
    \label{fig::hiv}
\end{figure}

\subsection{Financial factor selection}\label{subsec::fundrep}

In finance, analysts often aim to select factors that drive the performance of a particular asset. \cite{challet2021} applied FX knockoffs to factor selection, and as a benchmark, they tested which US equities explain the performance of an index fund for the energy sector (XLE). Here, the ground truth is available since the index fund is a weighted combination of a known list of stocks.

We perform the same analysis for ten index funds of key sectors of the US economy, including energy, technology, and more (see Appendix \ref{appendix::realdata}). Here, $\bigY$ is the index fund's daily log return and $\bX$ contains the daily log returns of each stock in the S\&P 500 since $2013$, so $p \approx 500$ and $n \approx 2300$. We compute fixed-X MVR and SDP knockoffs and apply LCD, LSM, and MLR statistics. Figure \ref{fig::fund_rep} shows the number of true and false discoveries summed across all index funds with $q=0.05$. In particular, MLR statistics make $35\%$ and $78\%$ more discoveries than the LCD for MVR and SDP knockoffs (respectively), and the LSM makes more than $5$ times fewer discoveries than the MLR statistics. Thus, MLR statistics substantially outperform the lasso-based statistics. Appendix \ref{appendix::realdata} also shows that the FDP (averaged across all index funds) is well below $5\%$ for each method.

\begin{figure}
    \centering
    \includegraphics[width=0.66\linewidth]{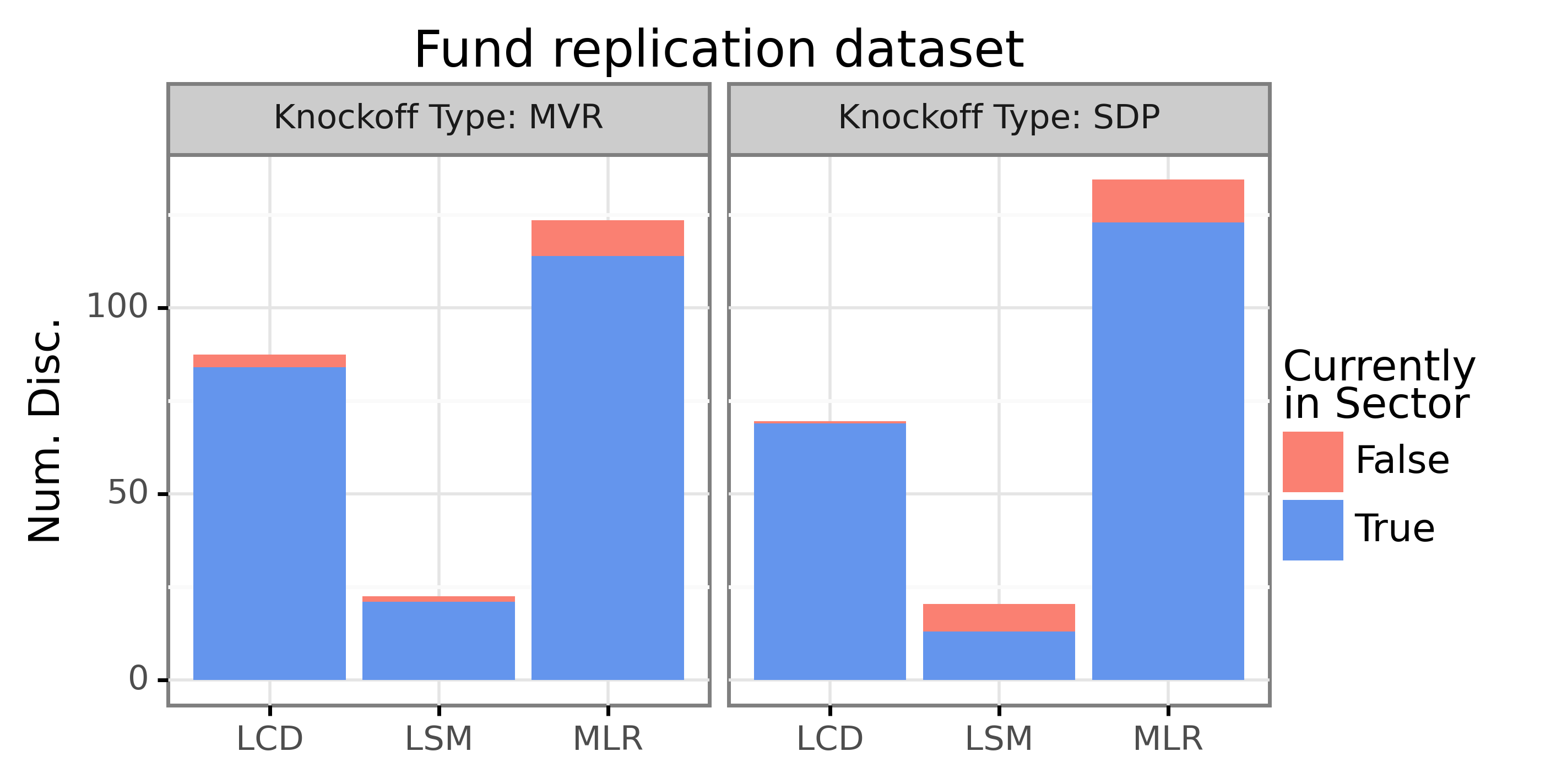}
    \caption{This figure shows the total number of discoveries made by each method in the fund replication dataset inspired by \cite{challet2021}, summed across all ten index funds. See Appendix \ref{appendix::realdata} for a table showing that the average FDP for each method is below the nominal level of $q=0.05$.}
    \label{fig::fund_rep}
\end{figure}

\subsection{Graphical model discovery for gene networks}

Lastly, we consider the problem of recovering a gene network from single-cell RNAseq data. Our analysis follows \cite{nodewiseknock}, who model gene expression log counts as a Gaussian graphical model (see \cite{nodewiseknock} for justification of the Gaussian assumption). In particular, they develop an extension of FX knockoffs that detects edges in Gaussian graphical models while controlling the FDR across discovered edges. They applied this method to RNAseq data from \cite{zheng2017ggmdata}. The ground truth is not available, so following \cite{nodewiseknock}, we only evaluate methods based on the number of discoveries they make.

We replicate this analysis and compare LCD, LSM, and MLR statistics. Figure \ref{fig::ggm} plots the number of discoveries as a function of $q \in [0, 0.5]$. MLR statistics make the most discoveries for nearly every value of $q$, although often by a small margin. For small $q$, the LSM statistic performs poorly, and for large $q$, the LCD statistic performs poorly, whereas the MLR statistic is consistently powerful.

\begin{figure}
    \centering
    \includegraphics[width=0.66\linewidth]{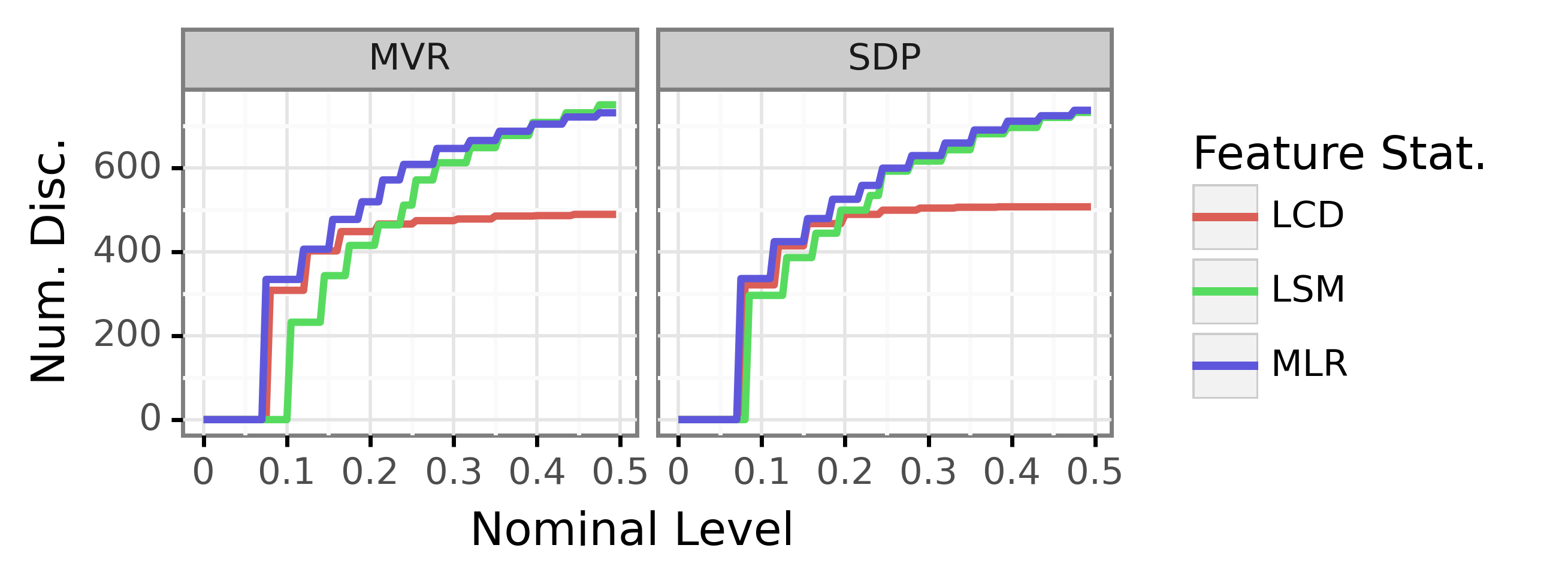}
    \caption{This figure shows the number of discoveries made by LCD, LSM, and MLR statistics when used to detect edges in a Gaussian graphical model for gene expression data, as in \cite{nodewiseknock}.}
    \label{fig::ggm}
\end{figure}

\section{Discussion}\label{sec::discussion}

This paper introduces masked likelihood ratio statistics, a class of asymptotically Bayes-optimal knockoff statistics. We show in simulations and three applications that MLR statistics are efficient and powerful. However, our work leaves open several directions for future research.

\begin{itemize}[topsep=0pt, itemsep=0.5pt, leftmargin=*]
    \item MLR statistics are asymptotically Bayes optimal. However, it might be worthwhile to develop \textit{minimax-optimal} knockoff-statistics, e.g., by computing a ``least-favorable" prior.

    \item Our theory requires a ``local dependency" condition which is challenging to verify analytically, although it can be diagnosed using the data at hand. It might be interesting to investigate (i) precisely when this condition holds and (ii) if MLR statistics are still optimal when it fails.
    
    \item We only consider classes of MLR statistics designed for binary GLMs and generalized additive models. However, other types of MLR statistics could be more powerful, e.g., those based on Bayesian additive regression trees \citep{bart2010}.

    \item In practice, analysts may prefer to discover features with large effect sizes. E.g., in Section \ref{subsec::hiv_motiv}, the P90.M variant has a large estimated OLS coefficient; thus, while it is particularly hard to discover, it may be particularly valuable to discover. In principle, the Bayesian framework in Section \ref{subsec::amlr} could be used to find knockoff statistics which asymptotically maximize many different notions of power, e.g., the sum of squared coefficient sizes across all discovered variables.

\end{itemize}

\section{Acknowledgements}

The authors thank John Cherian, Kevin Guo, Lucas Janson, Lihua Lei, Basil Saeed, Anav Sood, and Timothy Sudijono for valuable comments. A.S. is partially supported by a Citadel GQS PhD Fellowship, the Two Sigma Graduate Fellowship Fund, and an NSF Graduate Research Fellowship. W.F. is partially supported by the NSF DMS-1916220 and a Hellman Fellowship from Berkeley.

\bibliography{references}
\bibliographystyle{apalike}

\appendix

\newpage

\section{An illustration of the importance of the order of $|W|$}\label{appendix::wstat_synth}

Section \ref{subsec::intuition} argues intuitively that a good knockoff statistic should roughly achieve the following goals:
\begin{enumerate}
    \item For each $j$, it should maximize $P\opt(W_j > 0)$.
    \item The order of $\{|W_j|\}_{j=1}^p$ should match the order of $\{P\opt(W_j > 0)\}_{j=1}^p$---i.e., $|W_j|$ should be an increasing function of $P\opt(W_j > 0)$. 
\end{enumerate}
Sections \ref{sec::mlr} formalizes (and slightly modifies) these goals to develop an asymptotically optimal test statistic. However, to build intuition, we now give a concrete (if contrived) example showing the importance of the second goal.

Consider a setting with $p = 50$ features where $25$ features have a large signal size with $P\opt(W_j > 0) = 99.9\%$, $10$ features have a moderate signal size with $P\opt(W_j > 0) = 75\%$, the last $15$ features are null with $P\opt(W_j > 0) = 50\%$, and $\{\sign(W_j)\}_{j=1}^p$ are independent. In this case, what absolute values should $W$ take to maximize power? 

To make any discoveries and control the FDR at level $q = 0.05$, we must ensure that $> 95\%$ of the $k$ feature statistics with the largest absolute values have positive signs (for some $k \ge 20$ due to the ceiling function in Step 2 in Section \ref{subsec::intuition}). Since only the features with large signal sizes have a $>95\%$ chance of being positive, making any discoveries is extremely unlikely unless the features with large signal sizes generally have the highest absolute values. Figure \ref{fig::wstatplotsynth} illustrates this argument---in the ``random prioritization" setting, we sample $|W_j| \iid \Unif(0,1)$, and in the ``oracle prioritization" setting, we set $|W_j| = P\opt(W_j > 0)$. As expected, knockoffs makes zero discoveries with random prioritization and $\approx 30$ discoveries with oracle prioritization.

\begin{figure}[!hb]
    \centering
    \includegraphics[width=\linewidth]{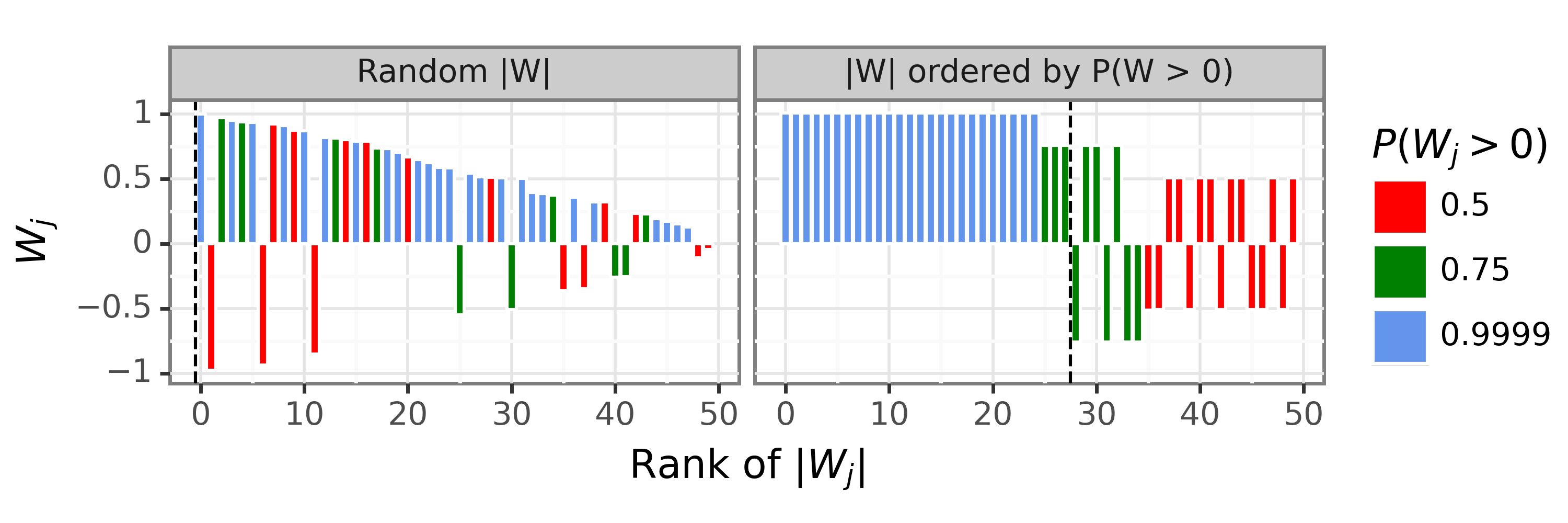}
    \caption{In a simple (if contrived) example from Appendix \ref{appendix::wstat_synth}, this figure illustrates the importance of ensuring that $\{|W_j|\}_{j=1}^p$ has roughly the same order as $\{P\opt(W_j > 0)\}_{j=1}^p$. The dotted black line shows the discovery threshold.}
    \label{fig::wstatplotsynth}
\end{figure}

\section{Main proofs and interpretation}\label{appendix::proofs}

In this section, we prove the main results of the paper. We also offer additional discussion of these results.

\subsection{Knockoffs as inference on masked data}

 In this section, we prove Propositions \ref{prop::mxdistinguish} and \ref{prop::fxdistinguish}, Lemma \ref{lem::mlrvalid}, and one more related corollary which will be useful when proving Theorem \ref{thm::avgopt}. As notation, for any matrices $M_1, M_2 \in \R^{n \times p}$, let $[M_1, M_2]_{\swap{j}}$ denote the matrix $[M_1, M_2]$ but with the $j$th column of $M_1$ and $M_2$ swapped: similarly, $[M_1, M_2]_{\swap{J}}$ swaps all columns $j \in J$ of $M_1$ and $M_2$.
 
\begingroup
\def\theproposition{\ref{prop::mxdistinguish}}
\begin{proposition} Let $\tbX$ be model-X knockoffs such that $\bX_j \ne \tbX_j$ a.s. for $j \in [p]$. Then $W = w([\bX, \tbX], \bigY)$ is a valid feature statistic if and only if:
\begin{enumerate}[topsep=0pt, leftmargin=*]
    \setlength{\parskip}{0pt}
    \setlength{\itemsep}{0pt plus 1pt}
    \item $|W|$ is a function of the masked data $D$.
    \item For all $j \in [p]$, there exists a $D$-measurable random vector $\widehat{\bX}_j$ such that $W_j > 0$ if and only if $\widehat{\bX}_j = \bX_j$.
\end{enumerate}
\begin{proof} \underline{Forward direction}: Suppose $W$ is a valid feature statistic; we will now show conditions (i) and (ii). To show (i), note that observing $\{\bX_j, \tbX_j\}_{j=1}^p$ is equivalent to observing $[\bX, \tbX]_{\swap{J}}$ for some unobserved $J \subset [p]$ chosen uniformly at random. Define $[\bX^{(1)}, \bX^{(2)}] \defeq [\bX, \tbX]_{\swap{J}}$ and let $W' = w([\bX^{(1)}, \bX^{(2)}], \bigY)$. Then by the swap invariance property of knockoffs, we have that $|W| = |W'|$. Since $|W'|$ is a function of $D$, this implies $|W|$ is a function of $D$ as well, which proves (i).

To show (ii), we construct $\widehat{\bX}_j$ as follows. Let $O_j \in \R^n$ be any ``other" random vector chosen such that $O_j \not \in \{\bX_j^{(1)}, \bX_j^{(2)}\}$. Then define
\begin{equation*}
    \widehat{\bX}_j \defeq \begin{cases}
        \bX_j^{(1)} & W_j' > 0 \\
        \bX_j^{(2)} & W_j' < 0 \\
        O_j & W_j' = 0.
    \end{cases}
\end{equation*}
Intuitively, we set $\widehat{\bX}_j = O_j$ if and only if $W_j = 0 \Leftrightarrow W_j' = 0$, since this will guarantee that $\widehat{\bX}_j \ne \bX_j$ whenever $W_j = 0$.

Note that $\widehat{\bX}_j$ is a function of $[\bX^{(1)}, \bX^{(2)}], \bigY$ and therefore is $D$-measurable. To show $\widehat{\bX}_j$ it is well-defined (does not depend on $J$), note that  $\widehat{\bX}_j \in \{\bX_j, \tbX_j, O_j\}$ can only take one of three values conditional on $D$. Thus, it suffices to show that the events $\widehat{\bX}_j = O_j$ and $\widehat{\bX}_j = \bX_j$ do not depend on the random set $J$. 

To show that the event $\widehat{\bX}_j = O_j$ does not depend on $J$, recall $\widehat{\bX}_j = O_j$ iff $|W_j'| = 0$; since $|W_j'| = |W_j|$, this event does not depend on $J$.

To show that the event $\widehat{\bX}_j = \bX_j$ does not depend on $J$, it suffices to show $\widehat{\bX}_j = \bX_j$ if and only if $W_j > 0$, which also shows (ii). There are two cases. In the first case, if $j \not \in J$, then $\bX_j^{(1)} = \bX_j$ by definition of $\bX^{(1)}$ and also $W_j' = W_j$ by the ``flip-sign" property of $w$. Thus $\widehat{\bX}_j = \bX_j^{(1)} = \bX_j$ if and only if $W_j > 0$. The second case is analogous: if $j \in J$, then $W_j' = - W_j$, so $\widehat{\bX}_j = \bX_j^{(2)} = \bX_j$ if and only if $W_j' < 0 \Leftrightarrow W_j > 0$. In both cases,  $W_j > 0$ if and only if $\widehat{\bX}_j = \bX_j$, proving (ii).

\underline{Backwards direction}: To show $W = w([\bX, \tbX], \bigY)$ is a valid feature statistic, it suffices to show the flip-sign property, namely that $W' \defeq w([\bX, \tbX]_{\swap{J}}, \bigY) = -1_{J} \odot W$, where $\odot$ denotes elementwise multiplication and $-1_{J}$ is the vector of all ones but with negative ones at the indices in $J$. To do this, note that $D$ is invariant to swaps of $\bX$ and $\tbX$, so $|W| = |W'|$ because by assumption $|W|, |W'|$ are a function of $D$. Furthermore, for any $j \in [p]$, we have that $W_j > 0$ if and only if $\widehat{\bX}_j = \bX_j$; however, since $\widehat{\bX}_j$ is also a function of $D$, we have that $\sign(W_j) = \sign(W'_j)$ if and only if $j \not \in J$. This completes the proof.
\end{proof}
\end{proposition}
\addtocounter{proposition}{-1}
\endgroup

The proof of Proposition \ref{prop::fxdistinguish} is identical to the proof of Proposition \ref{prop::mxdistinguish}, so we omit it for brevity. 

We now prove Lemma \ref{lem::mlrvalid}.

\begingroup
\def\thetheorem{\ref{lem::mlrvalid}}
\begin{lemma}
Equations (\ref{eq::mlr_def}) and (\ref{eq::fx_mlr}) define valid MX and FX knockoff statistics, respectively.
\begin{proof} For the MX case, we will show that for any $J \subset [p]$, $\mlr\uppi([\bX, \tbX]_{\swap{J}}, \bigY) = -1_J \odot \mlr\uppi([\bX, \tbX], \bigY)$, where $\odot$ denotes elementwise multiplication and $-1_J \in \R^p$ is the vector of ones but with negative ones at the indices in $J$. To show this, note that the masked data $D$ is invariant to swaps. Therefore, applying Eq. \ref{eq::mlr_def} yields
\begin{equation}
    \mlr_j\uppi([\bX, \tbX]_{\swap{J}}, \bigY) = \begin{cases}
    \log\left(\frac{P_j\bayes(\bX_j \mid D)}{P_j\bayes(\tbX_j \mid D)} \right) & j \not \in J \\
    \log\left(\frac{P_j\bayes(\tbX_j \mid D)}{P_j\bayes(\bX_j \mid D)} \right) & j \in J.
    \end{cases}
\end{equation}
Since $\log(x/y) = \log(x) - \log(y)$ is an antisymmetric function, this proves that
\begin{equation}
    \mlr_j\uppi([\bX, \tbX]_{\swap{J}}, \bigY) = \begin{cases}
    \mlr_j\uppi([\bX, \tbX], \bigY) & j \not \in J \\
    - \mlr_j\uppi([\bX, \tbX], \bigY) & j \in J.
    \end{cases}
\end{equation}
which completes the proof. The proof for the FX case is analogous.
\end{proof}
\end{lemma}
\endgroup

Finally, the following corollary of Propositions \ref{prop::mxdistinguish} and \ref{prop::fxdistinguish} will be important when proving Theorem \ref{thm::avgopt}. 

\begin{corollary}\label{cor::determdiff} Let $W, W'$ be two knockoff feature statistics. Then in the same setting as Propositions \ref{prop::mxdistinguish} and \ref{prop::fxdistinguish}, the event $\sign(W_j) = \sign(W_j')$ is a deterministic function of the masked data $D$.
\begin{proof} We give the proof for the model-X case, and the fixed-X case is analogous. First, note that the events $W_j = 0$ and $W_j' = 0$ are $D$-measurable events since $|W_j|, |W_j'|$ are $D$-measurable by Proposition \ref{prop::mxdistinguish}. Therefore, the only non-trivial case is the case where $W_j, W_j' \ne 0$, which we now consider.

By Proposition (\ref{prop::mxdistinguish}), there exist $D$-measurable vectors $\widehat{\bX}_j, \widehat{\bX}_j'$ such that $W_j > 0 \Leftrightarrow \widehat{\bX}_j = \bX_j$ and $W_j' > 0 \Leftrightarrow \widehat{\bX}_j' = \bX_j$. Since $\widehat{\bX}_j, \widehat{\bX}_j'$ must take one of exactly two distinct values, this  implies that
\begin{equation*}
    \sign(W_j) = \sign(W_j') \Leftrightarrow \widehat{\bX}_j = \widehat{\bX}_j'
\end{equation*}
where the right-most expression is $D$-measurable since $\widehat{\bX}_j, \widehat{\bX}_j'$ are $D$-measurable. This completes the proof.
\end{proof}
\end{corollary}

\subsection{Proof of Proposition \ref{prop::intuit_opt}}

\begingroup
\def\theproposition {\ref{prop::intuit_opt}}
\begin{proposition} Given data $\bigY, \bX$ and knockoffs $\tbX$, let $\MLR\uppi$ be the MLR statistics with respect to some Bayesian model $P\bayes$. Let $W$ be any other valid knockoff feature statistic. Then,
\begin{equation}\tag{\ref{eq::bestsigns}}
    P\bayes(\MLR_j\uppi > 0 \mid D) \ge P\bayes(W_j > 0 \mid D).
\end{equation}
Furthermore, $\{|\MLR_j\uppi|\}_{j=1}^p$ has the same order as $\{P\bayes(\MLR_j\uppi > 0 \mid D)\}_{j=1}^p$. More precisely,
\begin{equation}\tag{\ref{eq::bestmags}}
    P\bayes(\MLR_j\uppi > 0 \mid D) = \frac{\exp(|\MLR_j\uppi|)}{1 + \exp(|\MLR_j\uppi|)}.
\end{equation}
\begin{proof} First, we prove Eq. (\ref{eq::bestsigns}). Let $\widehat{\bX}_j\opt = \argmax_{\bx \in \{\bX_j, \tbX_j\}} P_j\bayes(\bx \mid D)$ be the ``best guess" of the value of $\bX_j$ based on $D$, and observe that by definition $\MLR_j\uppi \defeq \log(P_j\uppi(\bX_j \mid D)) - \log(P_j\uppi(\tbX_j \mid D)) > 0$ if and only if $\widehat{\bX}_j\opt = \bX_j$. Similarly, Proposition \ref{prop::mxdistinguish} proves that there exists some alternative $D$-measurable random variable $\widehat{\bX}_j$ such that $W_j > 0$ if and only if $\widehat{\bX}_j > 0$. However, we note that by definition of $\widehat{\bX}_j\opt$,
\begin{equation}
    P\bayes(\MLR_j\uppi > 0 \mid D) = P\bayes(\widehat{\bX}_j\opt = \bX_j \mid D) \ge P\bayes(\widehat{\bX}_j = \bX_j \mid D) = P\bayes(W_j > 0 \mid D),
\end{equation}
which completes the proof of Eq. (\ref{eq::bestsigns}).

To prove Eq. (\ref{eq::bestmags}), observe $P_j\bayes(\bX_j \mid D) = 1 - P_j\bayes(\tbX_j \mid D)$, since conditional on $D$ we observe the set $\{\bX_j, \tbX_j\}$. Therefore,
\begin{equation*}
    |\MLR_j\uppi| = \log\left(\frac{\max_{\bx \in \{\bX_j, \tbX_j\}} P_j\bayes(\bx \mid D)}{1 - \max_{\bx \in \{\bX_j, \tbX_j\}} P_j\bayes(\bx \mid D)}\right) = \log\left(\frac{P\bayes(\MLR_j\uppi > 0 \mid D)}{1-P\bayes(\MLR_j\uppi > 0 \mid D)}\right),
\end{equation*}
where the second step uses the fact that $P\bayes(\MLR_j\uppi > 0 \mid D) = P\bayes\left(\widehat{\bX}_j\opt = \bX_j \mid D \right) = \max_{\bx \in \{\bX_j, \tbX_j\}} P_j\bayes(\bx \mid D)$ for $\widehat{\bX}_j\opt$ as defined above. This completes the proof for model-X knockoffs; the proof in the fixed-X case is analogous and just replaces $\{\bX_j, \tbX_j\}$ with $\{\bX_j^T \bigY, \tbX_j^T \bigY\}$.
\end{proof}
\end{proposition}
\addtocounter{proposition}{-1}
\endgroup

\subsection{How far from optimality are MLR statistics in finite samples?}\label{subsec::finitesampleopt}

Our main result (Theorem \ref{thm::avgopt}) shows that MLR statistics \textit{asymptotically} maximize the number of discoveries made by knockoffs under $P\bayes$. However, before rigorously proving Theorem \ref{thm::avgopt}, we give intuition suggesting that even in finite samples, MLR statistics are probably nearly optimal anyway.

Recall from Section \ref{subsec::mlr} that in finite samples, MLR statistics (i) maximize $P\bayes(\MLR_j\uppi > 0 \mid D)$ for each $j \in [p]$ and (ii) ensure that the absolute values of the feature statistics $\{|\MLR_j\uppi|\}_{j=1}^p$ have the same order as the probabilities $\{P\bayes(\MLR_j\uppi > 0)\}_{j=1}^p$. As per Proposition \ref{prop::exactopt}, this strategy is exactly optimal when the vector of signs $\sign(\MLR\uppi)$ are conditionally independent given $D$, but in general, it is possible to exploit conditional dependencies among the coordinates of $\sign(\MLR\uppi)$ to slightly improve power. However, we argue below that it is challenging to even slightly improve power without unrealistically strong dependencies.

To see why this is the case, consider a very simple setting with $p=6$ features and FDR level $q=0.2$, so knockoffs will make discoveries exactly when the first five $W$-statistics with the largest absolute values have positive signs. Suppose that $W_1, \dots, W_5 \mid D$ are perfectly correlated and satisfy $P\bayes(W_1 > 0 \mid D) = \dots = P\bayes(W_5 > 0 \mid D) = 70\%$, and $W_6 \Perp W_{1:5} \mid D$ satisfies $P\bayes(W_6 > 0 \mid D) = 90\%$. Since $W_6$ has the highest chance of being positive, MLR statistics will assign it the highest absolute value, in which case knockoffs will make discoveries with probability $70\% \cdot 90\% = 63\%$. However, in this example, knockoffs will be more powerful if we ensure that $W_1, \dots, W_5$ have the five largest absolute values, since their signs are perfectly correlated and thus $P\bayes(W_1 > 0, \dots, W_5 > 0 \mid D) = 70\% > 63\%$.\footnote{Note that there is nothing special about positive correlations in this example: one can also find similar examples where negative correlations among $\MLR\uppi$ can be prioritized to slightly increase power.} 

This example has two properties which shed light on the more general situation.
First, to even get a slight improvement in power required extremely strong dependencies among the coordinates of $\sign(\MLR\uppi)$, which is not realistic. Indeed, empirically in Figure \ref{fig::dependence}, the coordinates of $\sign(\MLR\uppi)$ appear to be almost completely conditionally uncorrelated even when $\bX$ is extremely highly correlated. Thus, although it may be possible to slightly improve power by exploiting dependencies among $\sign(\MLR\uppi)$, the \textit{magnitude} of the improvement in power is likely to be small. Second, the reason that it is possible to exploit dependencies to improve power in this case is because knockoffs has a ``hard" threshold where one can only make any discoveries if one makes at least 5 discoveries, and exploiting conditional correlations among the vector $\sign(\MLR\uppi)$ can slightly improve the probability that we reach that initial threshold. However, this ``threshold" phenomenon is less important in situations where knockoffs are guaranteed to make at least a few discoveries; thus, if the number of discoveries grows with $p$, this effect should be insignificant asymptotically.

\subsection{Proof of Theorem \ref{thm::avgopt}}\label{subsec::mainproof}

\underline{Notation}: For any vector $x \in \R^n$ and $k \le n$, we let $\bar x_k \defeq \frac{1}{k} \sum_{i=1}^k x_i$ be the sample mean of the first $k$ elements of $x$. For $k > n$, we let $\bar x_k \defeq \frac{n}{k} \bar x$ equal the sample mean of the vector $x$ plus $k - n$ additional zeros. Additionally, for $x \in \R^n$ and a permutation $\kappa : [n] \to [n]$, $\kappa(x)$ denotes the coordinates of $x$ permuted according to $\kappa$, so that $\kappa(x)_i = x_{\kappa(i)}$. Throughout this section, all probabilities $\P$ and expectations $\E$ are taken under $P\bayes$.

\underline{Main idea}: There are two main ideas behind Theorem \ref{thm::avgopt}. First, for any feature statistic $W$, we will compare the power of $W$ to the power of a ``soft" version of the SeqStep procedure, which depends only on the conditional expectation of $\sign(W)$ instead of the realized values of $\sign(W)$. Roughly speaking, if the coordinates of $\sign(W)$ obey a strong law of large numbers, the power of SeqStep and the power of the ``soft" version of SeqStep will be the same asymptotically. Second, we will show that MLR statistics $\MLR\uppi$ exactly maximize the power of the ``soft" version of the SeqStep procedure. Taken together, these two results imply that MLR statistics are asymptotically optimal.

To make this precise, for a feature statistic $W$, let $\sorted(W)$ denote $W$ sorted in decreasing order of its absolute values, and let $R = \I(\sorted(W) > 0) \in \{0,1\}^p$ be the vector indicating where $\sorted(W)$ has positive entries. The number of discoveries made by knockoffs only depends on $R$. Indeed, for any vector $\eta \in [0,1]^p$ and any desired FDR level $q \in (0,1)$, define
\begin{equation}\label{eq::psidef}
    \psi_q(\eta) \defeq \max_{k \in \N} \left\{k : \frac{k-k \bar \eta_k + 1}{k \bar \eta_k} \le q\right\} \text{ and } \tau_q(\eta) = \left\lceil \frac{\psi_q(\eta)+1}{1+q}\right\rceil,
\end{equation}
where by convention we set $\frac{x}{0} = \infty$ for any $x \in \R_{> 0}$ and we remind the reader that for $k > p$, $\bar \eta_k = \frac{p}{k} \bar \eta$. It turns out that knockoffs makes exactly $\tau_q(R)$ discoveries. For brevity, we refer the reader to Lemma B.3 of \cite{mrcknock2022} for a formal proof of this: however, to see this intuitively, note that $k - k \bar R_k$ (resp. $k \bar R_k$) counts the number of negative (resp. positive) entries in the first $k$ coordinates of $\sorted(W)$, so this definition lines up with the definition of the data-dependent threshold in Section \ref{subsec::knockreview}.

Now, let $\delta \defeq \E[R \mid D] \in [0,1]^p$ be the conditional expectation of $R$ given the masked data $D$ (defined in Equation \ref{eq::maskeddata}). The ``soft" version of SeqStep simply applies the functions $\psi_q$ and $\tau_q$ to the conditional expectation $\delta$ instead of the realized indicators $R$. Intuitively speaking, our goal will be to apply a law of large numbers to show the following asymptotic result:
\begin{equation*}
    |\tau_q(\delta) - \tau_q(R)| = o_p\left(\# \text{ of non-nulls}\right).
\end{equation*}
Once we have shown this, it will be straightforward to show that MLR statistics are asymptotically optimal, since MLR statistics maximize $\tau_q(\delta)$ in finite samples. 

We now begin to prove Theorem \ref{thm::avgopt} in earnest. In particular, the following pair of lemmas tells us that if $\bar R_k$ converges to $\bar \delta_k$ uniformly in $k$, then $\tau_q(\delta) \approx \tau_q(R)$.

\begin{lemma}\label{lem::ddt2slln} Let $W = w([\bX, \tbX], \bigY)$ be any feature statistic with $R, \delta, \psi_q, \tau_q$ as defined earlier. Fix any $k_0 \in [p]$ and sufficiently small $\epsilon > 0$ such that $\eta \defeq 3(1+q)\epsilon < q$. Define the event
\begin{equation*}
    A_{k_0,\epsilon} = \left\{\max_{k_0 \le k \le p} |\bar R_k - \bar \delta_k| \le \epsilon\right\}.
\end{equation*}
Then on the event $A_{k_0,\epsilon}$, we have that 
\begin{equation}\label{eq::bound_onA}
    \frac{1}{1+3\epsilon} \tau_{q-\eta}(R) - k_0 - 1 \le \tau_q(\delta) \le (1+3\epsilon) \tau_{q+\eta}(R) + k_0 + 1.
\end{equation}
This implies that
\begin{align}
        |\tau_q(R) - \tau_q(\delta)| 
    &\le
        p \I(A_{k_0,\epsilon}^c)  + \big[\tau_{q+\eta}(R) - \tau_{q-\eta}(R)\big] + k_0 + 1 + 3 \epsilon \tau_{q+\eta}(R). \label{eq::ddt2slln}
\end{align}
\begin{proof} Note the proof is entirely algebraic (there is no probabilistic content). We proceed in two steps, first showing Equation (\ref{eq::bound_onA}), then Equation (\ref{eq::ddt2slln}).

\underline{Step 1}: We now prove Equation (\ref{eq::bound_onA}). To start, define the sets
\begin{equation*}
    \mcR = \left\{k \in [p] : \frac{k - k \bar R_k + 1}{k \bar R_k} \le q + \eta \right\} \text{ and } \mcD = \left\{k \in [p] : \frac{k - k \bar \delta_k + 1}{k \bar \delta_k} \le q \right\}
\end{equation*}
and recall that by definition $\psi_{q + \eta}(R) = \max(\mcR)$, $\psi_{q}(\delta) = \max(\mcD)$. To analyze the difference between these quantities, fix any $k \in \mcD \setminus \mcR$. Then by definition of $\mcD$ and $\mcR$, we know
\begin{align*}
    \frac{k - k \bar \delta_k + 1}{k \bar \delta_k} \le q < q + \eta < \frac{k - k \bar R_k + 1}{k \bar R_k}.
 \end{align*}
However, Lemma \ref{lem::alg} (proved in a moment) tells us that this implies the following algebraic identity:
\begin{equation*}
    \bar \delta_k - \bar R_k \ge \frac{\eta}{3(1+q)} = \frac{3(1+q)\epsilon}{3(1+q)} = \epsilon.
\end{equation*}
However, on the event $A_{k_0, \epsilon}$ this cannot occur for any $k \ge k_0$. Therefore, on the event $A_{k_0, \epsilon}$, $\mcD \setminus \mcR \subset \{1, \dots, k_0-1\}$. This implies that
\begin{equation}\label{eq::psidif1}
    \psi_q(\delta) - \psi_{q+\eta}(R) = \max(\mcD) - \max(\mcR) \le \begin{cases}
        0 & \max(\mcD) \ge k_0 \\
        k_0 - 1 & \max(\mcD) < k_0.
    \end{cases}
\end{equation}
We can combine these conditions by writing that $\psi_q(\delta) - \psi_{q+\eta}(R) \le k_0 - 1$. Using the definition of $\tau_q(\cdot)$, we conclude
\begin{align*}
        \tau_q(\delta) - \tau_{q+\eta}(R) 
    &=
        \ceil{\frac{\psi_q(R)+1}{1+q}} - \ceil{\frac{\psi_{q+\eta}(\delta)+1}{1+q+\eta}} \\
    &\le
        1 + \frac{\psi_q(R)+1}{1+q} - \frac{\psi_{q+\eta}(\delta)+1}{1+q+\eta} \\
    &\le
        2 + \frac{\psi_q(R) - \psi_{q+\eta}(\delta)}{1+q} + \left(\frac{1}{1+q} - \frac{1}{1+q+\eta}\right) \psi_{q+\eta}(R) \\
    &=
        2 + \frac{\psi_q(R) - \psi_{q+\eta}(\delta)}{1+q} +  \frac{3 \epsilon}{1+q+\eta} \psi_{q+\eta}(R)& \text{ by def. of $\eta$}\\
    &\le
        2 + \frac{k_0 - 1}{1+q} + \frac{3 \epsilon}{1+q+\eta} \psi_{q+\eta}(R)&\text{ by Eq. (\ref{eq::psidif1})}\\
    &\le
        k_0 + 1 + 3 \epsilon \tau_{q+\eta}(R) & \text{ by def. of $\tau_{q}(R)$.}
\end{align*}
This proves the upper bound, namely that $\tau_q(\delta) \le (1+3\epsilon) \tau_{q+\eta}(R) + k_0 + 1$. To prove the lower bound, note that we can swap the role of $R$ and $\delta$ and apply the upper bound to $q' = q - \eta$. %(Note that it is possible that $q' < 0$, which is a slightly counter-intuitive scenario, but all of the algebra still goes through in this scenario.) 
Then if we take $\eta' = 3(1+q') \epsilon < \eta < 1$, applying the upper bound yields
\begin{equation*}
    \tau_{q'}(R) \le (1+3\epsilon) \tau_{q'+\eta'}(\delta) + k_0 + 1.
\end{equation*}
Observe that $\tau_q(\cdot)$ is nondecreasing in $q$. Furthermore, since $\eta' < \eta$, we have that $q' + \eta' = q - \eta + \eta' \le q$. Therefore, $\tau_{q' + \eta'}(\delta) \le \tau_q(\delta)$. Applying this result, we conclude
\begin{equation*}
    \tau_{q-\eta}(R) = \tau_{q'}(R) \le (1+3\epsilon) \tau_{q'+\eta'}(\delta) + k_0 + 1 \le (1+3\epsilon) \tau_{q}(\delta) + k_0 + 1.
\end{equation*}
This implies the lower bound $\frac{1}{1+3\epsilon} \tau_{q-\eta}(R) - k_0 -1 \le \tau_q(\delta)$.

\underline{Step 2}: Now, we show Equation (\ref{eq::ddt2slln}) follows from Equation (\ref{eq::bound_onA}). To see this, we consider the two cases where $\tau_q(\delta) \ge \tau_q(R)$ and vice versa and apply Equation (\ref{eq::bound_onA}). In particular, on the event $A_{k_0, \epsilon}$, then:
\begin{align*}
    |\tau_q(\delta) - \tau_q(R)|
    &=
    \begin{cases}
    \tau_q(\delta) - \tau_q(R) & \tau_q(\delta) \ge \tau_q(R) \\ 
    \tau_q(R) - \tau_q(\delta) & \tau_q(\delta) \le \tau_q(R)
    \end{cases} \\
    &\le 
    \begin{cases}
    (1+3\epsilon) \tau_{q+\eta}(R) + k_0 + 1 - \tau_q(R) & \tau_q(\delta) \ge \tau_q(R) \\
    \tau_q(R) - \frac{1}{1+3\epsilon} \tau_{q-\eta}(R) + k_0 + 1 & \tau_q(\delta) \le \tau_q(R) 
    \end{cases} & \text{ by Eq. (\ref{eq::bound_onA})}\\
    &=
    k_0 + 1 + \begin{cases}
     \tau_{q+\eta}(R) - \tau_q(R) - 3\epsilon \tau_{q+\eta}(R) & \tau_q(\delta) \ge \tau_q(R) \\
    \tau_q(R) - \tau_{q-\eta}(R) + \frac{3\epsilon}{1+3\epsilon} \tau_{q-\eta}(R)& \tau_q(\delta) \le \tau_q(R) 
    \end{cases} \\
    &\le k_0 + 1 + \tau_{q+\eta}(R) - \tau_{q-\eta}(R) + 3 \epsilon \tau_{q+\eta}(R),
\end{align*}
where the last line follows because $\tau_q(R)$ is monotone in $q$. This implies Equation (\ref{eq::ddt2slln}), because $|\tau_q(\delta) - \tau_q(R)| \le p$ trivially on the event $A_{k_0, \epsilon}^c$ because $\tau_q(\delta), \tau_q(R) \in [p]$.
\end{proof}
\end{lemma}

The following lemma proves a very simple algebraic identity used in the proof of Lemma \ref{lem::ddt2slln}.

\begin{lemma}\label{lem::alg} For any $x, y \in [0,1]$, $k \in \N,$ and any $\gamma \in (0,1)$, suppose that $\frac{1+k-kx}{kx} \le q < q + \gamma \le \frac{1+k-ky}{ky}$. Then
\begin{equation*}
    x - y \ge \frac{\gamma}{(1+q)(1+q+\gamma)} \ge \frac{\gamma}{3(1+q)}. 
\end{equation*}
\begin{proof} By assumption, $x \ne 0$, since otherwise $\frac{1+k-kx}{kx} = \infty$ by convention. For $x > 0$, we have that
\begin{equation}\label{eq::algxeq}
    \frac{1+k-kx}{kx} \le q \implies 1 + k - kx \le kqx \implies x \ge \frac{k+1}{k(1+q)}.
\end{equation}
Now, there are two cases. If $y = 0$, the inequality holds trivially:
\begin{equation*}
    x - y = x \ge \frac{k+1}{k} \cdot \frac{1}{1+q} \ge \frac{\gamma}{3(1+q)}.
\end{equation*}
Alternatively, if $y > 0$, we observe similarly to before that
\begin{equation}\label{eq::algyeq}
    \frac{1+k-ky}{ky} \ge q+\gamma \implies y \le \frac{k+1}{k(1+q+\gamma)}.
\end{equation}
Combining Equations (\ref{eq::algxeq})--(\ref{eq::algyeq}) yields the result:
\begin{align*}
        x - y 
    &\ge 
        \frac{k+1}{k} \left(\frac{1}{1+q} - \frac{1}{1+q+\gamma}\right) 
    =
        \frac{k+1}{k} \frac{\gamma}{(1+q+\gamma)(1+q)} 
    \ge
        \frac{\gamma}{3(1+q)}.
\end{align*}
\end{proof}
\end{lemma}

We are now ready to prove Theorem \ref{thm::avgopt}. As a reminder, we consider an asymptotic regime with data $\bX^{(n)} \in \R^{n \times p_n}, \bigY^{(n)} \in \R^{n}$ and knockoffs $\tbX^{(n)}$, where $P\bayes_n$ is the Bayesian model. We let $D^{(n)}$ denote the masked data for knockoffs as defined in Section \ref{subsec::maskeddata} and let $\numnonnull$ denote the expected number of non-nulls under $P\bayes_n$. We will analyze the limiting normalized number of discoveries of feature statistics $W\upn = w_n([\bX\upn, \tbX\upn], \bigY)$ with rejection set $S\upn(q)$, defined as the expected number of discoveries divided by the expected number of non-nulls:
\begin{equation}
    \Power_q(w_n) = \frac{\E_{P\bayes_n}[|S\upn(q)|]}{\numnonnull}. \tag{\ref{eq::powerdef}}
\end{equation}

For convenience, we restate Theorem \ref{thm::avgopt} and then prove it.

\begingroup
\def\thetheorem{\ref{thm::avgopt}}
\begin{theorem} 
For each $n$, let $\MLR\uppi = \mlr_n\uppi([\bX^{(n)}, \tbX^{(n)}], \bigY^{(n)})$ denote the MLR statistics with respect to $P_n\bayes$ and let $W = w_n([\bX\upn, \tbX\upn], \bigY\upn)$ denote any other sequence of feature statistics. Assume the following:
\begin{itemize}
    \item $\lim_{n \to \infty} \Power_q(\mlr_n\uppi)$ and $\lim_{n \to \infty} \Power_q(w_n)$ exist for each $q \in (0,1)$.
    \item The expected number of non-nulls grows faster than $\log(p_n)^4$. Formally, assume that for some $\gamma > 0$, $\lim_{n \to \infty} \frac{\numnonnull}{\log(p_n)^{4+\gamma}} = \infty$.
    \item Conditional on $D\upn$, the covariance between the signs of $\MLR\uppi$ decays exponentially off the diagonal. That is, there exist constants $C\ge 0, \rho \in (0,1)$ such that
\begin{equation}
    |\cov_{P\bayes_n}(\I(\MLR\uppi_i > 0), \I(\MLR\uppi_j > 0) \mid D\upn)| \le C \rho^{|i-j|}. \tag{\ref{eq::expcovdecay}}
\end{equation}
\end{itemize}

Then for all but countably many $q \in (0,1)$,
\begin{equation}
    \lim_{n \to \infty} \Power_q(\mlr_n\uppi) \ge \lim_{n \to \infty} \Power_q(w_n).
\end{equation}
\begin{proof} Note that in this proof, all expectations and probabilities are taken over $P_n\bayes$.

The proof proceeds in three main steps, but we begin by introducing some notation and outlining the overall strategy. Following Lemma \ref{lem::ddt2slln}, let $R = \I(\sorted(W) > 0)$ and $R\uppi = \I(\sorted(\MLR\uppi) > 0)$, and let $\delta = \E[W \mid D\upn]$ and $\delta\uppi = \E[\MLR\uppi \mid D\upn]$ be their conditional expectations. (Note that $W, R, \delta, \MLR\uppi, R\uppi$ and $\delta\uppi$ all change with $n$---however, we omit this dependency to lighten the notation.) As in Equation (\ref{eq::psidef}), we can write the number of discoveries made by $W$ and $\MLR\uppi$ as a function of $R\uppi$ and $R$, so:
\begin{equation*}
    \Power_q(\mlr_n\uppi) - \Power_q(w_n) = \frac{\E\left[\tau_q(R\uppi)\right]}{\numnonnull} - \frac{\E\left[\tau_q(R)\right]}{\numnonnull}.
\end{equation*}
We will show that the limit of this quantity is nonnegative, and the main idea is to make the approximations $\tau_q(R\uppi) \approx \tau_q(\delta\uppi)$ and $\tau_q(R) \approx \tau_q(\delta)$. In particular, we can decompose
\begin{align}
        \Power_q(\mlr_n\uppi) - \Power_q(w_n)
    &=
        \frac{\E\left[\tau_q(R\uppi) - \tau_q(R)\right]}{\numnonnull} \nonumber \\
    &=
        \frac{\E[\tau_q(R\uppi) - \tau_q(\delta\uppi)]}{\numnonnull} + \frac{\E[\tau_q(\delta\uppi) - \tau_q(\delta)]}{\numnonnull} + \frac{\E\left[\tau_q(\delta) - \tau_q(R)\right]}{\numnonnull} \nonumber \\
    &\ge 
        \frac{\E[\tau_q(\delta\uppi) - \tau_q(\delta)]}{\numnonnull} - \frac{\E|\tau_q(R\uppi) - \tau_q(\delta\uppi)|}{\numnonnull} - \frac{\E|\tau_q(\delta) - \tau_q(R)|}{\numnonnull}. \label{eq::mainstrategy}
\end{align}
In particular, Step 1 of the proof is to show that $\tau_q(\delta\uppi) \ge \tau_q(\delta)$ holds deterministically, for fixed $n$. This implies that the first term of Equation (\ref{eq::mainstrategy}) is nonnegative for fixed $n$. In Step 2, we show that as $n \to \infty$, the second and third terms of Equation (\ref{eq::mainstrategy}) vanish. In Step 3, we combine these results and take limits to yield the final result.

\underline{Step 1}: In this step, we show that $\tau_q(\delta\uppi) \ge \tau_q(\delta)$ holds deterministically for fixed $n$. To do this, it suffices to show that $\bar \delta\uppi_k \ge \bar \delta_k$ for each $k \in [p_n]$. To see this, recall that $\tau_q(\delta\uppi)$ and $\tau_q(\delta)$ are increasing functions of $\psi_q(\delta\uppi)$ and $\psi_q(\delta)$, as defined below:
\begin{equation}\label{eq::taudefreminder}
    \psi_q(\delta\uppi) = \max_{k \in \N} \left\{\frac{k - k \bar \delta\uppi_k + 1}{k \bar \delta\uppi_k} \le q \right\} \text{ and } \psi_q(\delta) = \max_{k \in \N} \left\{\frac{k-k\bar\delta_k + 1}{k \bar \delta_k} \le q \right\}
\end{equation}
where for $k > n$ we use the convention of ``padding" $\delta$ and $\delta\uppi$ with extra zeros, so (e.g.) $\bar \delta_k = \frac{n}{k} \bar \delta$ for $k > n$.

Since the function $\gamma \mapsto \frac{k-k\gamma+1}{k\gamma}$ is decreasing in $\gamma$, $\bar\delta\uppi_k \ge \bar \delta_k$ implies that $\frac{k - k \bar \delta\uppi_k + 1}{k \bar \delta\uppi_k} \le \frac{k - k \bar \delta_k + 1}{k \bar \delta_k}$ for each $k$, and therefore $\psi_q(\delta\uppi) \ge \psi_q(\delta)$, which implies $\tau_q(\delta\uppi) \ge \tau_q(\delta)$. Thus, it suffices to show that $\bar\delta\uppi_k \ge \bar \delta_k$ holds for each $k$. 

Intuitively, it makes sense that $\bar \delta\uppi_k \ge \bar \delta_k$ for each $k$, since $\MLR\uppi$ maximizes $\P(\MLR_j\uppi > 0 \mid D)$ coordinate-wise and chooses $\MLR\uppi$ so that $\delta\uppi$ is sorted in decreasing order. To prove this formally, we first argue that conditional on $D\upn$, $R\uppi$ is a deterministic function of $R$. Recall that according to Corollary \ref{cor::determdiff}, the event $\sign(W_j) \ne \sign(\MLR_j\uppi)$ is completely determined by the masked data $D\upn$. Furthermore, since $R\uppi$ and $R$ are random permutations of the vectors $\I(W > 0)$ and $\I(\MLR\uppi > 0)$ where the random permutations only depend on $|W|$ and $|\MLR\uppi|$, this implies there exists a random vector $\xi \in \{0,1\}^{p_n}$ and a random permutation $\sigma \in S_{p_n}$ such that $R\uppi = \xi \oplus \sigma(R)$ and $\xi, \sigma$ are deterministic conditional on $D\upn$. Note that here, $\oplus$ denotes the generalized parity function, so 
\begin{equation}\label{eq::parity}
    b \oplus x \defeq \begin{cases} x & b = 0 \\ 1 - x & b = 1 \end{cases} \text{ for } x \in [0,1], b \in \{0,1\}
\end{equation}
which guarantees that $0 \oplus 0 = 1 \oplus 1 = 0$ and $0 \oplus 1 = 1 \oplus 0 = 1$, etc.

The intuition here is that following Proposition \ref{prop::mxdistinguish}, fitting a feature statistic $W$ is equivalent to observing $D\upn$, assigning an ordering to the features, and then guessing which one of $\{\bX_j,\tbX_j\}$ is the true feature and which is a knockoff, where $W_j > 0$ if and only if this ``guess" is correct. Since these decisions are made as deterministic functions of $D\upn$, $\MLR\uppi$ can only be different than $W$ in that (i) it may make different guesses, flipping the sign of $W$ (as represented by $\xi)$, and (ii) its absolute values may be sorted in a different order (as represented by $\sigma$).

Now, since $\xi$ and $\sigma$ are deterministic functions of $D\upn$, this implies that
\begin{equation*}
    \delta\uppi_i = \E[R_i\uppi \mid D\upn] = \E[\xi_i \oplus R_{\sigma(i)} \mid D\upn] = \begin{cases}
        1 - \delta_{\sigma(i)} & \xi_i = 1\\ 
        \delta_{\sigma(i)} & \xi_i = 0.
    \end{cases}
\end{equation*}
However, by definition, $\E[R_i\uppi \mid D\upn] = \P(\sorted(\MLR\uppi)_i > 0 \mid D\upn)$, and Proposition \ref{prop::intuit_opt} implies that $\P(\MLR\uppi_i > 0 \mid D\upn) \ge 0.5$ for all $i \in [p_n]$: since the ordering of $\MLR\uppi$ is deterministic conditional on $D\upn$, this also implies $\delta_i\uppi = \P(\sorted(\MLR\uppi)_i > 0 \mid D\upn) \ge 0.5$. Therefore,  $\delta\uppi_i \ge \delta_{\sigma(i)}$ for each $i \in [p_n]$. Additionally, by construction $\MLR\uppi$ ensures that $\delta\uppi_1 \ge \delta\uppi_2 \ge \dots \ge \delta\uppi_{p_n}$. If $\delta_{(1)}, \dots, \delta_{(p_n)}$ are the order statistics of $\delta$ in decreasing order, this implies that $\delta\uppi_i \ge \delta_{(i)}$ for all $i$. Therefore,
\begin{equation*}
    \bar \delta\uppi_k = \frac{1}{k} \sum_{i=1}^k \delta\uppi_i \ge \frac{1}{k} \sum_{i=1}^k \delta_{(i)} \ge \frac{1}{k} \sum_{i=1}^n \delta_i.
\end{equation*}
By the previous analysis, this proves that $\tau_q(\delta\uppi) \ge \tau_q(\delta)$.

\underline{Step 2}: In this step, we show that $\frac{\E|\tau_q(\delta\uppi) - \tau_q(R\uppi)|}{\numnonnull} \to 0$ for all but countably many $q \in (0,1)$, as well as the analogous result for $R$ and $\delta$. We first prove the result for $R\uppi$ and $\delta\uppi$, and in particular, for any fixed $v > 0$, we will show that $\limsup_{n \to \infty} \frac{\E|\tau_q(\delta\uppi) - \tau_q(R\uppi)|}{\numnonnull} \le v$. Since we will show this for any arbitrary $v > 0$, this implies $\frac{\E|\tau_q(\delta\uppi) - \tau_q(R\uppi)|}{\numnonnull} \to 0$.

We begin by applying Lemma \ref{lem::ddt2slln}. In particular, fix any $k_n \in [p_n]$, any $\epsilon > 0$, and define
\begin{equation*}
    A_n = \left\{\max_{k_n \le k \le p_n} |\bar R_k\uppi - \bar \delta_k\uppi| \le \epsilon\right\}.
\end{equation*}
Then by Lemma \ref{lem::ddt2slln},
\begin{equation*}
    |\tau_q(R\uppi) - \tau_q(\delta\uppi)| \le p_n \I(A_n^c)  + \tau_{q+\eta}(R\uppi) - \tau_{q-\eta}(R\uppi) + k_n + 1 + 3 \epsilon \tau_{q+\eta}(R\uppi).
\end{equation*}
where $\eta = 3(1+q)\epsilon$. Therefore,
\begin{equation}
    \frac{\E|\tau_q(R\uppi) - \tau_q(\delta\uppi)|}{\numnonnull} \le \frac{p_n \P(A_n^c)}{\numnonnull} + \frac{k_n + 1}{\numnonnull} + \frac{\E[\tau_{q+\eta}(R\uppi)] - \E[\tau_{q-\eta}(R\uppi)]}{\numnonnull} + \frac{3 \epsilon \E[\tau_{q+\eta}(R\uppi)]}{\numnonnull}.
\end{equation}
We now analyze these terms in order: while doing so, we will choose a sequence $\{k_n\}$ and constant $\epsilon > 0$ which guarantee the desired result. Note that eventually, our choice of $\epsilon$ will depend on $q$, so the convergence is not necessarily uniform, but that does not pose a problem for our proof.

\textit{First term}: To start, we will first apply a finite-sample concentration result to bound $\P(A_n^c)$. In particular, we show in Corollary \ref{cor::expdecay} that if $X_1, \dots, X_n$ are mean-zero,  $[-1,1]$-valued random variables satisfying the exponential decay condition from Equation (\ref{eq::expcovdecay}), then there exists a universal constant $C' > 0$ depending only on $C$ and $\rho$ such that
\begin{equation}\label{eq::permconcv1}
   \P\left(\max_{n_0 \le i \le n} |\bar X_i| \ge t \right) \le n \exp(-C' t^{2 } n_0^{1/4}).
\end{equation}
Furthermore, Corollary \ref{cor::expdecay} shows that this result holds even if we permute $X_1, \dots, X_n$ according to some arbitrary \textit{fixed} permutation $\sigma$. Now, observe that conditional on $D\upn$, $R_j\uppi - \delta_j\uppi$ is a zero-mean, $[-1,1]$-valued random variable which is a fixed permutation of $\I(\MLR\uppi > 0)$ minus its (conditional) expectation. Since $\I(\MLR\uppi > 0)$ obeys the conditional exponential decay condition in Equation (\ref{eq::expcovdecay}), we can apply Corollary \ref{cor::expdecay} to $R_j\uppi - \delta_j\uppi$:
\begin{equation}\label{eq::permconcapp}
    \P(A_n^c \mid D\upn) \le p_n \exp(-C' \epsilon^{2 } k_n^{1/4})
\end{equation}
which implies by the tower property that $p_n \P(A_n^c) \le p_n^2 \exp(-C' \epsilon^{2 } k_n^{1/4})$. Now, suppose we take
\begin{equation*}
    k_n = \ceil{\log(p_n)^{4 + \gamma}}.
\end{equation*}
Then observe that $\epsilon$ is fixed, so as $n \to \infty$, $k_n^{1/4} \epsilon^{2 } = \Omega(\log(p_n)^{1+\gamma/4})$. Thus
\begin{equation*}
    \log(p_n \P(A_n^c)) \le 2 \log(p_n) - \Omega\left(\log(p_n)^{1 + \gamma/4}\right) \to - \infty.
\end{equation*}
Therefore, for this choice of $k_n$, we have shown the stronger result that $p_n \P(A_n^c) \to 0$. Of course, this implies $\frac{p_n \P(A_n^c)}{\numnonnull} \to 0$ as well.

\textit{Second term}: This term is easy, as we assume in the statement that $\frac{k_n}{\numnonnull} \sim \frac{\log(p_n)^{4 + \gamma}}{\numnonnull} \to 0$.

\textit{Third term}: We will now show that for all but countably many $q \in (0,1)$, for any sufficiently small $\epsilon$ (and thus for any sufficiently small $\eta$), $\limsup_{n \to \infty} \frac{\E[\tau_{q+\eta}(R\uppi)] - \E[\tau_{q-\eta}(R\uppi)]}{\numnonnull} \le v/2$ for any fixed $v > 0$.

To do this, recall by assumption that for all $q \in (0,1)$, we have that $\lim_{n \to \infty} \frac{\E[\tau_q(R\uppi)]}{\numnonnull}$ exists and converges to some (extended) real number $L(q)$. Furthermore, we show in Lemma \ref{lem::pferbound} that $L(q)$ is always finite---this is intuitively a consequence of the fact that knockoffs controls the false discovery rate, and thus the expected number of discoveries cannot exceed the number of non-nulls by more than a constant factor. Importantly, since $\tau_q(R\uppi)$ is increasing in $q$, the function $L(q)$ is increasing in $q$ for all $q \in (0,1)$: therefore, it is continuous on $(0,1)$ except on a countable set.

Supposing that $q$ is a continuity point of $L(q)$, there exists some $\beta > 0$ such that $|q-q'| \le \beta \implies |L(q) - L(q')| \le v/4$. Take $\epsilon$ to be any positive constant such that $\epsilon \le \frac{\beta}{3(1+q)}$ and thus $\eta \le \beta$. Then we conclude
\begin{align*}
        \limsup_{n \to \infty} \frac{\E[\tau_{q+\eta}(R\uppi)] - \E[\tau_{q-\eta}(R\uppi)]}{\numnonnull} 
    % &= 
    %     \limsup_{n \to \infty} \frac{\E[\tau_{q+\beta}(R\uppi)] - \E[\tau_{q-\beta}(R\uppi)]}{\numnonnull}  \\
    &= 
        L(q+\eta) - L(q-\eta) & \text{ because } \frac{\E[\tau_{q}(R\uppi)]}{\numnonnull} \to L(q) \text{ pointwise} \\
    &\le 
        \frac{v}{2}. & \text{ by continuity}
\end{align*}

\textit{Fourth term}: We now show that for all but countably many $q \in (0,1)$, for any sufficiently small $\epsilon$, $\lim_{n \to \infty} \frac{3 \epsilon \E[\tau_{q+\eta}(R\uppi)]}{\numnonnull} = 3 \epsilon L(q+\eta) \le v / 2$. However, this is simple, since Lemma \ref{lem::pferbound} tells us that $L(q)$ is finite and continuous except at countably many points. Thus, we can take $\epsilon$ sufficiently small so that $L(q+\eta) = L(q+3(1+q)\epsilon) \le L(q) + 1$, and then also take $\epsilon < \frac{v}{6( L(q) + 1)}$ so that $3 \epsilon L(q+ \eta) \le v/2$.

Combining the results for all four terms, we see the following: for each $v > 0$, there exists a sequence $\{k_n\}$ and a constant $\epsilon$ guaranteeing that 
\begin{align*}
    \limsup_{n \to \infty} \frac{\E|\tau_q(R\uppi) - \tau_q(\delta\uppi)|}{\numnonnull} \le v.
\end{align*}
Since this holds for all $v > 0$, we conclude $\lim_{n \to \infty} \frac{\E|\tau_q(R\uppi) - \tau_q(\delta\uppi)|}{\numnonnull} = 0$ as desired. 

Lastly in this step, we need to show the same result for $R$ and $\delta$ in place of $R\uppi$ and $\delta\uppi$. However, the proof for $R$ and $\delta$ is identical to the proof for $R\uppi$ and $\delta\uppi$. The one subtlety worth mentioning is that we do not directly assume the exponential decay condition in Equation (\ref{eq::expcovdecay}) for $W$. However, as we argued in Step 1, we can write $\I(W > 0) = \xi \oplus \I(\MLR\uppi > 0)$ for some random vector $\xi \in \{0,1\}^{p_n}$ which is a deterministic function of $D\upn$. As a result, we have that
\begin{align*}%\label{eq::expcovdecay}
    |\cov(\I(W_i > 0), \I(W_j > 0) \mid D\upn)| 
    = 
    |\cov(\I(\MLR\uppi_i > 0), \I(\MLR\uppi_j > 0) \mid D\upn)|
    \le C \rho^{|i-j|}.
\end{align*}
Thus, we also conclude that $\lim_{n \to \infty} \frac{\E|\tau_q(R) - \tau_q(\delta)|}{\numnonnull} = 0$.

\underline{Step 3: Finishing the proof}. Recall Equation (\ref{eq::mainstrategy}), which states that
\begin{align}
        \Power_q(\mlr_n\uppi) - \Power_q(w_n)
    &\ge 
        \frac{\E[\tau_q(\delta\uppi) - \tau_q(\delta)]}{\numnonnull} - \frac{\E|\tau_q(R\uppi) - \tau_q(\delta\uppi)|}{\numnonnull} - \frac{\E|\tau_q(\delta) - \tau_q(R)|}{\numnonnull}. \tag{\ref{eq::mainstrategy}}
\end{align}
In Step 1, we showed that $\tau_q(\delta\uppi) \ge \tau_q(\delta)$ for fixed $n$. Furthermore, in Step 2, we showed that the second two terms vanish asymptotically. As a result, we take limits and conclude
\begin{equation*}
    \liminf_{n \to \infty} \Power_q(\mlr_n\uppi) - \Power_q(w_n) \ge 0.
\end{equation*}
Furthermore, since we assume that the limits $\lim_{n\to \infty} \Power_q(\mlr_n\uppi), \lim_{n\to \infty} \Power_q(w_n)$ exist, this implies that
\begin{equation*}
    \lim_{n \to \infty} \Power_q(\mlr_n\uppi) - \lim_{n \to \infty} \Power_q(w_n) \ge 0.
\end{equation*}
This concludes the proof.
\end{proof}
\end{theorem}
\endgroup

\subsection{Relaxing the assumptions in Theorem \ref{thm::avgopt}}\label{subsec::assumptions}

In this section, we discuss a few ways to relax the assumptions in Theorem \ref{thm::avgopt}.

First, we can easily relax the assumption that the limits $L(q) \defeq \lim_{n \to \infty} \Power_q(w_n)$ and $L\opt(q) \defeq \lim_{n \to \infty} \Power_n(\mlr_n\uppi)$ exist for each $q \in (0,1)$. Indeed, the proof of Theorem \ref{thm::avgopt} only uses this assumption to argue that 
there exists a sequence $\eta_n \to 0$ such that $L(q+\eta_n) \to L(q), L(q-\eta_n) \to L(q)$ (and similarly for $L\opt(q)$). Thus, we do not need the limits $L(q)$ to exist for every $q \in (0,1)$: in contrast, the result of Theorem \ref{thm::avgopt} will hold (e.g.) for any $q$ such that $L(\cdot), L\opt(\cdot)$ are continuous at $q$. Intuitively, this means that the result in Theorem \ref{thm::avgopt} holds except at points $q$ that delineate a ``phase transition," where the power of knockoffs jumps in a discontinuous fashion as $q$ increases.

Second, it is important to note that the precise form of the local dependency condition (\ref{eq::expcovdecay}) is not crucial. Indeed, the proof of Theorem \ref{thm::avgopt} only uses this condition to show that the partial sums of $\I(\MLR\uppi > 0)$ converge to their conditional mean given $D$. To be precise, fix any permutation $\kappa : [p_n] \to [p_n]$ and let $R = \I(\kappa(\MLR\uppi) > 0)$ where $\kappa(\MLR\uppi)$ permutes $\MLR\uppi$ according to $\kappa$. Let $\delta = \E[R \mid D]$. Then the proof of Theorem \ref{thm::avgopt} will go through exactly as written if we replace Equation (\ref{eq::expcovdecay}) with the following condition:
\begin{equation}\label{eq::stronglawcondition}
    \P\left(\max_{k_n \le k \le p_n} |\bar R_k - \bar \delta_k| \mid D \right) = o(p_n^{-1})
\end{equation}
where $k_n$ is some sequence satisfying $k_n \to \infty$ and $\frac{k_n}{\numnonnull} \to 0$. 

The upshot is this: under any condition where each permutation of $\I(\MLR\uppi > 0)$ obeys a certain strong law of large numbers, we should expect Theorem \ref{thm::avgopt} to hold. Although it is unusual to require that a strong law holds for any fixed permutation of a vector, in some cases there is a ``worst-case" permutation where if Equation (\ref{eq::stronglawcondition}) holds for some choice of $\kappa$, then it holds for every choice of $\kappa$. For example, in Corollary \ref{cor::expdecay}, we show that if the exponential decay condition holds, then it suffices to show Equation (\ref{eq::stronglawcondition}) in the case where $\kappa$ is the identity permutation, since the identity permutation places the most correlated coordinates of $\MLR\uppi$ next to each other.

\subsection{Proof of Propositions \ref{prop::exactopt}-\ref{prop::gaussian_cond_ind}}

\begingroup
\def\theproposition{\ref{prop::exactopt}}
\begin{proposition}  If $\{\I(\MLR_j\uppi > 0)\}_{j=1}^p$ are conditionally independent given $D$ under $P\bayes$, then \newline $\Powervtwo\uppi(\mlr\uppi) \ge \Powervtwo\uppi(w)$ for any valid feature statistic $w$.
%the MLR statistics with respect to $P\bayes$ exactly maximize the expected number of discoveries under $P\bayes$. 
\begin{proof} Note that the proof here is essentially the same argument used in Proposition 2 of \cite{whiteout2021}, but for completeness we restate it here. Let $W = w([\bX, \tbX], \bigY)$ denote any other feature statistic. Let $S \subset [p]$ and $S\uppi \subset [p]$ denote the discovery sets based on $W$ and $\MLR\uppi$. 

It suffices to  show that $\E_{P\bayes}[|S|] \le \E_{P\bayes}[|S\uppi|]$. The argument from proof of Theorem \ref{thm::avgopt} (the beginning of Appendix \ref{subsec::mainproof}) shows that the number of discoveries $|S|$ is a monotone function of $\psi_q(\I(\sorted(W) > 0))$, where $\sorted(W)$ denotes $W$ sorted in decreasing order of its absolute values. Therefore, it suffices to show that if $R = \I(\sorted(W) > 0)$ and $R\uppi = \I(\sorted(\MLR\uppi) > 0)$, 
\begin{equation}
    \E_{P\bayes}[\psi_q(R)] \le \E_{P\bayes}[\psi_q(R\uppi)],
\end{equation}
where as defined in Eq. (\ref{eq::psidef}), $\psi_q(\eta) \defeq \max_{k} \left\{ k : \frac{k - k \bar \eta_k + 1}{k \bar \eta_k} \le q \right\}$ for any $\eta \in \{0,1\}^p$. To do this, we recall from the proof of Theorem \ref{thm::avgopt} that there exists a $D$-measurable vector $\xi \in \{0,1\}^p$ and a $D$-measurable permutation $\sigma : [p] \to [p]$ such that
\begin{equation*}
    R = \sigma(\xi \oplus R\uppi)
\end{equation*}
where for any vector $x \in \R^p$, $\sigma(x) \defeq  (x_{\sigma(1)}, \dots, x_{\sigma(p)})$ and $\oplus$ denotes the parity function (see Eq. \ref{eq::parity}). We now make a few observations:
\begin{enumerate}[itemsep=0.5pt, topsep=0pt, leftmargin=*]
    \item We assume $\{\I(\MLR_j\uppi > 0)\}_{j=1}^p$ are conditionally independent. Since the magnitudes of $\MLR\uppi$ are $D$-measurable by Proposition \ref{prop::mxdistinguish}, $R\uppi$ is equal to a $D$-measurable permutation of $\I(\MLR\uppi > 0)$. Therefore, the entries of $R\uppi$ are conditionally independent given $D$.
    \item Since $\sigma$ and $\xi$ are $D$-measurable, the entries of $R$ are also conditionally independent given $D$.    
    \item Since $\MLR_j\uppi$ maximizes $P\bayes(\MLR_j\uppi > 0 \mid D)$ among all feature statistics by Proposition \ref{prop::intuit_opt}, this implies that for any $j$, $P\bayes(R_j\uppi > 0 \mid D) \ge \frac{1}{2}$. Thus, $P\bayes(R_j > 0 \mid D) = P\bayes(\xi_j \oplus R_j\uppi > 0 \mid D) \le P\bayes(R_j\uppi > 0 \mid D)$ for all $j$. 
\end{enumerate}
Thus, we can create a coupling $\tilde{R}$ such that $\tilde{R}$ has the same marginal law as $\sigma(R\uppi)$, but $\tilde{R} \ge R$ a.s. (by the third observation above). This implies that $\frac{k-k \bar{\tilde R}_k +1}{k \bar{\tilde R}_k} \le \frac{k-k \bar R_k +1}{k \bar R_k}$ for all $k$, and therefore $\psi_q(\tilde R) \ge \psi_q(R)$. Therefore
\begin{equation*}
    \E_{P\bayes}[\psi_q(R)] \le \E[\psi_q(\tilde R)] = \E_{P\bayes}[\psi_q(\sigma(R\uppi))].
\end{equation*}
Therefore, to complete the proof, it suffices to show that $\E_{P\bayes}[\psi_q(\sigma(R\uppi))] \le \E_{P\bayes}[\psi_q(R\uppi)]$---to simplify notation, take $R = \sigma(R\uppi)$, i.e., assume $\xi = 0$ without loss of generality. Note that by Proposition \ref{prop::intuit_opt}, the entries of $\delta\uppi \defeq \E[R\uppi \mid D] \in [0,1]^p$ are arranged in decreasing order. To show the desired result, let $\delta \defeq \E[R \mid D]$ and fix any $i < j$ such that $\delta_i < \delta_j$ are ``misordered" (i.e. not in decreasing order). It is sufficient to show that $\E[\psi_q(R) \mid D] \le \E[\psi_q(R_{\swap{\{i,j\}}}) \mid D]$, since conditional on $D$, $R\uppi = \sigma^{-1}(R)$ is simply the result of iteratively swapping elements of $R$ to sort $\delta$ in decreasing order. 

To show this, for any $r_i, r_j \in \{0,1\}$, let $(R_1, \dots, r_i, \dots, r_j, \dots, R_p)$ denote the vector which replaces the $i$th and $j$th entries of $R$ with $r_i$ and $r_j$, respectively. 
Since the entries of $R \mid D$ are conditionally independent, after conditioning on $R_{-\{i,j\}}$, we can write out the relevant conditional expectations:
\begin{align*}
    &
        \E[\psi_q(R) - \psi_q(R_{\swap{\{i,j\}}}) \mid D, R_{-\{i,j\}}] \\
    =& 
        \sum_{r_i, r_j \in \{0,1\}} [\psi_q((R_1, \dots, r_i, \dots, r_j, \dots, R_p)) - \psi_q(R_1, \dots, r_i, \dots, r_j, \dots, R_p)] \delta_i^{r_i} (1-\delta_i)^{r_i} \delta_j^{r_j} (1-\delta_j)^{r_j} \\
    =&
        \left[\psi_q((R_1, \dots, 1, \dots, 0, \dots, R_p)) - \psi_q((R_1, \dots, 0, \dots, 1, \dots, R_p))\right]\left[\delta_i (1-\delta_j) - \delta_j (1 - \delta_i) \right] \\
    =& 
        \left[\psi_q((R_1, \dots, 1, \dots, 0, \dots, R_p)) - \psi_q((R_1, \dots, 0, \dots, 1, \dots, R_p))\right]\left[\delta_i - \delta_j \right] \\
    \le& 
        0.
\end{align*}
where the first equality uses conditional independence and the definition of expectation, the second equality cancels relevant terms, the third equality is simple algebra, and the final inequality uses the fact that $\delta_i < \delta_j$ by assumption but $\psi_q((R_1, \dots, 1, \dots, 0, \dots, R_p)) - \psi_q((R_1, \dots, 0, \dots, 1, \dots, R_p)) \ge 0$. In particular, to see this latter point, define $r^{(1,0)} = (R_1, \dots, 1, \dots, 0, \dots, R_p), r^{(0,1)} \defeq (R_1, \dots, 0, \dots, 1, \dots, R_p)$ and note that the partial averages obey $\bar r_k^{(1,0)} \ge \bar r_k^{(0,1)}$ for every $k \in [p]$, which implies $\psi_q(r^{(1,0)} ) \ge \psi_q(r^{(0,1)})$ by definition of $\psi_q$.

Thus, by the tower property, $\E[\psi_q(R) \mid D] \le \E[\psi_q(R_{\swap{\{i,j\}}}) \mid D]$, which completes the proof.
\end{proof}
\end{proposition}
\endgroup

\begingroup
\def\theproposition{\ref{prop::gaussian_cond_ind}}
\begin{proposition} Suppose that (i) $\bX$ are FX knockoffs or Gaussian conditional MX knockoffs \citep{condknock2019} and (ii) under $P\opt$, $\bigY \mid \bX \sim \mcN(\bX \beta, \sigma^2 I_n)$. Then under $P\opt$, $\{\I(\MLR_j\oracle > 0)\}_{j=1}^p \mid D$ are conditionally independent.
\begin{proof} This result is already proved for the fixed-X case by \cite{whiteout2021}, so we only prove it for the case where $\tbX$ are conditional Gaussian MX knockoffs. In particular, define $\widehat{\bX}_j\opt \defeq \argmax_{\bx \in \{\bX_j, \tbX_j\}} P_j\opt(\bx \mid D)$ and recall that $\MLR_j\oracle> 0$ if and only if $\widehat{\bX}_j\opt = \bX_j$. Therefore, to show that $\{\MLR_j\oracle > 0\}_{j=1}^p \mid D$ are conditionally independent under $P\opt$, it suffices to show that $\{\bX_j\}_{j=1}^p \mid D$ are conditionally independent under $P\opt$.

Fix any value $d = (\bigY^{(0)}, \{\bx_j, \tbx_j\}_{j=1}^p)$ and let $\bX^{(0)} \in \R^{n \times p}$ denote a possible value for the design matrix which is consistent with observing $D = d$ in the sense that $\bX_j^{(0)} \in \{\bx_j, \tbx_j\}$ for all $j \in [p]$. It suffices to show the factorization
\begin{equation*}
    P\opt(\bX = \bX^{(0)} \mid D = d) \propto \prod_{j=1}^p q_j(\bX_j^{(0)})
\end{equation*}
for some functions $q_1, \dots, q_p : \R^n \to \R_{\ge 0}$ which may depend on the value $d$. To do this, observe that
\begin{align*}
        P\opt(\bX = \bX^{(0)} \mid D = d) 
    &\propto
        P\opt_{\bigY \mid \bX}(\bigY^{(0)} \mid \bX^{(0)}) \cdot P\opt(\bX = \bX^{(0)} \mid \{\bX_j, \tbX_j\} = \{\bx_j, \tbx_j\}_{j=1}^p \text{ for } j = 1, \dots, p) \\
    &\propto
        \exp\left(-\frac{1}{2\sigma^2} \|\bigY^{(0)} - \bX^{(0)} \beta\|_2^2\right) \frac{1}{2^p}.
\end{align*}
where the last line uses the Gaussian linear model assumption that under $P\opt$, $\bigY \mid \bX \sim \mcN(\bX \beta, \sigma^2 I_n)$ for some fixed $\beta \in \R^p, \sigma^2 \ge 0$ as well as the pairwise exchangeability of $\{\bX_j, \tbX_j\}$. Continuing yields,
\begin{align*}
        P\opt(\bX = \bX^{(0)} \mid D = d) 
    &\propto
        \exp\left(-\frac{1}{2\sigma^2} \left(\beta^T {\bX^{(0)}}^T \bX^{(0)} \beta - 2 {\bigY^{(0)}}^T \bX^{(0)} \beta \right) \right) \\
    &\propto 
        \exp\left(\frac{{\bigY^{(0)}}^T \bX^{(0)} \beta}{\sigma^2} \right).
\end{align*}
Here, the last step uses the key assumption that $\tbX$ are conditional Gaussian MX knockoffs, in which case $\bX^T \bX = \tbX^T \tbX$ and $\bX^T \tbX = \bX^T \bX - S$ for some diagonal matrix $S \in \R^{p \times p}$. In other words, the value of $\bX^T \bX$ is $D$-measurable, and thus conditional on $D = d$, the value of $\beta^T {\bX^{(0)}}^T \bX^{(0)} \beta $ is a constant. At this point, we conclude that
\begin{align*}
        P\opt(\bX = \bX^{(0)} \mid D = d) 
    &\propto
        \prod_{j=1}^p \exp\left(\frac{{\bigY^{(0)}}^T \bX_j^{(0)} \beta_j}{\sigma^2 }\right).
\end{align*}
which completes the proof by the factorization argument above.

\end{proof}
\end{proposition}
\endgroup
\addtocounter{proposition}{-1}

\subsection{Maximizing the expected number of true discoveries}\label{appendix::tpr}

Theorem \ref{thm::avgopt} shows that MLR statistics maximize the (normalized) expected number of discoveries, but not necessarily the expected number of \textit{true} discoveries. In this section, we give a sketch of the derivation of AMLR statistics and prove that they maximize the expected number of true discoveries.

This section uses the notation introduced in Section \ref{subsec::mainproof}. All probabilities and expectations are taken over $P\bayes$. As a reminder, for any feature statistic $W$, let $R = \I(\sorted(W) > 0)$, let $\delta = \E[R \mid D]$, and let $\psi_q(\cdot)$ be as defined in Equation (\ref{eq::psidef}) so that knockoffs makes $\tau_q(R) = \ceil{\frac{\psi_q(R) + 1}{1+q}}$ discoveries.

\subsubsection{Proof of Corollary \ref{cor::amlr_valid}}

\begingroup
\def\thecorollary{\ref{cor::amlr_valid}}
\begin{corollary} AMLR statistics from Definition \ref{def::amlr} are valid knockoff statistics.
\begin{proof} The signs of the AMLR statistics are identical to the signs of the MLR statistics. Therefore, by Propositions \ref{prop::mxdistinguish} and \ref{prop::fxdistinguish} (in the MX and FX case, respectively), it suffices to show that the absolute values of the AMLR statistics are functions of the masked data. However, the AMLR statistics magnitudes are purely a function of (i) the magnitudes of the MLR statistics and (ii) $\nu_j$, which is the ratio of conditional probabilities given the masked data $D$. These conditional probabilities by definition are functions of $D$, and since MLR statistics are valid knockoff statistics by Lemma \ref{lem::mlrvalid}, the MLR magnitudes are also a function of the masked data $D$ by Proposition \ref{prop::mxdistinguish}. Thus, the AMLR statistic magnitudes are a function of $D$, which concludes the proof.
\end{proof}
\end{corollary}
\endgroup

\subsubsection{Proof sketch and intuition}\label{subsubsec::tpr_intuition} 

\textbf{Proof sketch}: The key idea behind the proof of Theorem \ref{thm::avgopt} is to observe that:
\begin{enumerate}[topsep=0pt]
    \item The number of discoveries $\tau_q(R)$ only depends on cumulative averages of $R$, denoted $\bar R_k = \frac{1}{k} \sum_{i=1}^k R_i$.
    \item As $p \to \infty$, $\bar R_k \toas \bar \delta_k$ under suitable assumptions. Thus, $\tau_q(R) \approx \tau_q(\delta)$.
    \item If $R\uppi = \I(\sorted(\MLR\uppi) > 0)$ are MLR statistics with $\delta\uppi = \E[R\uppi \mid D]$, then $R\uppi$ is asymptotically optimal because $\tau_q(\delta\uppi) \ge \tau_q(\delta)$ holds in finite samples for any choice of $\delta$. Thus we conclude:
    \begin{equation}
        \tau_q(R\uppi) \approx \tau_q(\delta\uppi) \ge \tau_q(\delta) \approx \tau_q(R).
    \end{equation}
    In particular, this holds because MLR statistics ensure $\delta\uppi$ is in descending order.
\end{enumerate}

To show a similar result for the number of \textit{true} discoveries, we repeat the three steps used in the proof of Theorem \ref{thm::avgopt}. To do this, let $I_j$ be the indicator that the feature corresponding to the $j$th coordinate of $R_j$ is non-null, and let $B_j = I_j R_j$ be the indicator that $\sorted(W)_j > 0$ \textit{and} that the corresponding feature is non-null. Let $b = \E[B \mid D]$. Then:

\begin{enumerate}[topsep=0pt]
    \item Let $T_q(R, B)$ denote the number of \textit{true} discoveries. We claim that $T_q(R, B)$ is a function of the successive partial means of $R$ and $B$. To see this, recall from Section \ref{subsec::mainproof} that knockoffs will make $\tau_q(R)$ discoveries, and in particular it will make discoveries corresponding to any of the first $\psi_q(R)$ coordinates of $R$ which are positive. Therefore,
    \begin{equation}\label{eq::tqrb_definition}
        T_q(R, B) = \sum_{j=1}^{\psi_q(R)} B_j = \psi_q(R) \cdot \frac{1}{\psi_q(R)} \sum_{j=1}^{\psi_q(R)} B_j.
    \end{equation}
    Since $\psi_q(R)$ only depends on the successive averages of $R$ and the second term is itself a successive average of $\{B_j\}$, this finishes the first step.
    
    \item The second step is to show that as $p \to \infty$, $\bar B_k \toas \bar b_k, \bar R_k \toas \bar \delta_k$ and therefore $T_q(R, B) \approx T_q(\delta, b)$. This is done using the same techniques as Theorem \ref{thm::avgopt}, although it requires an extra assumption that $B$ also obeys the local dependency condition (\ref{eq::expcovdecay}). Like the original condition, this condition also only depends on the posterior of $B \mid D$, so it can be diagnosed using the data at hand.
    
    \item To complete the proof, we define the adjusted MLR statistic $\AMLR\uppi \in \R^n$ with corresponding $\tilde{R}, \tilde{\delta}, \tilde{b}$ such that $T_q(\tilde{\delta}, \tilde{b}) \ge T_q(\delta, b)$ holds in finite samples for any other feature statistic $W$. It is easy to see that $\AMLR\uppi$ must have the same \textit{signs} as the original MLR statistics $\MLR\uppi$, since the signs of $\MLR\uppi$ maximize $\delta\uppi$ and $b\uppi$ coordinatewise. However, the \textit{absolute values} of $\AMLR\uppi$ may differ from those of $\MLR\uppi$, since it is not always true that sorting $\delta$ in decreasing order maximizes $T_q(\delta, b)$. 
\end{enumerate}

It turns out that the absolute values of the AMLR statistics in Eq. \eqref{eq::amlr} yield vectors $\tilde\delta, \tilde b$ which maximize $T_q(\tilde\delta, \tilde b)$ up to an $O(1)$ additive constant. Theorem \ref{thm::amlr_power_technical} formally proves this, but we now give some intuition. 

\textbf{Intuition}: To determine the optimal absolute values for AMLR statistics, assume WLOG by relabelling the variables that $|\MLR_1\uppi| \ge |\MLR_2\uppi| \ge \dots \ge |\MLR_p\uppi|$. Let $S \subset [p]$ denote an optimization variable representing the set of variables with the $K$ largest absolute values for the AMLR statistics, for some $K$. We will try to design $S$ such that we can make as many true discoveries within $S$ as possible.

As argued above, AMLR and MLR statistics should have the same signs. Thus, roughly speaking, we can discover all features with positive signs among $S$ whenever
\begin{equation*}
    \frac{1}{|S|} \sum_{j \in S} \I(\MLR_j\uppi > 0) \ge (1+q)^{-1}.
\end{equation*}
Making our usual approximation $\I(\MLR_j\uppi > 0) \approx \E[\MLR_j\uppi \mid D] = \delta_j\uppi$, this is equivalent to the constraint
\begin{equation}\label{eq::costs_heuristic}
    \sum_{j \in S} (1+q)^{-1} - \delta_j\uppi \ge 0. 
\end{equation}
Furthermore, if we can discover all of the features with positive signs in $S$, we make exactly $\sum_{j \in S} B_j\uppi \approx \sum_{j \in S} b_j\uppi$ true discoveries, where $B_j\uppi$ is the indicator that the $j$th MLR statistic is positive and the $j$th feature is nonnull. Maximizing this approximate objective subject to the constraint in Eq. \eqref{eq::costs_heuristic} yields the optimization problem:
\begin{equation}
    \max_{S \subset [p]} \sum_{j \in S} b_j\uppi \text{ s.t. } \sum_{j \in S} (1+q)^{-1} - \delta_j\uppi \ge 0.
\end{equation}
In other words, including $j \in S$ has ``benefit" $b_j\uppi$ and ``cost" $(1+q)^{-1} - \delta_j\uppi$. This is a simple integer linear program with one constraint, often known as a ``knapsack" problem. An approximately optimal solution to this problem is to do the following:
\begin{itemize}
    \item Include all variables with ``negative" costs, meaning $\delta_j\uppi = P\bayes(\MLR_j\uppi > 0 \mid D) \ge (1+q)^{-1}$. This is accomplished by ensuring that these features have the largest absolute values.
    \item Prioritize all other variables in descending order of the ratio of the benefit to the cost, $\frac{b_j\uppi}{(1+q)^{-1} - \delta_j\uppi}$.
\end{itemize}
This solution is indeed accomplished by the AMLR formula (Eq. \eqref{eq::amlr}), which gives the highest absolute values to features with negative costs; then, all other absolute values have the same order as the benefit-to-cost ratios $\frac{b_j\uppi}{(1+q)^{-1} - \delta_j\uppi} = \nu_j = \frac{P\uppi(\MLR_j\uppi > 0, j \in \mcH_1(\theta\opt) \mid D)}{(1+q)^{-1} - P\uppi(\MLR_j\uppi > 0 \mid D)}$.

\subsubsection{Theorem statement and proof}\label{subsubsec::amlr_formal}

We now show that AMLR statistics asymptotically maximize the true positive rate. To do this, we require two additional regularity conditions beyond those assumed in Theorem \ref{thm::avgopt}. First, we need a condition that the number of non-nulls under $P_n\bayes$ is not too heavy-tailed; namely, that its coefficient of variation is uniformly bounded.

\begin{assumption}\label{assump::sparsity_second_moment_bound} There exists a constant $C \in \R$ such that as $n \to \infty$, 
$$\frac{\sqrt{\mathrm{Var}_{P_n\bayes}(|\mcH_1(\theta\opt)|)}}{\numnonnull} \le C,$$
where $\numnonnull \defeq \E_{P_n\bayes}[|\mcH_1(\theta\opt)|]$ is the expected number of non-nulls under $P_n\bayes$.
\end{assumption}

Assumption \ref{assump::sparsity_second_moment_bound} is needed for a technical reason. As we will see in Step 3 of the proof, combining this assumption with Lemma \ref{lem::pferbound} ensures that the normalized number of discoveries is uniformly integrable, which is necessary to show that certain error terms converge in $L^1$ to zero. Nonetheless, this assumption is already quite mild, and it is satisfied in previously studied linear and polynomial sparsity regimes \citep{donohojin2004, weinstein2017, ke2020}.

Next, we need an additional local dependence condition.

\begin{assumption}[Additional local dependence condition]\label{assump::localdep_v2} Let $I_j^+ = \I(\MLR_j\bayes > 0, j \in \mcH_1(\theta\opt))$ indicate the event that $j$ is non-null and $\MLR_j\bayes$ is positive. Let $I_j^- = \I(\MLR_j\bayes < 0, j \in \mcH_1(\theta\opt))$ indicate the event that $j$ is non-null and $\MLR_j\bayes$ is negative. We assume that there exist constants $C\ge 0, \rho \in (0,1)$ such that for all $i, j \in [p]$:

\begin{equation}
    |\cov_{P\bayes_n}(I_i^+, I_j^+ \mid D\upn)| \le C \rho^{|i-j|}.
\end{equation}
\begin{equation}
    |\cov_{P\bayes_n}(I_i^-, I_j^+ \mid D\upn)| \le C \rho^{|i-j|}.
\end{equation}
\begin{equation}
    |\cov_{P\bayes_n}(I_i^-, I_j^- \mid D\upn)| \le C \rho^{|i-j|}.
\end{equation}
\end{assumption}

Assumption \ref{assump::localdep_v2} is needed because it implies that for any feature statistic $W$, $\{\I(W_i > 0, i \in \mcH_1(\theta\opt))\}_{i=1}^p$ obey the same local dependence condition.

\begin{lemma}\label{lem::localdep_v2} Assume Assumption \ref{assump::localdep_v2}. Then for any feature statistic $W$ and all $i, j \in [p]$,
\begin{equation*}
    \cov_{P_n\bayes}(\I(W_i > 0, i \in \mcH_1(\theta\opt)), \I(W_j > 0, j \in \mcH_1(\theta\opt)) \mid D\upn) \le C \rho^{|i-j|}.
\end{equation*}
\begin{proof} By Corollary \ref{cor::determdiff}, the event $\sign(W_j) \ne \sign(\MLR_j\pi)$ is $D\upn$-measurable. This implies that for each $j \in [p]$, there is a deterministic (conditional on $D\upn$) choice of $I_j^+, I_j^-$ such that either $\I(W_j > 0, j \in \mcH_1(\theta\opt)) = I_j^+$ or $\I(W_j > 0, j \in \mcH_1(\theta\opt)) = I_j^-$. As a result, we have that
\begin{align*}
        &|\cov_{P_n\bayes}(\I(W_i > 0, i \in \mcH_1(\theta\opt)), \I(W_j > 0, j \in \mcH_1(\theta\opt)) \mid D\upn)| \\
    \le 
        &\max(|\cov_{P\bayes_n}(I_i^+, I_j^+ \mid D\upn)|, |\cov_{P\bayes_n}(I_i^+, I_j^- \mid D\upn)|, (|\cov_{P\bayes_n}(I_i^-, I_j^+ \mid D\upn)|, |\cov_{P\bayes_n}(I_i^-, I_j^- \mid D\upn)|) \\
    \le 
        &C \rho^{|i-j|}
\end{align*}
where the last step follows by Assumption \ref{assump::localdep_v2}.
\end{proof}
\end{lemma}

\begin{theorem}\label{thm::amlr_power_technical} Suppose the conditions of Theorem \ref{thm::avgopt} plus Assumptions \ref{assump::sparsity_second_moment_bound} and \ref{assump::localdep_v2} hold. Let $\amlr_n\uppi$ denote the AMLR statistics with respect to $P_n\bayes$. Then for any sequence of feature statistics $\{w_n\}_{n \in \N}$, the following holds for all but countably many $q \in [0,1]$:
\begin{equation}
    \liminf_{n \to \infty} \frac{\TP_q\bayes(\amlr_n\uppi) - \TP_q\bayes(w_n)}{s_n} \ge 0,
\end{equation}
where $s_n$ is the expected number of non-nulls under $P_n\bayes$ as defined in Assumption \ref{assump::sparsity}, and $\TP_q\bayes(w_n)$ is the expected number of \textit{true} discoveries made by feature statistic $w_n$ under $P_n\bayes$ as defined in Section \ref{appendix::tpr}.
\begin{proof} The proof is in three steps.

\underline{Step 1: Notation and setup.}  Throughout, we use the notation and ideas from Section \ref{subsubsec::tpr_intuition} and the proof of Theorem \ref{thm::avgopt} (Section \ref{subsec::mainproof}), although to ease readability, we will try to give reminders about notation when needed. In particular:
\begin{itemize}
    \item Define $W = w_n([\bX\upn, \tbX\upn], \bigY\upn) \in \R^{p_n}$ and $\AMLR\uppi =$ 
    $\amlr_n\uppi([\bX\upn, \tbX\upn], \bigY\upn) \in \R^{p_n}$. For simplicity, we suppress the dependence on $n$.
    \item Define $R = \I(\sorted(W) > 0)$ and $\tilde R = \I(\sorted(\AMLR\uppi) > 0)$. 
    \item Let $\sigma, \tilde\sigma : [p_n] \to [p_n]$ denote random the permutations such that $\sigma(W)$ and $\tilde\sigma(\AMLR\uppi)$ are sorted in descending order of absolute values; with this notation, note that $R_j \defeq \I(W_{\sigma(j)} > 0), \tilde R_j \defeq \I(\AMLR\uppi_{\tilde\sigma(j)} > 0)$.
    \item Let $I_j = \I(\sigma(j) \in \mcH_1(\theta\opt))$ and $\tilde I_j \defeq  \I(\tilde \sigma(j) \in \mcH_1(\theta\opt))$ be the indicators that the feature statistic with the $j$th largest absolute value of $W$ (resp. $\AMLR\uppi$) represents a non-null feature.
    \item Let $B_j \defeq R_j I_j$ be the indicator that the feature with the $j$th largest absolute value among $W$ is non-null \textit{and} that its feature statistic is postive. Similarly, $\tilde B_j \defeq \tilde R_j \tilde I_j$ is the indicator that the feature with the $j$th largest AMLR statistic is non-null \textit{and} its AMLR statistic is positive. Let $B, \tilde B \in \R^{p_n}$ denote the vectors of these indicators.
    \item Define $\tilde\delta = \E[\tilde R \mid D\upn], \delta = \E[R \mid D\upn], \tilde b = \E[\tilde B \mid D\upn]$ and $b = \E[B \mid D\upn]$ to be the conditional expectations of these quantities given the masked data.
    \item Throughout, we only consider feature statistics whose values are nonzero a.s., because one can provably increase the power of knockoffs by ensuring that each coordinate of $W$ is nonzero.
\end{itemize}

Equation \eqref{eq::tqrb_definition} shows that we can write
\begin{equation}
\TP_q\uppi(w_n) \defeq \E_{P_n\bayes}[|S_{w_n} \cap \mcH_1(\theta\opt)|] = \E_{P_n\bayes}[T_q(R,B)] \text{ where } T_q(R,B) \defeq \psi_q(R) \bar B_{\psi_q(R)} \defeq \sum_{j=1}^{\psi_q(R)} B_j
\end{equation}
and similarly $\TP_q\uppi(\amlr_n\uppi) = \E_{P_n\bayes}[T_q(\tilde R, \tilde B)]$. Therefore it suffices to show that
\begin{equation}
    \liminf_{n \to \infty} \numnonnull^{-1}\left(T_q(\tilde R, \tilde B) - T_q(R, B) \right) \ge 0.
\end{equation}
To do this, we make the following approximation using the triangle inequality:
\begin{align}
        \numnonnull^{-1}\left(T_q(\tilde R, \tilde B) - T_q(R, B) \right) 
    \ge 
        \numnonnull^{-1} \left(\underbrace{T_q(\tilde \delta, \tilde b) - T_q(\delta, b)}_{\text{Term 1}} - \underbrace{|T_q(R, B) - T_q(\delta, b)|}_{\text{Term 2}} - \underbrace{|T_q(\tilde R, \tilde B) - T_q(\tilde \delta, \tilde b)|}_{\text{Term 3}}\right). \label{eq::amlr_threeterms} 
\end{align}
Step 2 of the proof shows that Term 1 is asymptotically positive, and Step 3 shows that Terms 2 and 3 are asymptotically negligible in expectation (i.e., of order $o(\numnonnull)$). This is sufficient to complete the proof.

\underline{Step 2: Analyzing Term 1.} In this step, we show that 
\begin{equation}\label{eq::amlr_step2_goal}
    \liminf_{n \to \infty} \E\left[\frac{T_q(\tilde\delta, \tilde b) - T_q(\delta, b)}{\numnonnull}\right] \ge 0.
\end{equation}
To do this, Step 2a shows that we may assume $\sign(\AMLR\uppi) = \sign(W)$ and therefore $\tilde\delta = \tilde\sigma(\sigma^{-1}(\delta))$ and $\tilde b = \tilde\sigma(\sigma^{-1}(\delta))$ (with the usual notation that $\sigma(\delta) \defeq (\delta_{\sigma(1)}, \dots, \delta_{\sigma(p)})$). Step 2b then proves Eq. \eqref{eq::amlr_step2_goal}.

\underline{Step 2a}: Define $\hat W = |W| \cdot \sign(\AMLR\uppi)$ to have the absolute values of $W$ and the signs of the AMLR statistics. We claim that if $\hat \delta, \hat b$ are defined analogously for $\hat W$ instead of $W$, then $T_q(\delta, b) \le T_q(\hat\delta, \hat b)$. 

To see this, we prove that (i) $\delta_j \le \hat\delta_j$ holds elementwise and (ii) $b_j \le \hat b_j$ holds elementwise. Results (i) and (ii) complete the proof of Step 2(a) since $T_q(\delta, b) = \sum_{j=1}^{\psi_q(\delta)} b_{j}$ is nondecreasing in its inputs (namely because $b$ is nonnegative and $\psi_q(\delta)$ is nondecreasing in its inputs).

To show (i),  Proposition \ref{prop::intuit_opt} shows that $\delta_j \defeq P_n\bayes(W_{\sigma(j)} > 0 \mid D\upn) \le P_n\bayes(\AMLR_{\sigma(j)} > 0 \mid D\upn) = P_n\bayes(\hat W_{\sigma(j)} > 0 \mid D\upn) = \hat \delta_j$, where this argument also uses the facts that (a) MLR and AMLR statistics have the same signs and (b) the permutation $\sigma$ depends only on $|W|$ and thus is $D\upn$-measurable. 

To show (ii), we note that the law of total probability yields
\begin{equation*}
    \delta_j = P_n\bayes(W_{\sigma(j)} > 0 \mid I_j = 1, D\upn) P_n\bayes(I_j = 1 \mid D\upn) + \frac{1}{2} P_n\bayes(I_j = 0 \mid D\upn)
\end{equation*}
where above we use the fact that $P_n\bayes(W_{\sigma(j)} > 0 \mid I_j = 0, D\upn) = \frac{1}{2}$---i.e., under the null, $W_{\sigma(j)}$ is conditionally symmetric. Thus, $\delta_j \le \hat \delta_j$ holds iff $P_n\bayes(W_{\sigma(j)} > 0 \mid I_j = 1, D\upn) \le P_n\bayes(\hat W_{\sigma(j)} > 0 \mid I_j = 1, D\upn)$. Using this result, we conclude
\begin{align*}
        b_j 
    &= 
        P_n\bayes(W_{\sigma(j)} > 0, I_j = 1 \mid D\upn) \\
    &=
        P_n\bayes(W_{\sigma(j)} > 0 \mid I_j = 1, D\upn) P_n\bayes(I_j = 1 \mid D\upn) \\
    &\le 
        P_n\bayes(\hat W_{\sigma(j)} > 0 \mid I_j = 1, D\upn) P_n\bayes(I_j = 1 \mid D\upn) & \text{ since } \delta_j \le \hat \delta_j \text{ by result (i)}\\
    &=
        \hat b_j.
\end{align*}
This proves that $T(\hat\delta, \hat b) \ge T(\delta, b)$. Yet in this step, we seek to show that $T(\tilde \delta, \tilde b) \ge T(\delta, b)$; thus, we may assume that $W = \hat W$ and thus $\sign(W) = \sign(\AMLR\uppi)$. This implies that $\delta, b$ and $\tilde\delta, \tilde b$ take the same values but in different orders;  formally, $\tilde\delta = \tilde\sigma(\sigma^{-1}(\delta))$ and $\tilde b = \tilde\sigma(\sigma^{-1}(\delta))$.

\underline{Step 2b}: Now we show Eq. \eqref{eq::amlr_step2_goal}. Recall that $T_q(\delta, b) = \sum_{j=1}^{\psi_q(\delta)} b_j$ is the partial sum of the first $\psi_q(\delta)$ elements of $b$, where $\psi_q(\delta) \defeq \max \left\{k : \frac{k - k \bar \delta_k + 1}{k \bar \delta_k} \le q \right\}$ is the maximum integer such that $\frac{k - k \bar \delta_k + 1}{k \bar \delta_k} \le q$. It follows that $T_q(\delta, b)$ is bounded by the following quantity:
\begin{equation*}
    T_q(\delta, b) \le \max_{S \subset [p]} \sum_{j \in S} b_j \text{ s.t. } \frac{|S| - |S| \bar \delta_S +1}{|S| \bar \delta_S} \le q,
\end{equation*}
where the notation $\bar \delta_S = \frac{1}{|S|} \sum_{j \in S} \delta_j$ defines the average value of $\delta_j$ averaged over the set $S$. Indeed, this inequality follows because $T_q(\delta, b)$ is precisely the solution to this optimization problem when $S$ is constrained to be a contiguous set of the form $\{1, \dots, k\}$ for some $k$. Relaxing this constraint to allow $S$ to be an arbitrary subset of $[p]$ should only increase the value of the objective. Manipulating this optimization problem yields:
\begin{align*}
    \max_{S \subset [p]} \sum_{j \in S} b_j \text{ s.t. } \frac{|S| - |S| \bar \delta_S +1}{|S| \bar \delta_S} \le q  
    =&
    \max_{S \subset [p]} \sum_{j \in S} b_j \text{ s.t. } |S| - |S| \bar \delta_S +1 \le (|S| \bar \delta_S) q \\
    =&
    \max_{S \subset [p]} \sum_{j =1}^p \I(j \in S) b_j \text{ s.t. } \sum_{j=1}^p \I(j \in S) \left(1 - (1+q) \delta_j \right) \le -1 \\
    =&
    \max_{S \subset [p]} \sum_{j =1}^p \I(j \in S) b_j \text{ s.t. } \sum_{j=1}^p \I(j \in S) \left(\frac{1}{1+q} - \delta_j \right) \le -\frac{1}{1+q}.
\end{align*}
This is an integer linear program with $p$ integer decision variables $x_j \defeq \I(j \in S)$ and one constraint:
\begin{equation*}
    = \max_{x_1, \dots, x_p \in \{0,1\}} \sum_{j=1}^p b_j x_j \text{ s.t. } \sum_{j=1}^p \left(\frac{1}{1+q} - \delta_j \right) x_j \le - \frac{1}{1+q}.
\end{equation*}
Such problems---commonly referred to as knapsack problems---are well studied. The maximum value is bounded by the following greedy strategy:
\begin{itemize}[topsep=0pt]
    \item Let $S\obvious = \{j \in [p] : \delta_j \ge \frac{1}{1+q}\}$ be the set of coordinates such that the constraint coefficient on $x_j$ is negative. This is an ``obvious" set because for $j \in S\obvious$, setting $x_j = 1$ never decreases the objective value (since $b_j \ge 0$) and never increases the constraint value (since $(1+q)^{-1} - \delta_j \le 0$). Thus, any optimal solution must have $x_j = 1$ for all $j \in S\obvious$.
    \item Let $n\obvious = |S\obvious|$ denote the number of obvious coordinates, and let $S\nonobvious = [p] \setminus S\obvious$ denote the non-obvious coordinates. 
    \item After setting $x_j = 1$ for all coordinates $j \in S\obvious$, we should sort the coordinates in $S\nonobvious$ in descending order of the ratio $\frac{b_j}{(1+q)^{-1} - \delta_j}$ and include as many coordinates of $S\nonobvious$ as possible until we hit the constraint that $\sum_{j=1}^p x_j ((1+q)^{-1} - \delta_j) \le - \frac{1}{1+q}$. Then, if we include one additional coordinate, the value of this solution (which violates the constraint by a small amount) bounds the optimal value to the overall objective problem. Indeed, this is because this solution actually has a higher objective value than the solution to the relaxed LP which only requires $x_1, \dots, x_p \in [0,1]$ instead of $x_1, \dots, x_p \in \{0,1\}$ \citep{knapsack_book}.
\end{itemize}
To make this strategy precise, let $\gamma : [p] \to [p]$ denote any permutation with the following properties:
\begin{itemize}[topsep=0pt]
    \item $\gamma(1), \dots, \gamma(n\obvious) \in S\obvious$. I.e., the first $n\obvious$ coordinates specified by $\gamma$ are the set $S\obvious$.
    \item For $i, j > n\obvious$, $\gamma(i) \ge \gamma(j)$ if and only if $\frac{b_i}{(1+q)^{-1} - \delta_i} \ge \frac{b_j}{(1+q)^{-1} - \delta_j}$. I.e., $\gamma$ orders the rest of the coordinates by the ratio $\frac{b_j}{(1+q)^{-1} - \delta_j}$.
\end{itemize}
Then, let $k\opt \defeq \max\left\{k \in [p] : \sum_{j=1}^k ((1+q)^{-1} - \delta_{\gamma(j)}) \le - \frac{1}{1+q} \right\}$ denote the maximum value of $k$ such that setting $x_{\gamma(1)}, \dots, x_{\gamma(k)} = 1$ yields a feasible solution to the integer LP. Then we have that
$$T_q(\delta, b) \le \max_{x_1, \dots, x_p \in \{0,1\}} \left[\sum_{j=1}^p b_j x_j \text{ s.t. } \sum_{j=1}^p \left(\frac{1}{1+q} - \delta_j \right) x_j \le - \frac{1}{1+q}\right] \le \sum_{j=1}^{k\opt + 1} b_{\gamma(j)}.$$
\begin{remark}\label{rem::sigma_not_unique_but_its_fine} $\gamma$ is not uniquely specified by the construction above, but the bound holds for any such $\gamma$.
\end{remark}

Step 2a shows that for some permutation $\kappa : [p] \to [p]$, we can write $\tilde\delta = \kappa(\delta)$ and $\tilde b = \kappa(b)$. We note that $\kappa$ must satisfy the following:
\begin{itemize}[topsep=0pt]
    \item $\kappa(1), \dots, \kappa(n\obvious) \in S\obvious$. This is because the AMLR statistics with the top absolute values are constructed to be the AMLR statistics such that $P_n\bayes(\AMLR_j\uppi > 0 \mid D) = P_n\bayes(\MLR_j\uppi > 0 \mid D) \ge (1+q)^{-1}$ (see Definition \ref{def::amlr}), which exactly coincides with the definition of the set $S\obvious$.
    \item For $i, j > n\obvious$, $\kappa(i) \ge \kappa(j)$ if and only if $\frac{b_i}{(1+q)^{-1} - \delta_i} \ge \frac{b_j}{(1+q)^{-1} - \delta_j}$. This again follows from Definition \ref{def::amlr}, as the absolute values of AMLR statistics are explicitly chosen to guarantee this.
\end{itemize}
In other words, $\kappa$ satisfies the same properties as $\gamma$ above. Thus, we may take $\gamma = \kappa$. This yields the bound $T_q(\delta, b) \le \sum_{j=1}^{k\opt + 1} b_{\gamma(j)}$. However, we also know that
\begin{align*}
        T_q(\tilde \delta, \tilde b) 
    = 
        \sum_{j=1}^{\psi_q(\tilde \delta)} \tilde b_j
    =
        \sum_{j=1}^{\psi_q(\tilde \delta)} b_{\gamma(j)}
    =
        \sum_{j=1}^{k\opt} b_{\gamma(j)},
\end{align*}
where the last step follows because
\begin{align*}
    \psi_q(\tilde\delta) = \psi_q(\kappa(\delta)) 
    &\defeq \max\left\{k \in [p] : \frac{k - \sum_{j=1}^k \delta_{\gamma(j)} + 1}{\sum_{j=1}^k \delta_{\gamma(j)}} \le q \right\} & \text{ by definition} \\ 
    &= \max\left\{k \in [p] : \sum_{j=1}^k ((1+q)^{-1} - \delta_{\gamma(j)}) \le - \frac{1}{1+q}\right\} & \text{ by algebraic manipulation}  \\
    &\defeq k\opt. & \text{ by definition}
\end{align*}
To summarize, this implies that
\begin{equation*}
    T_q(\delta, b) \le \sum_{j=1}^{k\opt + 1} b_{\gamma(j)} \le \sum_{j=1}^{k\opt} b_{\gamma(j)} + 1 = T_q(\tilde \delta, \tilde b) + 1.
\end{equation*}
This completes the proof of Step 2, since as a consequence,
\begin{equation}
    \E\left[\frac{T_q(\tilde\delta, \tilde b) - T_q(\delta, b)}{\numnonnull}\right] \ge - \frac{1}{\numnonnull} \to 0.
\end{equation}

\underline{Step 3}: In this step, we show that $\frac{\E[|T_q(\delta, b) - T_q(R,B)|]}{\numnonnull} \to 0$ holds for all but countably many $q \in [0,1]$. The same logic applies to the term involving $T_q(\tilde R, \tilde B) - T_q(\tilde\delta, \tilde b)$.

To do this, we will essentially bound $|T_q(\delta, b) - T_q(R, B)|$ by the quantities $\max_{k \ge k_0} |\bar B_k - \bar b_k|$ and $|\psi_q(R) - \psi_q(\delta)|$; as we shall see, these are both $o_{L_1}(\numnonnull)$. To begin, the triangle inequality yields:
\begin{align}
        |T_q(R, B) - T_q(\delta, b)|
    &=
        \left|\psi_q(R) \bar B_{\psi_q(R)} - \psi_q(\delta) \bar b_{\psi_q(\delta)} \right| & \text{ by definition} \nonumber\\
    &=
        \left|\psi_q(R) \bar B_{\psi_q(R)} - \psi_q(R) \bar b_{\psi_q(\delta)} + \psi_q(R) \bar b_{\psi_q(\delta)} - \psi_q(\delta) \bar b_{\psi_q(\delta)} \right| \nonumber\\
    &\le 
        \psi_q(R) |\bar B_{\psi_q(R)} - \bar b_{\psi_q(\delta)}| + \bar b_{\psi_q(\delta)} |\psi_q(R) - \psi_q(\delta)| & \text{ triangle inequality} \nonumber\\
    &\le
        \underbrace{\psi_q(R) |\bar B_{\psi_q(R)} - \bar b_{\psi_q(\delta)}|}_{\text{Term (a)}} + \underbrace{|\psi_q(R) - \psi_q(\delta)|}_{\text{Term (b)}} & \text{ since } \bar b_{\psi_q(\delta)} \in [0,1]. \label{eq::true_power_int_result_1}
\end{align}
First, we bound Term (a). Applying the triangle inequality (again) plus a simple algebraic lemma yields:
\begin{align*}
    \psi_q(R) |\bar B_{\psi_q(R)} - \bar b_{\psi_q(\delta)}|
    &\le 
    \psi_q(R) |\bar B_{\psi_q(R)} - \bar b_{\psi_q(R)}| + \psi_q(R) |\bar b_{\psi_q(R)} - \bar b_{\psi_q(\delta)}| \\
    &\le 
    \psi_q(R) |\bar B_{\psi_q(R)} - \bar b_{\psi_q(R)}| + 2 |\psi_q(R) - \psi_q(\delta)| & \text{ by Lemma \ref{lem::algebra4powerresult}}.
\end{align*}
The second line is not obvious, but it follows because Lemma \ref{lem::algebra4powerresult} shows that $|\bar b_{\psi_q(R)} - \bar b_{\psi_q(\delta)}| \le \frac{2|\psi_q(R) - \psi_q(\delta)|}{\max(\psi_q(R), \psi_q(\delta))}$. This algebraic result follows intuitively because $\bar b_k \in [0,1]$ holds for all $k$; thus, $|\bar b_{\psi_q(R)} - \bar b_{\psi_q(\delta)}|$ cannot be large unless $|\psi_q(R) - \psi_q(\delta)|$ is also large.

Combining this bound on Term (a) with the initial result in Eq. \eqref{eq::true_power_int_result_1} yields:
\begin{equation*}
    {\numnonnull}^{-1} \E\left[|T_q(R, B) - T_q(\delta, b)|\right] \le 3\underbrace{ {\numnonnull}^{-1} \E\left[|\psi_q(R) - \psi_q(\delta)|\right]}_{\text{Term (b)}} + \underbrace{{\numnonnull}^{-1} \E\left[\psi_q(R) |\bar B_{\psi_q(R)} - \bar b_{\psi_q(R)}|\right]}_{\text{Term (c)}}.
\end{equation*}
To show that Term (b) vanishes, recall that Step 2 of Theorem \ref{thm::avgopt} shows that $\numnonnull^{-1} \E[|\tau_q(R) - \tau_q(\delta)|] \to 0$ for all but countably many $q \in [0,1]$. However, $\tau_q(R) = \ceil{\frac{\psi_q(R) + 1}{1+q}} = \frac{1}{1+q} \psi_q(R) + O(1)$ and similarly $\tau_q(\delta) = \frac{1}{1+q} \psi_q(\delta) + O(1)$. Therefore, ${\numnonnull}^{-1} \E\left[|\psi_q(R) - \psi_q(\delta)|\right] \to 0$ for all but countably many $q \in [0,1]$.

Thus, it suffices to show that Term (c) vanishes. To do this, fix a sequence of integers $\{k_n\}_{n=1}^{\infty}$ such that $k_n \sim \log(p_n)^5$. Separately considering the cases where $\psi_q(R) < k_n$ and $\psi_q(R) \ge k_n$ yields
\begin{align*}
    \E\left[{\numnonnull}^{-1} \psi_q(R) |\bar B_{\psi_q(R)} - \bar b_{\psi_q(R)}|\right] 
    &\le
    \E[|\bar B_{\psi_q(R)} - \bar b_{\psi_q(R)}| \numnonnull^{-1} k_n] + \E[\numnonnull^{-1} \psi_q(R) \max_{k \ge k_n} |\bar B_k - \bar b_k|] \\
    &\le 
    2 \numnonnull^{-1} k_n + \underbrace{\E[\numnonnull^{-1} \psi_q(R) \max_{k \ge k_n} |\bar B_k - \bar b_k|]}_{\text{Term (d)}}
\end{align*}
where the second line follows because  $\bar B_{\psi_q(R)}, \bar b_{\psi_q(R)} \in [0,1]$.  Assumption \ref{assump::sparsity} guarantees that that $\numnonnull^{-1} k_n \to 0$, so it suffices to show that Term (d) vanishes. To do this, we observe:
\begin{itemize}[topsep=0pt, leftmargin=*]
    \item Assumption \ref{assump::localdep_v2} plus Lemma \ref{lem::localdep_v2} shows that $B$ satisfies an exponential decay condition, so we may apply Lemma \ref{cor::expdecay}. Namely, for $k_n \sim \log(p_n)^5$, $\max_{k \ge k_n} |\bar B_k - \bar b_k| \toprob 0$ (see Theorem \ref{thm::avgopt}, Step 2 for more details).
    \item Lemma \ref{lem::pferbound} shows that there exists a universal constant $C$ such that $\E[\tau_q(R)^2] \le C \E_{P_n\bayes}[|\mcH_1(\theta\opt)|^2]$. Note that the coefficient of variation of $|\mcH_1(\theta\opt)|$ is bounded under Assumption \ref{assump::sparsity_second_moment_bound}, so there exists another universal constant $C'$ such that $\E_{P_n\bayes}[\psi_q(R)^2] \sim \E_{P_n\bayes}[\tau_q(R)^2] \le C' \numnonnull^2$. This implies that ${\numnonnull}^{-1} \psi_q(R)$ is uniformly integrable since it has a uniformly bounded second moment.
    \item The latter two observations imply that $\E[\numnonnull^{-1} \psi_q(R) \max_{k \ge k_n} |\bar B_k - \bar b_k|] \to 0$, since $\numnonnull^{-1} \psi_q(R)$ is uniformly integrable, and  $\max_{k \ge k_n} |\bar B_k - \bar b_k| \le 2$ is also uniformly bounded and vanishes in probability.
\end{itemize}
Together, this proves that for all but countably many $q \in [0,1]$, $\frac{\E[|T_q(\delta, b) - T_q(R, B)|]}{\numnonnull} \to 0$.
\end{proof}
\end{theorem}

The following algebraic lemma is used at the end of Step 3 of the proof of Theorem \ref{thm::amlr_power_technical}.

\begin{lemma}\label{lem::algebra4powerresult} For $b = (b_1, \dots, b_p) \in [0,1]^p$, fix $k, \ell \in [p]$. Then
\begin{equation*}
    \max(k, \ell) \cdot |\bar b_k - \bar b_\ell| \le 2 |k - \ell|.
\end{equation*}
\begin{proof} Define $m = \min(k, \ell)$ and $M = \max(k, \ell)$. The lemma holds trivially when $m = M$, so we may assume $m < M$. Then we have that
\begin{align*}
        \left|\bar b_k - \bar b_{\ell}\right|
    &\defeq 
        \left|\frac{1}{k} \sum_{i=1}^k b_i - \frac{1}{\ell} \sum_{i=1}^{\ell} b_i\right| \\
    &=
        \left|\sum_{i=1}^{m} b_i \left(\frac{1}{m} - \frac{1}{M}\right) - \frac{1}{M} \sum_{i=m+1}^{M} b_i \right| & \text{ by definition of } m, M \\
    &\le
        \sum_{i=1}^{m} b_i \left(\frac{1}{m} - \frac{1}{M}\right)
        + \frac{1}{M} \sum_{i=m+1}^{M} b_i & \text{ since } b_i \ge 0, m < M  \\
    &\le 
        m \left(\frac{1}{m} - \frac{1}{M}\right) + \frac{1}{M} \left(M - m\right) & \text{ since } b_i \le 1.
\end{align*}
Therefore we conclude that
\begin{equation*}
    M |\bar b_k - \bar b_{\ell}| \le M m \left(\frac{1}{m} - \frac{1}{M} \right) + \frac{M}{M} (M - m) = 2(M-m) = 2 |k - \ell|
\end{equation*}
which completes the proof.
\end{proof}
\end{lemma}

\subsection{Verifying the local dependence assumption in a simple setting}\label{subsec::unhappy}

We now verify the local dependency condition (\ref{eq::expcovdecay}) in the setting where $\bX^T \bX$ is block-diagonal and $\sigma^2$ is known. 

\begin{proposition}\label{prop::unhappy} Suppose $\bX^T \bX$ is a block-diagonal matrix with maximum block size $M \in \N$. Suppose $P\bayes$ is any Bayesian model such that (i) the model class $\mcP$ is the class of all Gaussian models of the form $\bigY \mid \bX \sim \mcN(\bX \beta, \sigma^2 I_n)$ for $\theta = (\beta, \sigma^2) \in \Theta \defeq \R^p \times \R_{\ge 0}$, (ii) the coordinates of $\beta$ are marginally independent under $P\bayes$ and (iii) $\sigma^2$ is a constant under $P\bayes$. Then if $\tbX$ are either fixed-X knockoffs or conditional Gaussian model-X knockoffs \citep{condknock2019}, the coordinates of $\sign(\MLR\uppi)$ are $M$-dependent conditional on $D$ under $P\bayes$, implying that Equation (\ref{eq::expcovdecay}) holds, e.g., with $C = 2^M$ and $\rho = \frac{1}{2}$.
\end{proposition}
\begin{proof} Define $R \defeq \I(\MLR\uppi > 0)$. We will prove the stronger result that if $J_1, \dots, J_m \subset [p]$ are a partition of $[p]$ corresponding to the blocks of $\bX^T \bX$, then $R_{J_1}, \dots, R_{J_m}$ are jointly independent conditional on $D$. As notation, suppose without loss of generality that $J_1, \dots, J_m$ are contiguous subsets and $\bX^T \bX = \diag{\Sigma_1, \dots, \Sigma_m}$ for $\Sigma_i \in \R^{|J_i| \times |J_i|}$.  All probabilities and expectations are taken under $P\bayes$.

We give the proof for model-X knockoffs; the proof for fixed-X knockoffs is quite similar. Recall by Proposition \ref{prop::mxdistinguish} that we can write $R_j = \I(W_j > 0) = \I(\bX_j = \widehat{\bX}_j)$ where $\widehat{\bX}_j$ is a function of the masked data $D$. Therefore, to show $R_{J_1}, \dots, R_{J_m}$ are independent conditional on $D$, it suffices to show $\bX_{J_1}, \dots, \bX_{J_m}$ are conditionally independent given $D$. To do this, it will first be useful to note that the likelihood is
\begin{align*}
        P_{\bigY \mid \bX}^{(\beta,\sigma)}(\bigY \mid \bX) 
    &\propto 
        \exp\left(-\frac{1}{2 \sigma^2}||\bigY - \bX \beta||_2^2 \right) \\
    &\propto
        \exp\left(\frac{2 \beta^T \bX^T \bigY - \beta^T \bX^T \bX \beta}{2 \sigma^2}\right) \\
    &\propto
        \prod_{i=1}^m \exp\left(\frac{2 \beta_{J_i}^T \bX_{J_i}^T \bigY - \beta_{J_i}^T \Sigma_i \beta_{J_i}}{2 \sigma^2}\right),
\end{align*}
where above, we only include terms depending on $\bX$, since these are the only terms relevant to the later stages of the proof. A subtle but important observation in the calculation above is that we can verify that $\bX^T \bX = \diag{\Sigma_1, \dots, \Sigma_m}$ having only observed $D$ without observing $\bX$. Indeed, this follows because for conditional Gaussian MX knockoffs, $\tbX^T \tbX = \bX^T \bX$ and $\tbX^T \bX$ only differs from $\bX^T \bX$ on the main diagonal (just like in the fixed-X case). With this observation in mind, we now abuse notation slightly and let $p(\cdot \mid \cdot)$ denote an arbitrary conditional density under $P\bayes$. Observe that
\begin{align*}
        p(\bX \mid D)
    &\propto 
        p(\bX, \bigY \mid \{\bX_j, \tbX_j\}_{j=1}^p) \\
    &=
        p(\bX \mid \{\bX_j, \tbX_j\}_{j=1}^p)
        p(\bigY \mid \bX) \,\,\,\,\,\,\,\,\,\, \text{ since } \bigY \Perp \tbX \mid \bX \\
    &=
        \frac{1}{2^p} p(\bigY \mid \bX) \,\,\,\,\,\,\,\,\,\,\,\,\,\,\,\,\,\,\,\, \text{ by pairwise exchangeability} \\
    &\propto
        \int_{\beta} p(\beta) p(\bigY \mid \bX, \beta) d\beta  \\
    &\propto
        \int_{\beta} \prod_{i=1}^m p(\beta_{J_i}) \exp\left(\frac{- \beta_{J_i}^T \Sigma_i \beta_{J_i}}{2 \sigma^2} \right) \exp\left(\frac{\beta_{J_i}^T \bX_{J_i}^T \bigY}{\sigma^2} \right) d\beta  \\
    &\propto
    \int_{\beta_{J_1}} \dots, \int_{\beta_{J_m}}  \prod_{i=1}^m p(\beta_{J_i}) \exp\left(\frac{- \beta_{J_i}^T \Sigma_i \beta_{J_i}}{2 \sigma^2} \right) \exp\left(\frac{\beta_{J_i}^T \bX_{J_i}^T \bigY}{\sigma^2} \right) d\beta_{J_1} \, \beta_{J_2}  \dots \, d\beta_{J_m}.
\end{align*}
At this point, we can iteratively pull out parts of the product. In particular, define the following function:
\begin{equation*}
    q_i(\bX_{J_i}) \defeq \int_{\beta_{J_i}} p(\beta_{J_i}) \exp\left(\frac{- \beta_{J_i}^T \Sigma_i \beta_{J_i}}{2 \sigma^2} \right) \exp\left(\frac{\beta_{J_i}^T \bX_{J_i}^T \bigY}{\sigma^2} \right) d \beta_{J_i}.
\end{equation*}
Since $\bigY, \sigma^2$ and $\Sigma_i$ are fixed, $q_i(\bX_{J_i})$ is a deterministic function of $\bX_{J_i}$ that does not depend on $\beta_{-J_i}$. Therefore, we can iteratively integrate as below:
\begin{align*}
        p(\bX \mid D)
    &\propto
        \int_{\beta_{J_1}} \dots, \int_{\beta_{J_m}} \prod_{i=1}^m p(\beta_{J_i}) \exp\left(\frac{- \beta_{J_i}^T \Sigma_i \beta_{J_i}}{2 \sigma^2} \right) \exp\left(\frac{\beta_{J_i}^T \bX_{J_i}^T \bigY}{\sigma^2} \right) d\beta_{J_1} \, d\beta_{J_2}  \dots \, d\beta_{J_m} \\
    &=
        \int_{\beta_{J_1}} \dots \int_{\beta_{J_{m-1}}} \prod_{i=1}^{m-1} p(\beta_{J_i}) \exp\left(\frac{- \beta_{J_i}^T \Sigma_i \beta_{J_i}}{2 \sigma^2} \right) \exp\left(\frac{\beta_{J_i}^T \bX_{J_i}^T \bigY}{\sigma^2} \right) q_m(\bX_{j_m}) d\beta_{J_1} \, d\beta_{J_2} \dots \, d\beta_{J_{m-1}} \\
    &=
        \prod_{i=1}^m q_i(\bX_{J_i}).
\end{align*}
This shows that $\bX_{J_1}, \dots, \bX_{J_m} \mid D$ are jointly (conditionally) independent since their density factors, thus completing the proof. For fixed-X knockoffs, the proof is very similar as one can show that the density of $\bX^T \bigY \mid D$ factors into blocks.
\end{proof}

\subsection{Intuition for the local dependency condition and Figure \ref{fig::dependence}}\label{appendix::dependencediscussion}

In Figure \ref{fig::dependence}, we see that even when $\bX$ is very highly correlated, $\cov_{P\bayes}(\sign(\MLR\uppi) \mid D)$ looks similar to a diagonal matrix, indicating that the local dependency condition (\ref{eq::expcovdecay}) holds well empirically. The empirical result is striking and may be surprising initially, this section offers some explanation.

For the sake of intuition, suppose that we are fitting model-X knockoffs and using the Bayesian model $P\bayes$ from Example \ref{ex::sparsegam} with the original features. Suppose we observe that the masked data $D$ is equal to some fixed value $d = (\smally, \{\bx_j, \tbx_j\}_{j=1}^p)$. After observing $D = d$, Appendix \ref{appendix::gibbs} shows how to sample from the posterior distribution $\bX \mid D = d$ via the following Gibbs sampling approach:
\begin{itemize}
    \item For each $j \in [p]$, initialize $\beta_j^{(0)}$ and $\bX_j^{(0)}$ to some value.
    \item For $i = 1, \dots, n_\sample$:
    \begin{enumerate}
        \item Set $\beta^{(i)} = \beta^{(i-1)}$ and $\bX^{(i)} = \bX^{(i-1)}$.
        \item For $j \in [p]$:
            \begin{enumerate}[(a)]
                \item Resample $\bX_j^{(i)}$ from the law of $\bX_j \mid \bX_{-j} = \bX_{-j}^{(i)}, \beta_{-j} = \beta_{-j}^{(i)}, D = d$. It may be helpful to recall $\bX_j^{(i)} \in \{\bx_j, \tbx_j\}$ holds deterministically.
                \item Resample $\beta_j^{(i)}$ from the law of $\beta_j \mid \bX = \bX^{(i)}, \beta_{-j} = \beta_{-j}^{(i)}, D = d$.
            \end{enumerate}
        \end{enumerate}

    \item Return samples $\bX^{(1)}, \dots, \bX^{(n_\sample)}$.
\end{itemize}

Now, recall that $\MLR_j\uppi > 0$ if and only if the Gibbs sampler consistently chooses $\bX_j^{(i)} = \bx_j$ instead of $\bX_j^{(i)} = \tbx_j$. Thus, to analyze $\cov_{P\bayes}(\I(\MLR_j\uppi > 0), \I(\MLR_k\uppi > 0) \mid D)$ for some fixed $j \ne k$, we must ask the following question: does the value of $\bX_k^{(i)} \in \{\bx_k, \tbx_k\}$ strongly affect how we resample $\bX_j^{(i)}$?

To answer this question, the following key fact (reviewed in Appendix \ref{appendix::gibbs}) is useful. At iteration $i$, step 2(a), for any $j$, let $r_j = \smally - \bX_{-j}^{(i)} \beta_{-j}^{(i)}$ be the residuals excluding feature $j$ for the current value of $\bX_{-j}$ and $\beta_{-j}$ in the Gibbs sampler. Then, standard calculations show that $P\bayes(\bX_j = \bx_j \mid D=d, \bX_{-j}=\bX_{-j}^{(i)}, \beta_{-j}=\beta_{-j}^{(i)})$ only depends on $\bX_{-j}^{(i)}$ and $\beta_{-j}^{(i)}$ through the following inner products:

\begin{equation*}
    \alpha_j \defeq \bx_j^T r_j = \bx_j^T \smally - \sum_{\ell \ne j} \bx_j^T \bX_\ell^{(i)} \beta_\ell^{(i)}
\end{equation*}
\begin{equation*}
    \tilde{\alpha}_j \defeq \tbx_j^T r_j = \tbx_j^T \smally - \sum_{\ell \ne j} \tbx_j^T \bX_\ell^{(i)} \beta_\ell^{(i)}.
\end{equation*}
Thus, the question we must answer is: how does the choice of $\bX_k^{(i)} \in \{\bx_k, \tbx_k\}$ affect the value of $\alpha_j, \tilde{\alpha}_j$? Heuristically, the answer is ``not very much," since $\bX_k^{(i)}$ only appears above through inner products of the form $\bx_j^T \bX_k^{(i)}$ and $\tbx_j^T \bX_k^{(i)}$, and by definition of the knockoffs we know that $\bx_j^T \bx_k \approx \bx_j^T \tbx_k$ and $\tbx_j^T \bx_k \approx \tbx_j^T \tbx_k$. Indeed, for fixed-X knockoffs, we know that this actually holds exactly, and for model-X knockoffs, the law of large numbers should ensure that these approximations are very accurate.

The main way that the choice of $\bX_k^{(i)}$ can significantly influence the choice of $\bX_j^{(i)}$ is that the choice of $\bX_k^{(i)}$ may change the value of $\beta_k^{(i)}$. In general, we expect this effect to be rather small, since in many highly-correlated settings, $\bx_k$ and $\tbx_k$ are necessarily highly correlated and thus the choice of $\bX_k^{(i)} \in \{\bx_k, \tbx_k\}$ should not affect the choice of $\beta_k^{(i)}$ too much. That said, there are a few known pathological settings where the choice of $\bX_k^{(i)} \in \{\bx_k, \tbx_k\}$ does substantially change the estimated value of $\beta_k$ (\cite{altsign2017, mrcknock2022}), and in these settings, the coordinates of $\sign(\MLR\uppi)$ may be strongly conditionally dependent. The good news is that using MVR knockoffs instead of SDP knockoffs should ameliorate this problem (see \cite{mrcknock2022}).

Overall, we recognize that this explanation is purely heuristic and does not fully explain the results in Figure \ref{fig::dependence}. However, it may provide some intuitive insight. A more rigorous theoretical analysis of $\cov_{P\bayes}(\I(\MLR\uppi > 0) \mid D)$ would be interesting; however, we leave this to future work.

\section{Technical proofs}\label{appendix::techproofs}

\subsection{Key concentration results}

The proof of Theorem \ref{thm::avgopt} relies on the fact that the successive averages of the vector $\I(\sorted(W) > 0) \in \R^p$ converge uniformly to their conditional expectation given the masked data $D\upn$. In this section, we give a brief proof of this result, which is essentially an application of Theorem 1 from \cite{doukhan2007}. For convenience, we first restate a special case of this theorem (namely, the case where the random variables in question are bounded and we have bounds on pairwise correlations) before proving the corollary we use in Theorem \ref{thm::avgopt}.

\begin{theorem}[\cite{doukhan2007}]\label{thm::doukhan2007} Suppose that $X_1, \dots, X_n$ are mean-zero random variables taking values in $[-1,1]$ such that $\var(\bar X_n) \le C_0 n$ for a constant $C_0 > 0$. Let $L_1, L_2 < \infty$ be constants such that for any $i \le j$,
\begin{equation*}
    |\cov(X_i, X_j)| \le 4\varphi(j-i) 
\end{equation*}
where $\{\varphi(k)\}_{k \in \N}$ is a nonincreasing sequence satisfyng
\begin{equation*}
    \sum_{s=0}^{\infty} (s+1)^k \varphi(s) \le L_1 L_2^k k! \text{ for all } k \ge 0.
\end{equation*}
Then for all $t \in (0,1)$, there exists a universal constant $C_1 > 0 $ only depending on $C_0, L_1$ and $L_2$ such that
\begin{equation*}
    \P\left(\bar X_n \ge t \right) \le \exp\left(-\frac{t^2}{C_0 n + C_1 t^{7/4} n^{7/4}}\right) \le \exp\left(-C' t^{2 } n^{1/4}\right),
\end{equation*}
where $C'$ is a universal constant only depending on $C_0, L_1, L_2$.
\end{theorem}
If we take $\varphi(s) = c \rho^s$, this yields the following corollary.

\begin{corollary}\label{cor::expdecay} Suppose that $X_1, \dots, X_n$ are mean-zero random variables taking values in $[-1,1]$. Suppose that for some $C \ge 0, \rho \in (0,1)$, the sequence satisfies
\begin{equation}\label{eq::expdecay}
    |\cov(X_i, X_j)| \le C \rho^{|i-j|}.
\end{equation}
Then there exists a universal constant $C'$ depending only on $C$ and $\rho$ such that
\begin{equation}\label{eq::expconc}
    \P(\bar X_n \ge t) \le \exp\left(-C' t^{2 } n^{1/4}\right).
\end{equation}
Furthermore, let $\pi : [n] \to [n]$ be any permutation. For $k \le n$, define $\bar X_k^{(\pi)} \defeq \frac{1}{k} \sum_{i=1}^k X_{\pi(i)}$ to be the sample mean of the first $k$ random variables after permuting $(X_1, \dots, X_n)$ according to $\pi$. Then for any $n_0 \in \N, t \ge 0$, 
\begin{equation}\label{eq::permconc}
    \sup_{\pi \in S_n} \P\left(\max_{n_0 \le i \le n} |\bar X_k^{(\pi)}| \ge t \right) \le n \exp(-C' t^{2 } n_0^{1/4}). 
\end{equation}
where $S_n$ is the symmetric group.

\begin{proof} The proof of Equation (\ref{eq::expconc}) follows an observation of \cite{doukhan2007}, where we note $\varphi(s) = C \exp(-a s)$ for $a = - \log(\rho)$. Then
\begin{equation*}
\sum_{s=0}^{\infty} (s+1)^k \exp(-as) \le \sum_{s=0}^{\infty} \prod_{i=1}^k (s+i) \exp(-as) = \frac{d^k}{dp^k} \left(\frac{1}{1-p}\right) \bigg|_{p=\exp(-a)} = k! \frac{1}{(1-\exp(-a))^k}.
\end{equation*}
As a result, $\sum_{s=0}^{\infty} (s+1)^k \varphi(s) \le C \left(\frac{1}{(1-\exp(-a))}\right)^k k!$, so we take $L_1 = \frac{1}{(1-\exp(-a))}$ and $L_2 = C$. Lastly, we observe that another geometric series argument yields
\begin{equation*}
    \var(\bar X_n) = \sum_{i=1}^n \sum_{j=1}^n \cov(X_i, X_j) \le \sum_{i=1}^n C \sum_{j=1}^n \rho^{|i-j|} \le n C \frac{2}{1-\rho}.
\end{equation*}
Thus, we take $C_0 = \frac{2C}{1-\rho}$ and apply Theorem \ref{thm::doukhan2007}, which yields the first result. To prove Equation (\ref{eq::permconc}), the main idea is that we can apply Equation (\ref{eq::expconc}) to each sample mean $|\bar X_k^{(\pi)}|$, at which point the Equation (\ref{eq::permconc}) follows from a union bound.

To prove this, note that if we rearrange $(X_{\pi(1)}, \dots, X_{\pi(k)})$ into their ``original order," then these variables satisfy the condition in Equation (\ref{eq::expdecay}). Formally, let $A = \{\pi(1), \dots,\pi(k)\}$ and let $\nu : A \to A$ be the permutation such that $\nu(\pi(i)) > \nu(\pi(j))$ if and only if $i > j$, for $i,j \in [k]$. Then define $Y_i = X_{\nu(\pi(i))}$ for $i \in [k]$, and note that
\begin{equation*}
    |\cov(Y_i, Y_j)| = |\cov(X_{\nu(\pi(i))}, X_{\nu(\pi(j))})| \le C \rho^{|\nu(\pi(i)) - \nu(\pi(j))|} \le C \rho^{|i-j|},
\end{equation*}
where in the last step, $|i-j| \le |\nu(\pi(i)) - \nu(\pi(j))|$ follows by construction of $\nu$. Since $\bar Y_k = \bar X_k^{(\pi)}$ by construction, this means we may apply Equation \eqref{eq::expconc} to $\bar X_k^{(\pi)}$ for each $k$.

Thus, by Equation (\ref{eq::expconc}), for any $\pi \in S_n$,
\begin{equation*}
    \P\left(\max_{n_0 \le k \le n} |\bar X_k^{(\pi)}| \ge t \right) \le \sum_{k=n_0}^n \P(|\bar X_k^{(\pi)}| \ge t) \le \sum_{k=n_0}^n \exp(-C' t^{2 } k^{1/4}) \le n \exp(-C' t^{2 } n_0^{1/4}).
\end{equation*}
This completes the proof.
\end{proof}
\end{corollary}

\subsection{Bounds on the expected number of false discoveries}\label{subsec::pferbound}

The proof of Theorem \ref{thm::avgopt} relied on the fact that $\lim_{n \to \infty} \Power_q(w_n)$ is finite whenever it exists. This is a consequence of the lemma below. The lemma below also proves a second moment bound which is needed when making a uniform integrability argument in Step 3 of the proof of Theorem \ref{thm::amlr_power_technical}.

\begin{lemma}\label{lem::pferbound} Fix any $q \in (0,1)$. Then there exist universal constants $C(q), C^{(2)}(q) \in \R$ such that for any Bayesian model $P\bayes$ and any valid knockoff statistic $W = w([\bX, \tbX], \bigY)$ with discovery set $S \subset [p]$:
\begin{enumerate}
    \item $\Power_q(w) \le C(q)$ where $C(q)$ is a finite constant depending only on $q$.
    \item $\E_{P\bayes}[|S|^2] \le C^{(2)}(q) \E[|\mcH_1(\theta\opt)|^2]$.
\end{enumerate}
Note that above, $\mcH_1(\theta\opt)$ denotes the random set of non-nulls under $P\bayes$.
\begin{proof} Recall that $P\bayes$ denotes the joint law of $(\bX, \bigY, \theta\opt)$. Throughout the proof, all expectations and probabilities are taken over $P\bayes$. Our strategy is to condition on the nuisance parameters $\theta\opt$. In particular, let $M(\theta\opt) = |\mcH_1(\theta\opt)|$ denote the number of non-nulls. To show the first result, it suffices to show
\begin{equation}\label{eq::pferbound}
    \E\left[|S| \mid \theta\opt \right] \le C(q) M(\theta\opt).
\end{equation}
Proving Equation (\ref{eq::pferbound}) proves the first result because it implies by the tower property that $\E[|S|] \le C(q) \E[M(\theta\opt)]$, and therefore $\Power_q(w) = \frac{\E[|S|]}{\E[M(\theta\opt)]} \le C(q)$.
For the second result, by the tower property it also suffices to show that
\begin{equation}\label{eq::secondmoment_bound}
    \E[|S|^2 \mid \theta\opt] \le C^{(2)}(q) M(\theta\opt)^2.
\end{equation}
The rest of the proof proceeds conditionally on $\theta\opt$, so we are essentially in the fully frequentist setting. Thus, for the rest of the proof, we will abbreviate $M(\theta\opt)$ as $M$. We will also assume the ``worst-case" values for the non-null coordinates of $W_j$: in particular, let $W'$ denote $W$ but with all of the non-null coordinates replaced with the value $\infty$, and let $S' \subset [p]$ be the discovery set made when applying SeqStep to $W'$. These are the ``worst-case" values in the sense that $|S'| \ge |S|$ deterministically (see \cite{mrcknock2022}, Lemma B.4), so it suffices to show that $\E[|S'|] \le C(q) M$ and $\E[|S'|^2] \le C(q) M$. 

As notation, let $U = \I(\sorted(W') > 0)$ denote the signs of $W'$ when sorted in descending order of absolute value. Following the notation in Equation (\ref{eq::psidef}), let $\psi(U) = \max \left\{k : \frac{k - k \bar U_k + 1}{k \bar U_k} \le q\right\}$, where $\bar U_k = \frac{1}{k} \sum_{i=1}^{\min(k,p)} U_i$. This ensures that $|S'| = \ceil{\frac{\psi(U)+1}{1+q}} \le \psi(U)$ is the number of discoveries made by knockoffs (\cite{mrcknock2022}, Lemma B.3). To prove the first result, it thus suffices to show $\E[\psi(U)] \le C(q) M$. To do this, let $K = \ceil{\frac{M+1}{1+q}}$ and fix any integer $c \in \N$ (we will pick a specific value for $c$ later). Observe that
\begin{align}
        \E[\psi(U)] 
    &\le 
        cK \P(\psi(U) \le c K) + \sum_{k=cK}^{\infty} k \P(\psi(U) = k) \\
    &\le
        cK + \sum_{k=cK}^{\infty} k \P\left(\Bin(k-M, 1/2) \ge \ceil{\frac{k+1}{1+q}}-M\right). \label{eq::bound_moments_of_psi}
\end{align}
where the second line follows because (i) the event $\psi(U) = k$ implies that at least $\ceil{\frac{k+1}{1+q}}$ of the first $k$ coordinates of $U$ are positive and (ii) the knockoff flip-sign property guarantees that conditional on $\theta\opt$, the null coordinates of $U_j$ are i.i.d. random signs conditional on the values of the non-null coordinates of $U$.\footnote{Without loss of generality we may assume that the absolute values of $W'$ are nonzero with probability one, since again, this only increases the number of discoveries made by knockoffs.} Thus, doing simple arithmetic, in the first $k$ coordinates of $U$, there are $k-M$ null i.i.d. signs, of which at least $\ceil{\frac{k+1}{1+q}} - M$ must be positive, yielding the expression above with the Binomial probability.

We now apply Hoeffding's inequality. To do so, we must ensure $\ceil{\frac{k+1}{1+q}} - M$ is larger than the mean of a $\Bin(k-M, 1/2)$ random variable. It turns out that it suffices to pick the value of $c$ to satisfy $c > \left(\frac{1}{1+q} - \frac{1}{2}\right)^{-1}$. To see why, fix any $k \ge cK$, so we may write $k = cK + \ell$ for some $\ell \ge 0$. Then for all such $k$, we have
\begin{align*}
        \frac{k+1}{1+q} - M - \frac{k-M}{2}
    &\ge 
        k \left(\frac{1}{1+q} - \frac{1}{2} \right) - \frac{M}{2} \\
    &= 
        (c K + \ell) \left(\frac{1}{1+q} - \frac{1}{2} \right) - \frac{M}{2} & \text{ since } k = cK + \ell \\
    &\ge 
        \frac{2 K - M}{2} + \ell \left(\frac{1}{1+q} - \frac{1}{2} \right) & \text{ since } c \ge \left(\frac{1}{1+q} - \frac{1}{2} \right)^{-1} \\
    &\ge 
        \ell \left(\frac{1}{1+q} - \frac{1}{2} \right) & \text{ since } K \ge \frac{M}{1+q} \ge \frac{M}{2} \text{ by definition.}
\end{align*}
Thus, we may apply Hoeffding's inequality for $k \ge cK$. Indeed, for any $\ell \ge 0$, the previous result  yields that for $k \ge cK$:
\begin{align*}
        \P\left(\Bin(k-M, 1/2) \ge \ceil{\frac{k+1}{1+q}}-M\right) 
    &\le
       \P\left(\Bin(k-M, 1/2) - \frac{k-M}{2} \ge (k - cK) \left(\frac{1}{1+q} - \frac{1}{2} \right)  \right) \\
    &\le 
        \exp\left(-2 (k - cK)^2 \left(\frac{1}{1+q} - \frac{1}{2} \right)^2 \right).
\end{align*}
As notation, set $\alpha_q = \frac{1}{1+q} - \frac{1}{2}$. Combining the previous equation with Eq. (\ref{eq::bound_moments_of_psi}), we obtain
\begin{align*}
        \E[\psi(U)] 
    &\le
        cK + \sum_{k=cK}^{\infty} k \exp\left(- 2 (k - cK)^2 \alpha_q^2 \right) \\
    &=
        cK + \sum_{\ell=0}^{\infty} (\ell + cK) \exp\left(-2\ell^2 \alpha_q^2\right) \\
    &=
        cK + cK \sum_{\ell=0}^{\infty} \exp(-2 \ell^2 \alpha_q^2) + \sum_{\ell=0}^{\infty} \ell \exp(-2 \ell^2 \alpha_q^2).
\end{align*}
Note that the sums $\sum_{\ell=0}^{\infty} \ell \exp(-2\ell^2\alpha_q^2)$ and $\sum_{\ell=0}^{\infty} \exp(-2\ell^2 \alpha_q^2)$ are both convergent. As a result, $\E[\psi(U)]$ is bounded by a constant multiple of $cK \sim \frac{c}{1+q} M$, where the constant depends on $q$ but nothing else. Since $\psi(U) \ge |S'| \ge |S|$ as previously argued, this completes the proof.

For the second statement, we note that by the same argument as above, we have that
\begin{align*}
        \E[\psi(U)^2] 
    &\le 
        (cK)^2 + \sum_{k=cK}^{\infty} k^2 \P\left(\Bin(k-M, 1/2) \ge \ceil{\frac{k+1}{1+q}}-M\right) \\
    &\le 
        (cK)^2 + \sum_{\ell=0}^{\infty} (\ell + cK)^2 \exp\left(-2\ell^2 \alpha_q^2\right) \\
    &=
        (cK)^2 \left(1 + \sum_{\ell=0}^{\infty} \exp\left(-2\ell^2 \alpha_q^2\right)\right) + cK \sum_{\ell=0}^{\infty} \ell \exp\left(-2\ell^2 \alpha_q^2\right) + \sum_{\ell=0}^{\infty} \ell^2 \exp\left(-2\ell^2 \alpha_q^2\right). 
\end{align*}
Once again, the three series above are convergent and the value they converge to depends only on $q$. Since $cK \sim M$ asymptotically, this implies that there exists some constant $C^{(2)}(q)$ such that $\E[\psi(U)^2] \le C^{(2)}(q) M^2$. This completes the proof since $\psi(U)^2 \ge |S|^2$.
\end{proof}
\end{lemma}

\section{Additional comparison to prior work}\label{appendix::priorwork}

\subsection{Comparison to the unmasked likelihood ratio}\label{appendix::unmasked_lr}

In this section, we compare MLR statistics to the earlier \textit{unmasked} likelihood statistic introduced by \cite{katsevichmx2020}, which this work builds upon. The upshot is that unmasked likelihood statistics give the most powerful ``binary $p$-values,'' as shown by \cite{katsevichmx2020}, but do not yield jointly valid knockoff feature statistics in the sense required for the FDR control proof in \cite{fxknock} and \cite{mxknockoffs2018}.

In particular, we call a statistic $T_j([\bX, \tbX], \bigY)$ a \textit{marginally symmetric knockoff statistic} if $T_j$ satisfies $T_j([\bX, \tbX]_{\swap{j}}, \bigY) = - T_j([\bX, \tbX], \bigY)$. Under the null, $T_j$ is marginally symmetric, so the quantity $k_j = \frac{1}{2} + \frac{1}{2} \I(T_j \le 0)$ is a valid \textit{``binary $p$-value"} which only takes values in $\{1/2, 1\}$. Theorem 5 of  \cite{katsevichmx2020} shows that for any marginally symmetric knockoff statistic, $P\opt(k_j = 1/2) = P\opt(T_j > 0)$ is maximized if $T_j > 0 \Leftrightarrow p\opt_{\bigY \mid \bX}(\bigY \mid [\bX_j, \bX_{-j}]) > p\opt_{\bigY \mid \bX}(\bigY \mid [\tbX_j, \bX_{-j}])$, where $p\opt_{\bigY \mid \bX}$ denotes the density of $\bigY \mid \bX$ under the true law $P\opt$ of the data. As such, one might initially hope to use the unmasked likelihood ratio as a knockoff statistic:
\begin{equation*}
    W_j^{\mathrm{unmasked}} = \log\left(\frac{p\opt_{\bigY \mid \bX}(\bigY \mid [\bX_j, \bX_{-j}])}{p\opt_{\bigY \mid \bX}(\bigY \mid [\tbX_j, \bX_{-j}])}\right).
\end{equation*}
However, a marginally symmetric knockoff statistic is not necessarily a valid knockoff feature statistic, which must satisfy the following stronger property \citep{fxknock, mxknockoffs2018}:
\begin{equation*}
    W_j([\bX, \tbX]_{\swap{J}}, \bigY) = \begin{cases}
       W_j([\bX, \tbX], \bigY) & j \not \in J \\
       - W_j([\bX, \tbX], \bigY) & j \in J,
    \end{cases}
\end{equation*}
for any $J \subset [p]$. This flip-sign property guarantees that the signs of the null coordinates of $W$ are \textit{jointly} i.i.d. and symmetric. However, the unmasked likelihood statistic does not satisfy this property, as changing the observed value of $\bX_i$ for $i \ne j$ will typically change the value of the likelihood $p\opt_{\bigY \mid \bX}(\bigY \mid [\bX_j, \bX_{-j}])$.

\subsection{Comparison to the adaptive knockoff filter}\label{appendix::adaknock}

In this section, we compare our methodological contribution, MLR statistics, to the adaptive knockoff filter described in \cite{RenAdaptiveKnockoffs2020}, namely their approach based on Bayesian modeling. The main point is that although MLR statistics and the procedure from \cite{RenAdaptiveKnockoffs2020} have some intuitive similarities, the procedures are different and in fact complementary, since one could use the Bayesian adaptive knockoff filter from \cite{RenAdaptiveKnockoffs2020} in combination with MLR statistics.

As review, recall from Section \ref{subsec::mlr} that valid knockoff feature statistics $W$ as initially defined by \cite{fxknock, mxknockoffs2018} must ensure that $|W|$ is a function of the masked data $D$, and thus $|W|$ cannot explicitly depend on $\sign(W)$. (It is also important to remember that $|W|$ determines the order and ``prioritization" of the SeqStep hypothesis testing procedure.) The key innovation of \cite{RenAdaptiveKnockoffs2020} is to relax this restriction: in particular, they define a procedure where the analyst sequentially reveals the signs of $\sign(W)$ in reverse order of their prioritization, and after each sign is revealed, the analyst may arbitrarily reorder the remaining hypotheses. The advantage of this approach is that revealing the sign of (e.g.) $W_1$ may reveal information that can be used to more accurately prioritize the hypotheses while still guaranteeing provable FDR control.

This raises the question: how should the analyst reorder the hypotheses after each coordinate of $\sign(W)$ is revealed? One proposal from \cite{RenAdaptiveKnockoffs2020} is to introduce an auxiliary Bayesian model for the relationship between $\sign(W)$ and $|W_j|$ (the authors also discuss the use of additional side information, although for brevity we do not discuss this here). For example, \cite{RenAdaptiveKnockoffs2020} suggest using a two-groups model where
\begin{equation}\label{eq::adaknockmodel}
    H_j \simind \Bern(k_j) \text{ and } W_j \mid H_j \sim \begin{cases} \mcP_1(W_j) & H_j = 1 \\ \mcP_0(W_j) & H_j = 0. \end{cases}
\end{equation}
Above, $H_j$ is the indicator of whether the $j$th hypothesis is non-null, and $\mcP_1$ and $\mcP_0$ are (e.g.) unknown parametric distributions that the analyst fits as they observe $\sign(W)$. With this notation, the proposal from \cite{RenAdaptiveKnockoffs2020} can be roughly summarized as follows:
\begin{enumerate}
    \item Fit an initial feature statistic $W$, such as an LCD statistic, and observe $|W|$.
    \item Fit an initial version of the model in Equation (\ref{eq::adaknockmodel}) and use it to compute $\gamma_j \defeq \P(W_j > 0, H_j = 1 \mid |W_j|)$.
    \item Observe $\sign(W_j)$ for $j = \argmin_j \{\gamma_j : \sign(W_j) \text{ has not yet been observed }\}$.
    \item Using $\sign(W_j)$, update the model in Equation (\ref{eq::adaknockmodel}), update $\{\gamma_j\}_{j=1}^p$, and return to Step 3.
    \item Terminate when all of $\sign(W)$ has been revealed, at which point $\sign(W)$ is passed to SeqStep in the reverse of the order that the signs were revealed.
\end{enumerate}

Note that in Step 3, the reason \cite{RenAdaptiveKnockoffs2020} choose $j$ to be the index minimizing $\P(W_j > 0, H_j = 1 \mid |W_j|)$ is that in Step 5, SeqStep observes $\sign(W)$ in reverse order of the analyst. Thus, the analyst should observe the least important hypotheses first so that SeqStep can observe the most important hypotheses first.

The main similarity between this procedure and MLR statistics is that both procedures, roughly speaking, attempt to prioritize the hypotheses according to $\P(W_j > 0)$, although we condition on the full masked data to maximize power. That said, there are two important differences. First, we and \cite{RenAdaptiveKnockoffs2020} both use an auxiliary Bayesian model---however, we take probabilities over a Bayesian model of the full dataset, whereas \cite{RenAdaptiveKnockoffs2020} only fit a working model of the law of $W$. Using the full data as opposed to only the statistics $W$ should lead to much higher power---for example, if $W$ are poor feature statistics which do not contain much relevant information about the dataset, the procedure from \cite{RenAdaptiveKnockoffs2020} will have low power. Thus, despite their initial similarity, these procedures are quite different.

The second and more important difference is that the procedure above is not a feature statistic. Rather, it is an extension of SeqStep that wraps on top of any initial feature statistic. This ``adaptive knockoffs" procedure augments the power of any feature statistic, although if the initial feature statistic $W$ has many negative signs to begin with or its absolute values $|W|$ are truly uninformative of its signs, the procedure may still be powerless. Since MLR statistics have provable optimality guarantees---namely, they maximize $P\bayes(W_j > 0 \mid D)$ and make $|W_j|$ a monotone function of $P\bayes(W_j > 0 \mid D)$---one might expect that using MLR statistics in place of a lasso statistic could improve the power of the adaptive knockoff filter. Similarly, using the adaptive knockoff filter in combination with MLR statistics could be more powerful than using MLR statistics alone.

\section{Gibbs sampling for MLR statistics}\label{appendix::gibbs}

\subsection{Proof of Eq. (\ref{eq::resamplexj})}\label{appendix::gibbs_sampler_result}

\begin{lemma}\label{lem::gibbs_sampler_result}
Fix any constants $\bx, \tbx \in \R^{n \times p}, \smally \in \R^n$, and $\theta \in \Theta$, and define $\bd = (\mathbf{y}, \{\bx_j, \tbx_j\}_{j=1}^p)$. Then
\begin{align*}
        \frac{P\bayes(\bX_j = \bx_j \mid \bX_{-j} = \bx_{-j}, \theta\opt = \theta, D = \bd)}{P\bayes(\bX_j = \tbx_j \mid \bX_{-j} = \bx_{-j}, \theta\opt = \theta, D = \bd)} 
    &=
        \frac{P_{\bigY \mid \bX}\uptheta(\smally \mid \bX_j = \bx_j, \bX_{-j} = \bx_{-j})}{P_{\bigY \mid \bX}\uptheta(\smally \mid \bX_j = \tbx_j, \bX_{-j} = \bx_{-j})}.%\label{eq::resamplexj}.
\end{align*}
as long as $(\bx, \tbx, \smally, \theta)$ lies in the support of $(\bX, \tbX, \bigY, \theta\opt)$ under $P\bayes$.
\begin{proof} Throughout this proof, we abuse notation and let probabilities of the form (e.g.) $P\uppi(\bigY = \smally, \bX = \bx)$ denote the density of this event with respect to the base measure of $P\uppi$. 
The definition of conditional probability yields
\begin{align*}
    \frac{P\bayes(\bX_j = \bx_j \mid \bX_{-j} = \bx_{-j}, \theta\opt = \theta, D = \bd)}{P\bayes(\bX_j = \tbx_j \mid \bX_{-j} = \bx_{-j}, \theta\opt = \theta, D = \bd)} 
    = 
    \frac{P\bayes(\bX_j = \bx_j, \bX_{-j} = \bx_{-j}, \theta\opt = \theta, D = \bd)}{P\bayes(\bX_j = \tbx_j, \bX_{-j} = \bx_{-j}, \theta\opt = \theta, D = \bd)}.
    %\frac{P\bayes(\bX_j = \bx_j, \bX_{-j} = \bx_{-j}, \theta\opt = \theta)}{P\bayes(\bX_j = \tbx_j, \bX_{-j} = \bx_{-j}, \tbX_{-j} = \tbx_{-j} \theta\opt = \theta)} 
\end{align*}
By definition of $D = \bd$, the event in the numerator is equivalent to the event $[\bX, \tbX] = [\bx, \tbx], \bigY = \smally, \theta\opt=\theta$ and the event in the denominator is equivalent to $[\bX, \tbX] = [\bx, \tbx]_{\mathrm{swap}(j)}, \bigY = \smally, \theta\opt=\theta$. Plugging this in yields
\begin{align*}
    =& 
    \frac{P\bayes([\bX, \tbX] = [\bx, \tbx], \theta\opt = \theta, \bigY = \smally)}{P\bayes([\bX, \tbX] = [\bx, \tbx]_{\mathrm{swap}(j)}, \theta\opt = \theta, \bigY = \smally)} \\
    =&
    \frac{\pi(\theta) P\bayes([\bX, \tbX] = [\bx, \tbx] \mid \theta\opt = \theta) P\bayes(\bigY = \smally \mid \theta\opt = \theta, [\bX, \tbX] = [\bx, \tbx])}{\pi(\theta) P\bayes([\bX, \tbX] = [\bx, \tbx]_{\mathrm{swap}(j)} \mid \theta\opt = \theta) P\bayes(\bigY = \smally \mid \theta\opt = \theta, [\bX, \tbX] = [\bx, \tbx]_{\mathrm{swap}(j)})}.
    %\frac{P\bayes(\bX_j = \bx_j, \bX_{-j} = \bx_{-j}, \theta\opt = \theta)}{P\bayes(\bX_j = \tbx_j, \bX_{-j} = \bx_{-j}, \tbX_{-j} = \tbx_{-j} \theta\opt = \theta)} 
\end{align*}
where the second line uses the chain rule of conditional probability and the fact that $\theta\opt$ has marginal density $\pi$ under $P\bayes$. Recall that by definition of $P\bayes$ (see Section \ref{subsec::problemstatement}), the law the data given $\theta\opt = \theta$ is simply $P\uptheta$. Furthermore, for all $\theta \in \Theta$, under $P\uptheta$, $[\bX, \tbX]$ are assumed to be pairwise exchangeable because they are valid knockoffs (see footnote 3). Therefore, cancelling terms, we conclude
\begin{align*}
    =&
    \frac{P\uptheta(\bigY = \smally \mid [\bX, \tbX] = [\bx, \tbx])}{P\uptheta(\bigY = \smally \mid [\bX, \tbX] = [\bx, \tbx]_{\mathrm{swap}(j)})} \\
    &=
    \frac{P_{\bigY \mid \bX}\uptheta(\smally \mid \bX_j = \bx_j, \bX_{-j} = \bx_{-j})}{P_{\bigY \mid \bX}\uptheta(\smally \mid \bX_j = \tbx_j, \bX_{-j} = \bx_{-j})}.%\label{eq::resamplexj}.
\end{align*}
where the last line follows since
$\bigY \Perp \tbX \mid \bX$, as $\tbX$ are valid knockoffs by assumption under $P\uptheta$.
\end{proof}
\end{lemma}

\subsection{Derivation of Gibbs sampling updates}\label{subsec::gibbs}

In this section, we derive the Gibbs sampling updates for the class of MLR statistics defined in Section \ref{subsec::gams}. First, for convenience, we restate the model and choice of $\pi$. 

\subsubsection{Model and prior}

First, we consider the model-X case. For each $j \in [p]$, let $\phi_j(\bX_j) \in \R^{n \times k_j}$ denote any vector of prespecified basis functions applied to $\bX_j$. We assume the following additive model:
\begin{equation*}
    \bigY \mid \bX, \beta, \sigma^2 \sim \mcN\left(\sum_{j=1}^p \phi_j(\bX_j) \beta\upj, \sigma^2 I_n \right)
\end{equation*}
with the following prior on $\beta\upj \in \R^{k_j}$:
\begin{equation*}
    \beta\upj \simind \begin{cases}
        0 \in \R^{k_j} & \text{w.p. } p_0 \\ 
        \mcN(0, \tau^2 I_{k_j}) & \text{w.p. } 1-p_0.
    \end{cases}
\end{equation*}
with the usual hyperpriors
\begin{equation*}
     \tau^2 \sim \invGamma(a_{\tau}, b_{\tau}), \sigma^2 \sim \invGamma(a_{\sigma}, b_{\sigma}) \text{ and } p_0 \sim \Beta(a_0, b_0).
\end{equation*} 

This is effectively a \textit{group} spike-and-slab prior on $\beta\upj$ which ensures group sparsity of $\beta\upj$, meaning that either the whole vector equals zero or the whole vector is nonzero. We use this group spike-and-slab prior for two reasons. First, it reflects the intuition that $\phi_j$ is meant to represent only a single feature and thus $\beta\upj$ will likely be entirely sparse (if $\bX_j$ is truly null) or entirely non-sparse. Second, and more importantly, the group sparsity will substantially improve computational efficiency in the Gibbs sampler. 

Lastly, for the fixed-X case, we assume exactly the same model but with the basis functions $\phi_j(\cdot)$ chosen to be the identity. Thus, this model is a typical spike-and-slab Gaussian linear model in the fixed-X case \citep{mcculloch1997}. It is worth noting that our implementation for the fixed-X case actually uses a slightly more general Gaussian mixture model as the prior on $\beta_j$, where the density $p(\beta_j) = \sum_{k=1}^m p_k \mcN(\beta_j; 0, \tau_k^2)$ for hyperpriors $\tau_0 = 0, \tau_k \simind \invGamma(a_k, b_k)$, and $(p_0, \dots, p_m) \sim \Dir(\alpha)$. However, for brevity, we only derive the Gibbs updates for the case of two mixture components.

\subsubsection{Gibbs sampling updates}\label{subsubsec::gibbs}

Following Section \ref{subsec::computation}, we now review the details of the MLR Gibbs sampler which samples from the posterior of $(\bX, \beta)$ given the masked data $D = \{\bigY, \{\bx_j, \tbx_j\}_{j=1}^p\}$.\footnote{This is a standard derivation, but we review it here for the reader's convenience.} As notation, let $\beta$ denote the concatenation of $\{\beta\upj\}_{j=1}^p$, let $\beta\negupj$ denote all of the coordinates of $\beta$ except those of $\beta\upj$, let $\gamma_j$ denote the indicator that $\beta\upj \ne 0$, and let $\phi(\bX) \in \R^{n \times \sum_j k_j}$ denote all of the basis functions concatenated together. Also note that although this section mostly uses the language of model-X knockoffs, when the basis functions $\phi_j(\cdot)$ are the identity, the Gibbs updates we are about to describe satisfy the sufficiency property required for fixed-X statistics, and indeed the resulting Gibbs sampler is actually a valid implementation of the fixed-X MLR statistic.

To improve the convergence of the Gibbs sampler, we slightly modify the meta-algorithm in Algorithm \ref{alg::mlrgibbs} to marginalize over the value of $\beta\upj$ when resampling $\bX_j$. To be precise, this means that instead of sampling $\bX_j \mid \bX_{-j}, \beta, \sigma^2$, we sample $\bX_j \mid \bX_{-j}, \beta\upnegj$. We derive this update in three steps, and along the way we derive the update for $\beta\upj \mid \bX, \beta\negupj, D$.

\underline{Step 1}: First, we derive the update for $\gamma_j \mid \bX, \beta\upnegj, D$. Observe
\begin{align*}
        \frac{\P(\gamma_j = 0 \mid \bX, \beta\negupj, D)}{\P(\gamma_j = 1 \mid \bX, \beta\negupj, D)}
    &=
        \frac{p_0 p(\bigY \mid \bX, \beta\upnegj, \beta\upj = 0)}{(1-p_0) p(\bigY \mid \bX, \beta\upnegj, \beta\upj \ne 0)}.
\end{align*}
Analyzing the numerator is easy, as the model specifies that if we let $\br = \bigY - \phi(\bX_{-j}) \beta\negupj$, then
\begin{equation*}
    p(\bigY \mid \bX, \beta\upnegj, \beta\upj = 0) \propto \det(\sigma^2 I_n)^{-1/2} \exp\left(-\frac{1}{2 \sigma^2} \|\br\|_2^2 \right).
\end{equation*}
For the denominator, observe that $\br, \beta\upj \mid \bX, \beta\negupj, \beta\upj \ne 0$ is jointly Gaussian: in particular,
\begin{equation}\label{eq::joint_r_betaj}
    (\beta\upj, \br) \mid \bX, \beta\negupj, \beta\upj \ne 0 \sim \mcN\left(0, \begin{bmatrix} \tau^2 I_{k_j} & \tau \phi_j(\bX_j)^T \\ \tau \phi_j(\bX_j) & \tau^2 \phi_j(\bX_j) \phi_j(\bX_j)^T + \sigma^2 I_n \end{bmatrix} \right).
\end{equation}
To lighten notation, let $Q_j \defeq I_{k_j} + \frac{\tau^2}{\sigma^2} \phi(\bX_j)^T \phi(\bX_j)$. Using the above expression plus the Woodbury identity applied to the density of $\bigY \mid \bX, \beta\negupj, \beta\upj \ne 0$, we conclude
\begin{equation*}
    \frac{\P(\gamma_j = 0 \mid \bX, \beta\negupj, D)}{\P(\gamma_j = 1 \mid \bX, \beta\negupj, D)} = \frac{p_0}{1-p_0} \det(Q_j)^{1/2} \exp\left(- \frac{\tau^2}{2 \sigma^4} \br^T \phi_j(\bX_j) Q_j^{-1} \phi_j(\bX_j)^T \br \right).
\end{equation*}
Since $Q_j$ is a $k_j \times k_j$ matrix, this quantity can be computed relatively efficiently.

\underline{Step 2}: Next, we derive the distribution of $\beta\upj \mid \bigY, \bX, \beta\negupj, \gamma_j$. Of course, the case where $\gamma_j = 0$ is trivial since then $\beta\upj = 0$ by definition: in the alternative case, note from Equation (\ref{eq::joint_r_betaj}) that we have
\begin{equation*}
    \beta\upj \mid \bigY, \bX, \beta\negupj, \gamma_j = 1 \sim \mcN\left(\frac{\tau^2}{\sigma^2} \phi_j^T \br - \frac{\tau^4}{\sigma^4} \phi_j^T \phi_j Q_j^{-1} \phi_j^T \br, \tau^2 I_{k_j} - \frac{\tau^4}{\sigma^2} \phi_j^T \phi_j + \frac{\tau^6}{\sigma^4} \phi_j^T \phi_j Q_j^{-1} \phi_j^T \phi_j  \right),
\end{equation*}
where above, we use $\phi_j$ as shorthand for $\phi_j(\bX_j)$ to lighten notation.

\underline{Step 3}: Lastly, we derive the update for $\bX_j$ given $\bX_{-j}, \beta\negupj, D$. In particular, for any vector $\bx$, let $\kappa(\bx) \defeq \P(\gamma = 0 \mid \bX_j = \bx, \bX_{-j}, \beta\negupj)$. Then by the law of total probability and the same Woodbury calculations as before,
\begin{align*}
        \P(\bX_j = \bx \mid \bX_{-j}, \beta\negupj, D) 
    \propto& 
        p(\bigY \mid \bX_j = \bx, \bX_{-j}, \beta\upnegj) \\
    &=
        \kappa(\bx) p(\bigY \mid \bX_j = \bx, \bX_{-j}, \beta\upnegj, \beta\upj = 0) \\
    &+ 
        (1-\kappa(\bx)) p(\bigY \mid \bX_j = \bx, \bX_{-j}, \beta\upnegj, \beta\upj \ne 0) \\
    &\propto
        \kappa(\bx) \exp\left(-\frac{1}{2 \sigma^2} \|\br\|_2^2 \right) \\
    &+ 
        (1-\kappa(\bx)) \det(Q_j(\bx))^{-1/2} \exp\left(-\frac{1}{2 \sigma^2} \|\br\|_2^2 + \frac{\tau^2}{2\sigma^4} \br^T \phi_j(\bx)^T Q_j(\bx)^{-1} \phi_j(\bx)^T \br \right) \\
    &\propto 
        \kappa(\bx) + (1-\kappa(\bx)) \det(Q_j(\bx))^{-1/2} \exp\left(\frac{\tau^2}{2\sigma^4} \br^T \phi_j(\bx)^T Q_j(\bx)^{-1} \phi_j(\bx)^T \br\right)
\end{align*}
where above $Q_j(\bx) = I_{k_j} + \frac{\tau^2}{\sigma^2} \phi_j(\bx)^T \phi_j(\bx)$ as before.

The only other sampling steps required in the Gibbs sampler are to sample from the conditional distributions of $\sigma^2, \tau^2$ and $p_0$; however, this is straightforward since we use conjugate hyperpriors for each of these parameters.

\subsubsection{Extension to binary regression}

We can easily extend the Gibbs sampler in the preceding section to handle the case where the response is binary via a latent variable approach. Indeed, let us start by considering the case of Probit regression, which means we observe $\bz = \I(\bigY \ge 0) \in \{0,1\}^n$ instead of the continuous outcome $\bigY$. Following \cite{albertchib1993}, we note that distribution of $\bigY \mid \bz, \bX, \beta$ is truncated normal, namely
\begin{equation}\label{eq::dataaugment}
    \bigY_i \mid \bz, \bX, \beta \simind \begin{cases}
        \TruncNorm(\mu_i, \sigma^2; (0,\infty)) & \bz_i = 1 \\
        \TruncNorm(\mu_i, \sigma^2; (-\infty, 0) & \bz_i = 0,
    \end{cases}
\end{equation}
where $\mu = \phi(\bX) \beta = \E[\bigY \mid \bX, \beta]$. Thus, when we observe a binary response $\bz$ instead of the continuous response $\bigY$, we can employ the same Gibbs sampler as in Section \ref{subsubsec::gibbs} except that after updating $\beta\upj \mid \bX, \beta\negupj, \bigY$, we resample the latent variables $\bigY$ according to Equation (\ref{eq::dataaugment}), which takes $O(n)$ computation per iteration (since we can continuously update the value of $\mu$ whenever we update $\bX$ or $\beta$ in $O(n)$ operations as well). As a result, the computational complexity of this algorithm is the same as that of the algorithm in Section \ref{subsubsec::gibbs}.  A similar formulation based on PolyGamma random variables is available for the case of logistic regression (see \cite{polygamma2013}).

\subsection{Proof and discussion of Lemma \ref{lem::gibbs}}\label{subsec::gibbsproof}

\begingroup
\def\thetheorem{\ref{lem::gibbs}}
\begin{lemma} Suppose that under $P\bayes$, (i) $p_j^{(i)} \in (0,1)$ a.s. for $j \in [p]$ and (ii) the support of $\theta\opt \mid \bX, \bigY$ equals the support of the marginal law of $\theta\opt$. Then as $n_\sample \to \infty$,
\begin{equation*}
        \log\left(\sum_{i=1}^{n_\sample} p_j^{(i)}\right) 
        - \log\left(\sum_{i=1}^{n_\sample} 1 - p_j^{(i)}\right)
    \toprob
        \MLR_j\uppi
    \defeq 
        \log\left(\frac{P_j\bayes(\bX_j \mid D)}{P_j\bayes(\tbX_j \mid D)}\right).
\end{equation*}
\begin{proof} This proof follows because by the derivations in Section \ref{subsec::computation}, Algorithm \ref{alg::mlrgibbs} is a standard Gibbs sampler as defined in Algorithm A.40 of \cite{robertcasella2004mcmc}, i.e., at each step it samples from the conditional law of one unknown variable given all the others (where the unknown variables are $\bX_1, \dots, \bX_p$ and $\theta\opt$), where everything is done conditional on $D$.
Furthermore, the condition (i) implies that the support of $\bX_j \mid D, \bX_{-j}, \theta\opt$ does not depend on $(\theta\opt, \bX_{-j})$. As a result, this theorem result is a direct consequence of Corollary 10.12 of \cite{robertcasella2004mcmc}, applied conditional on $D$. In particular, Corollary 10.12 proves that
\begin{equation*}
    \frac{1}{n_\sample} \sum_{i=1}^{n_\sample} p_j^{(i)} \toprob P_j\bayes(\bX_j \mid D)
\end{equation*}
The result then follows from the continuous mapping theorem (note that the first assumption for the lemma ensures $P_j\bayes(\bX_j \mid D) \in (0,1)$ so the continuous mapping theorem applies).
\end{proof}
\end{lemma}

We note also that the two assumptions of Lemma \ref{lem::gibbs} are satisfied in Example \ref{ex::sparsegam}. In particular, to show that the support of $\bX_j \mid D, \theta\opt, \bX_{-j}$ does not depend on $\theta\opt$, observe that Eq. (\ref{eq::resamplexj}) tells us that conditional on $D = \bd$ for $\bd = (\smally, \{\bx_j, \tbx_j\}_{j=1}^p)$, and for any $\theta \in \Theta$,
\begin{equation*}
    p_j^{(i)} = \frac{P\bayes(\bX_j = \bx_j \mid D = \bd, \theta\opt=\theta, \bX_{-j})}{P\bayes(\bX_j = \tbx_j \mid D = \bd, \theta\opt=\theta, \bX_{-j})} = \frac{P^{(\theta)}_{\bigY \mid \bX}(\smally \mid \bX_j = \bx_j, \bX_{-j})}{P^{(\theta)}_{\smally \mid \bX}(\smally \mid \bX_j = \tbx_j, \bX_{-j})}.
\end{equation*}
In Example \ref{ex::sparsegam}, the numerator and denominator of the above equation are Gaussian likelihoods, so the likelihood ratio is always finite and thus $p_j^{(i)} \in (0, 1)$. Similarly, Example \ref{ex::sparsegam} is a conjugate Gaussian (additive) spike-and-slab linear model. Well--known results for these models establish that the support of the Gibbs distributions for the linear coefficients $\beta$ and hyperparameters $\sigma^2, \tau^2, p_0$ is equal to the support of the marginal distribution for these parameters \citep{mcculloch1997}---see Appendix \ref{subsec::gibbs} for examples of detailed derivations showing this result.

There may, of course, be other Bayesian models $P\bayes$ where the assumptions of Lemma \ref{lem::gibbs} do not hold. In these settings, the assumptions in Lemma \ref{lem::gibbs} can be relaxed---see \cite{robertcasella2004mcmc}.

\subsection{Computing AMLR statistics}\label{subappendix::computing_amlr}

The AMLR statistics are a deterministic function of the MLR statistics $\{\MLR_j\uppi\}_{j=1}^p$ and $\{\nu_j\}_{j=1}^p$ where
\begin{equation*}
    \nu_j \defeq \frac{P\bayes(\MLR_j\uppi > 0, j \in \mcH_1(\theta\opt) \mid D)}{(1+q)^{-1} - P\bayes(\MLR_j\uppi > 0 \mid D)}.
\end{equation*}
By Proposition \ref{prop::intuit_opt},  $P\bayes(\MLR_j\uppi > 0 \mid D)$ is a function of $|\MLR_j\uppi|$, so computing the denominator of $\nu_j$ is straightforward since we have already established how to compute $\{\MLR_j\uppi\}_{j=1}^p$. To compute the numerator, recall that Algorithm \ref{alg::mlrgibbs} samples from the \textit{joint} posterior of $\theta\opt, \bX, \tbX \mid D$. Therefore, we can use the empirical mean of the samples from Algorithm \ref{alg::mlrgibbs} to compute this quantity:
\begin{equation*}
    P\bayes(\MLR_j\uppi > 0, j \in \mcH_1(\theta\opt) \mid D) \approx \frac{1}{n_{\sample}} \sum_{i=1}^{n_{\sample}} \I(\widehat{\bX}_j = \bX_j^{(i)}, j \in \mcH_1(\theta^{(i)}))
\end{equation*}
where $\widehat{\bX}_j = \argmax_{\bx \in \{\bX_j, \tbX_j\}} P_j(\bx \mid D)$ is the MLR ``guess" of the value of $\bX_j$.

\section{MLR statistics for group knockoffs}\label{appendix::groupmlr}

In this section, we describe how MLR statistics extend to the setting of group knockoffs \citep{daibarber2016}. In particular, for a partition $G_1, \dots, G_m \subset [p]$ of the features, group knockoffs allow analysts to test the \textit{group} null hypotheses $H_{G_j} : \bX_{G_j} \Perp \bigY \mid \bX_{-G_j}$, which can be useful in settings where $\bX$ is highly correlated and there is not enough data to discover individual null variables. % and it is clear that at least one feature in a group $G_j$ is non-null, but it is not clear which one is non-null due to the high correlations. 
In particular, knockoffs $\tbX$ are model-X \textit{group} knockoffs if they satisfy the \textit{group} pairwise-exchangeability condition $[\bX, \tbX]_{\swap{G_j}} \disteq [\bX, \tbX]$ for each $j \in [m]$. Similarly, $\tbX$ are fixed-X group knockoffs if (i) $\bX^T \bX = \tbX^T \tbX$ and (ii) $S = \bX^T \bX - \tbX^T \bX$ is block-diagonal, where the blocks correspond to groups $G_1, \dots, G_m$. Given group knockoffs, one computes a single knockoff feature statistic for each group.%, where the group knockoff feature statistics $W \in \R^m$

MLR statistics extend naturally to the group knockoff setting because we can treat each group of features $X_{G_j}$ as a single compound feature. In particular, the masked data for group knockoffs is 
\begin{equation}\label{eq::groupmaskeddata}
    D = \begin{cases}
        (\bigY, \{\bX_{G_j}, \tbX_{G_j}\}_{j=1}^m) & \text{ for model-X knockoffs} \\
        (\bX, \tbX, \{\bX_{G_j}^T \bigY, \tbX_{G_j}^T \bigY\}_{j=1}^m) & \text{ for fixed-X knockoffs,}
    \end{cases}
\end{equation}
and the corresponding MLR statistics are
\begin{equation*}
    \MLR_j\uppi = \log\left(\frac{P_{G_j}\bayes(\bX_{G_j} \mid D)}{P_{G_j}\bayes(\tbX_{G_j} \mid D)}\right) \text{ for model-X knockoffs,}
\end{equation*}
where $P_{G_j}\bayes$ above denotes the law of $\bX_{G_j} \mid D$ under $P\bayes$. For fixed-X knockoffs, we have
\begin{equation*}
    \MLR_j\uppi = \log\left(\frac{P\bayes_{G_j,\mathrm{fx}}(\bX_{G_j}^T \bigY \mid D)}{P\bayes_{G_j,\mathrm{fx}}(\tbX_{G_j}^T \bigY \mid D)}\right) \text{ for fixed-X knockoffs,}
\end{equation*}
where $P\bayes_{G_j,\mathrm{fx}}$ denotes the law of $\bX_{G_j}^T \bigY \mid D$ under $P\bayes$.

Throughout the paper, we have proved several optimality properties of MLR statistics, and if we treat $\bX_{G_j}$ as a single compound feature, all of these theoretical results (namely Proposition \ref{prop::intuit_opt} and Theorem \ref{thm::avgopt}) immediately apply to group MLR statistics as well.

To compute group MLR statistics, we can use exactly the same Gibbs sampling strategy as in Section \ref{subsec::gibbs}---indeed, one can just treat $X_{G_j}$ as a basis representation of a single compound feature and use exactly the same equations as derived previously. This method is implemented in \texttt{knockpy}.

\section{Additional details for the simulations}\label{appendix::simdetails}

\begin{figure}[!ht]
    \centering
    \includegraphics[width=\linewidth]{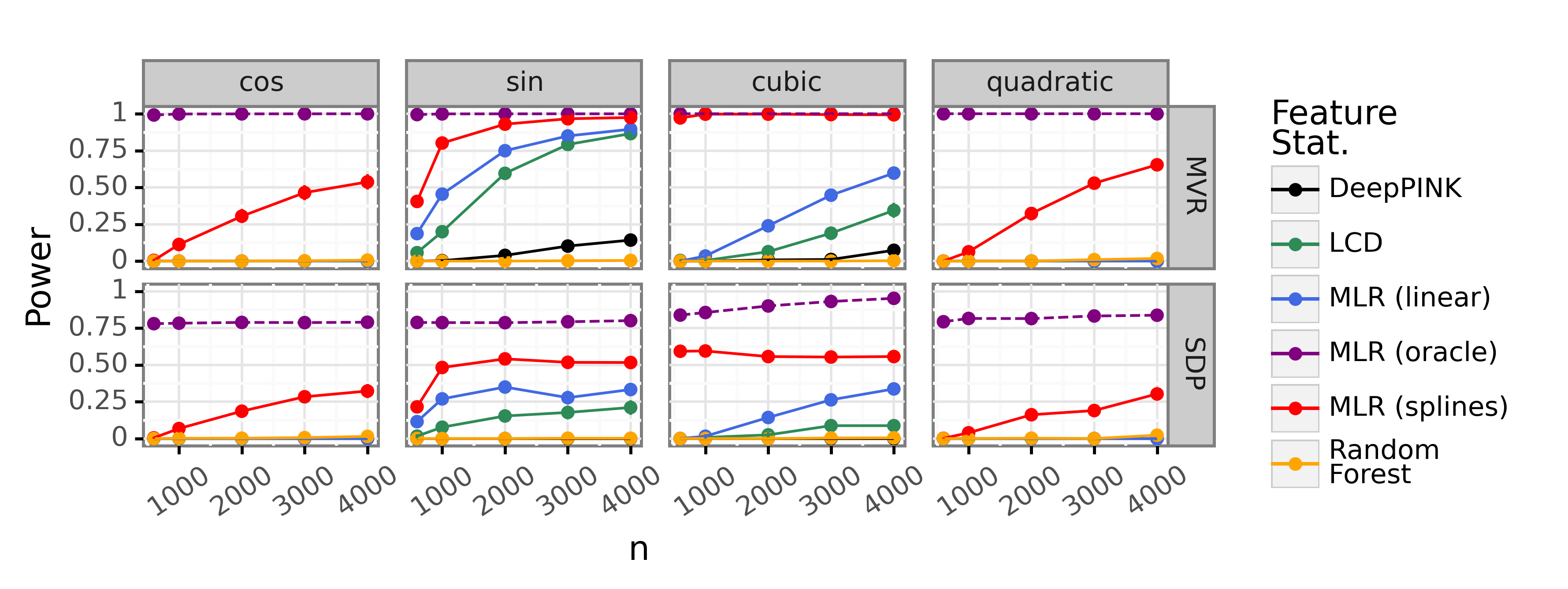}
    \caption{This plot is identical to Figure \ref{fig::nonlin} except it shows the results for $q=0.05$.}
    \label{fig::nonlin2}
\end{figure}

In this section, we describe the simulation settings in Section \ref{sec::sims}, and we also give the corresponding plot to Figure \ref{fig::nonlin} which shows the results when $q=0.05$. To start, we describe the simulation setting for each plot.

\begin{enumerate}
    \item \underline{Sampling $\bX$}: We sample each row of $\bX$ as an i.i.d. $\mcN(0, \Sigma)$ random vector in all simulations, with two choices of $\Sigma$. First, in the ``AR(1)" setting, we take $\Sigma$ to correspond to a nonstationary AR(1) Gaussian Markov chain, so $\bX$ has i.i.d. rows satisfying $X_j \mid X_1, \dots, X_{j-1} \sim \mcN(\rho_j X_{j-1}, 1)$ with $\rho_j \iid \min(0.99, \Beta(5,1))$. Note that the AR(1) setting is the default used in any plot where the covariance matrix is not specified. Second, in the ``ErdosRenyi" (ER) setting, we sample a random matrix $V$ such that $80\%$ of its off-diagonal entries (selected uniformly at random) are equal to zero; for the remaining entries, we sample $V_{ij} \iid \Unif((-1, -0.1) \cup (0.1, 1))$. To ensure the final covariance matrix is positive definite, we set $\Sigma = (V + V^T) + (0.1 + \lambda_{\min}(V^T + V)) I_p$ and then rescale $\Sigma$ to be a correlation matrix.
    
    \item \underline{Sampling $\beta$}: Unless otherwise specified in the plot, we randomly choose $s=10\%$ of the entries of $\beta$ to be nonzero and sample the nonzero entries as i.i.d. $\Unif([-\tau, -\tau/2] \cup [\tau/2, \tau])$ random variables with $\tau = 0.5$ by default. The exceptions are: (1) in Figure \ref{fig::amlr}, we set $\tau = 0.5$ and $s = 0.1$, (2) in Figure \ref{fig::misspec}, we set $\tau = 0.3$, vary $s$ between $0.05$ and $0.4$ as shown in the plot, and in some panels sample the non-null coefficients as $\Laplace(\tau)$ random variables, (3) in Figure \ref{fig::nonlin} we take $\tau = 2$ and $s=0.3$, (4) in Figure \ref{fig::logistic} we take $\tau = 1$.
    
    \item \underline{Sampling $\bigY$}: Throughout we sample $\bigY \mid \bX \sim \mcN(\bX \beta, I_n)$, with only two exceptions. First, in Figure \ref{fig::nonlin}, we sample $\bigY \mid \bX \sim \mcN(h(\bX) \beta, I_n)$ where $h$ is a nonlinear function applied elementwise to $\bX$, for $h(x)=\sin(x), h(x)=\cos(x), h(x)=x^2$ and $h(x)=x^3$. Second, in Figure \ref{fig::logistic}, $\bigY$ is binary and $\P(Y = 1 \mid X) = \frac{\exp(X^T \beta)}{1 + \exp(X^T \beta)}$.
    
    \item \underline{Sampling knockoffs}: We sample MVR and SDP Gaussian knockoffs using the default parameters from \texttt{knockpy} version 1.3, both in the fixed-X and model-X case. Note that in the model-X case, we use the true covariance matrix $\Sigma$ to sample knockoffs, thus guaranteeing finite-sample FDR control.
    
    \item \underline{Fitting feature statistics}: We fit the following types of feature statistics throughout the simulations: LCD statistics, LSM statistics, a random forest with swap importances \citep{knockoffsmass2018}, DeepPINK \citep{deeppink2018}, MLR statistics (linear variant), MLR statistics with splines, and the MLR oracle. In all cases we use the default hyperparameters from \text{knockpy} version 1.3, and we do not adjust the hyperparameters, so that the MLR statistics do not have well-specified priors. The exception is that the MLR oracle has access to the underlying data-generating process and the true coefficients $\beta$, which is why it is an ``oracle."
\end{enumerate}

Now, recall that in Figure \ref{fig::nonlin}, we showed the results for $q=0.1$ because several competitor feature statistics made no discoveries at $q=0.05$. Figure \ref{fig::nonlin2} is corresponding plot for $q=0.05$.

\section{Additional results for the real data applications}\label{appendix::realdata}

\subsection{HIV drug resistance}\label{appendix::hiv}

For the HIV drug resistance application,  Figures \ref{fig::hiv_pi_sdp}-\ref{fig::hiv_nnrti_mvr} show the same results as in Figure \ref{fig::wstatplot} but for all drugs in the protease inhibitor (PI) class; broadly, they show that MLR statistics have higher power because they ensure that the feature statistics with high absolute values are consistently positive, as discussed in Section \ref{subsec::hiv_motiv}. Note that in these plots, for the lasso-based statistics, we plot the normalized statistics $\frac{W_j}{\max_i|W_i|}$ so that the absolute value of each statistic is less than one. Similarly, for the MLR statistics, instead of directly plotting the masked likelihood ratio as per Equation (\ref{eq::mlr}), we plot 
\begin{equation*}
    W_j^{\star\star} \defeq 2\left(\logit^{-1}(|\MLR_j\uppi|) - 0.5\right) = 2\left(\P(\MLR_j\uppi > 0 \mid D) - 0.5\right)
\end{equation*}
because we find this quantity easier to interpret than a log likelihood ratio. In particular, $W_j^{\star\star} \approx 0$ if and only if $\MLR_j\uppi$ is roughly equally likely to be positive or negative under $P\bayes$, and $W_j^{\star\star} = 1$ when $\MLR_j\uppi$ is always positive under $P\bayes$.% The results in Figures \ref{fig::hiv_pi_sdp}-\ref{fig::hiv_pi_mvr} show that the MLR statistics are roughly well calibrated, if slightly conservative}  

Additionally, Figures \ref{fig::hiv_sdp} and \ref{fig::hiv_mvr} show the number of discoveries made by each feature statistic for SDP and MVR knockoffs, respectively, stratified by the drug in question. Note that the specific data analysis is identical to that of \cite{fxknock} and \cite{dbh2020} other than the choice of feature statistic---see either of those papers or \url{https://github.com/amspector100/mlr_knockoff_paper} for more details.

\begin{figure}
    \centering
    \includegraphics{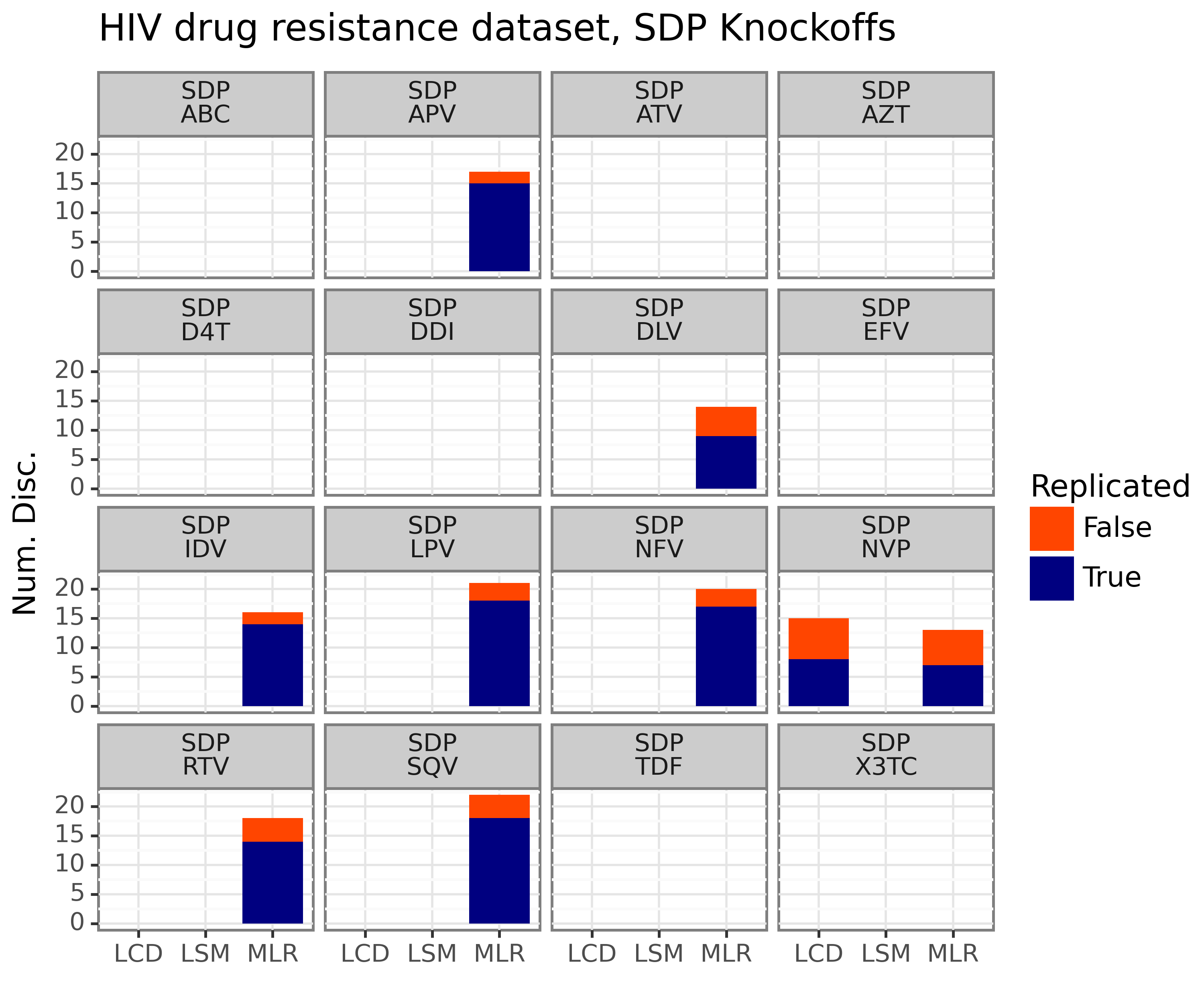}
    \caption{This figure shows the number of discoveries made by each feature statistic for each drug in the HIV drug resistance dataset.}
    \label{fig::hiv_sdp}
\end{figure}

\begin{figure}
    \centering
    \includegraphics{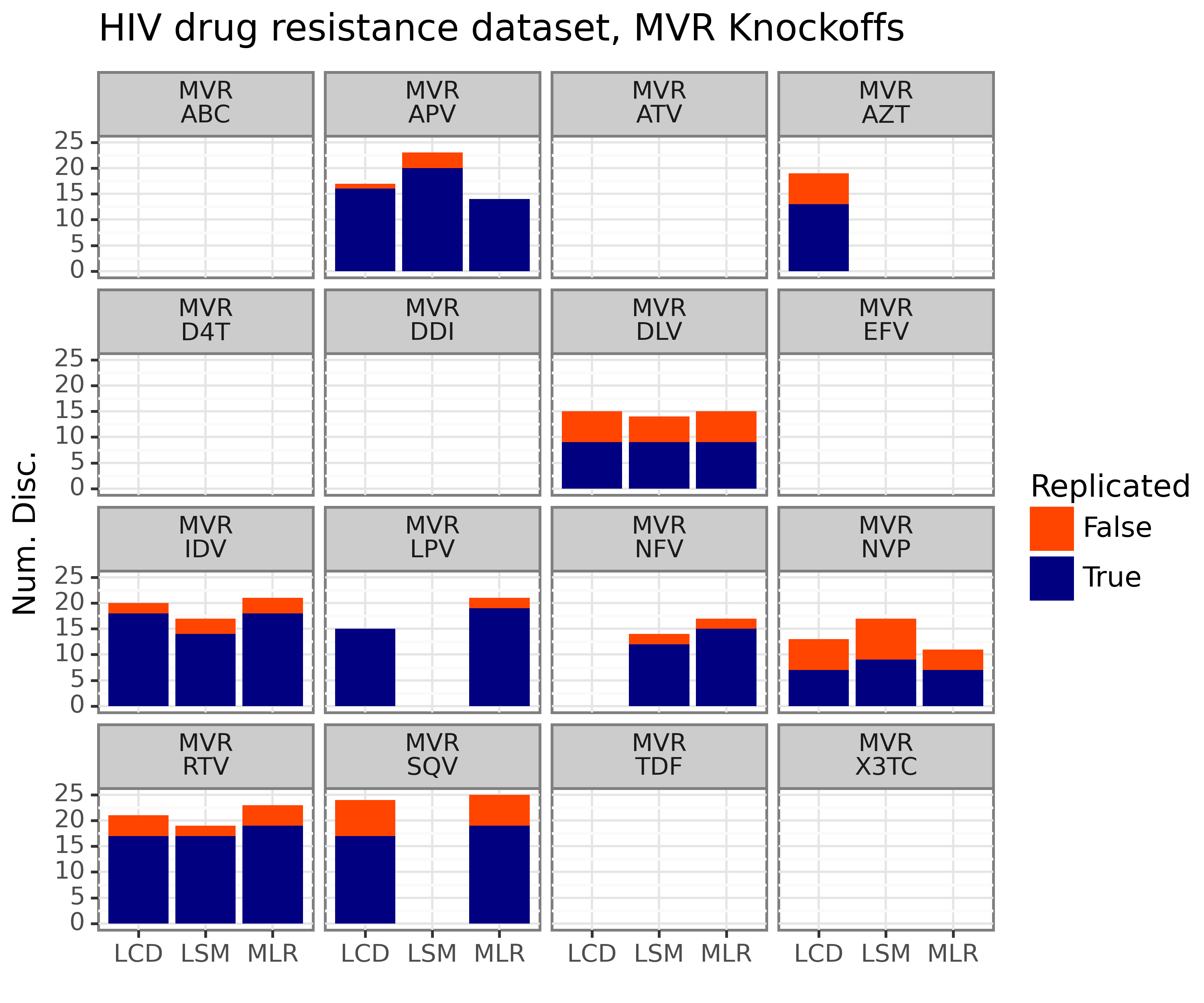}
    \caption{This figure shows the number of discoveries made by each feature statistic for each drug in the HIV drug resistance dataset.}
    \label{fig::hiv_mvr}
\end{figure}

\begin{figure}
    \centering
    \includegraphics{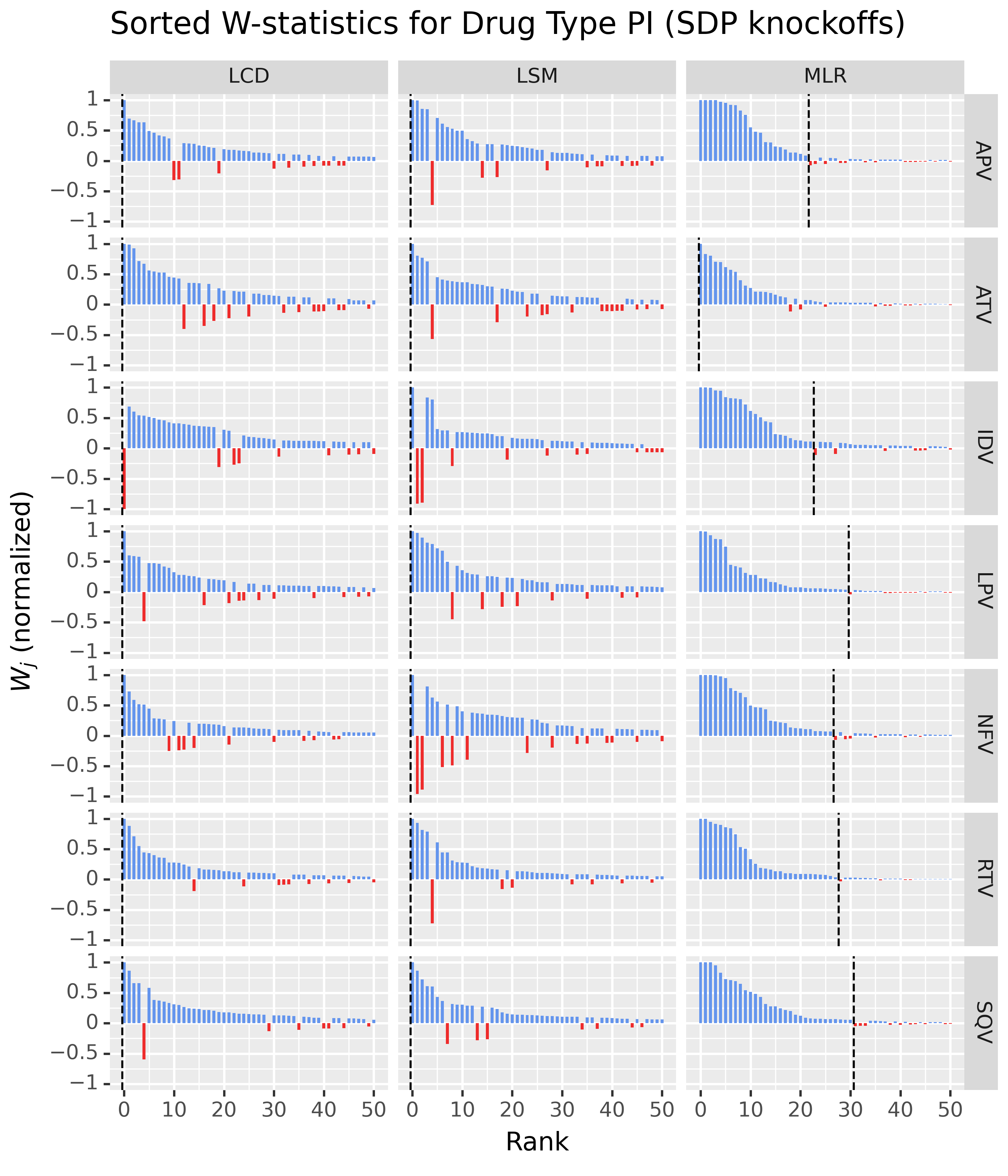}
    \caption{This figure is the same as Figure \ref{fig::wstatplot}, except it shows results for drugs in the PI class using SDP knockoffs.}
    \label{fig::hiv_pi_sdp}
\end{figure}

\begin{figure}
    \centering
    \includegraphics{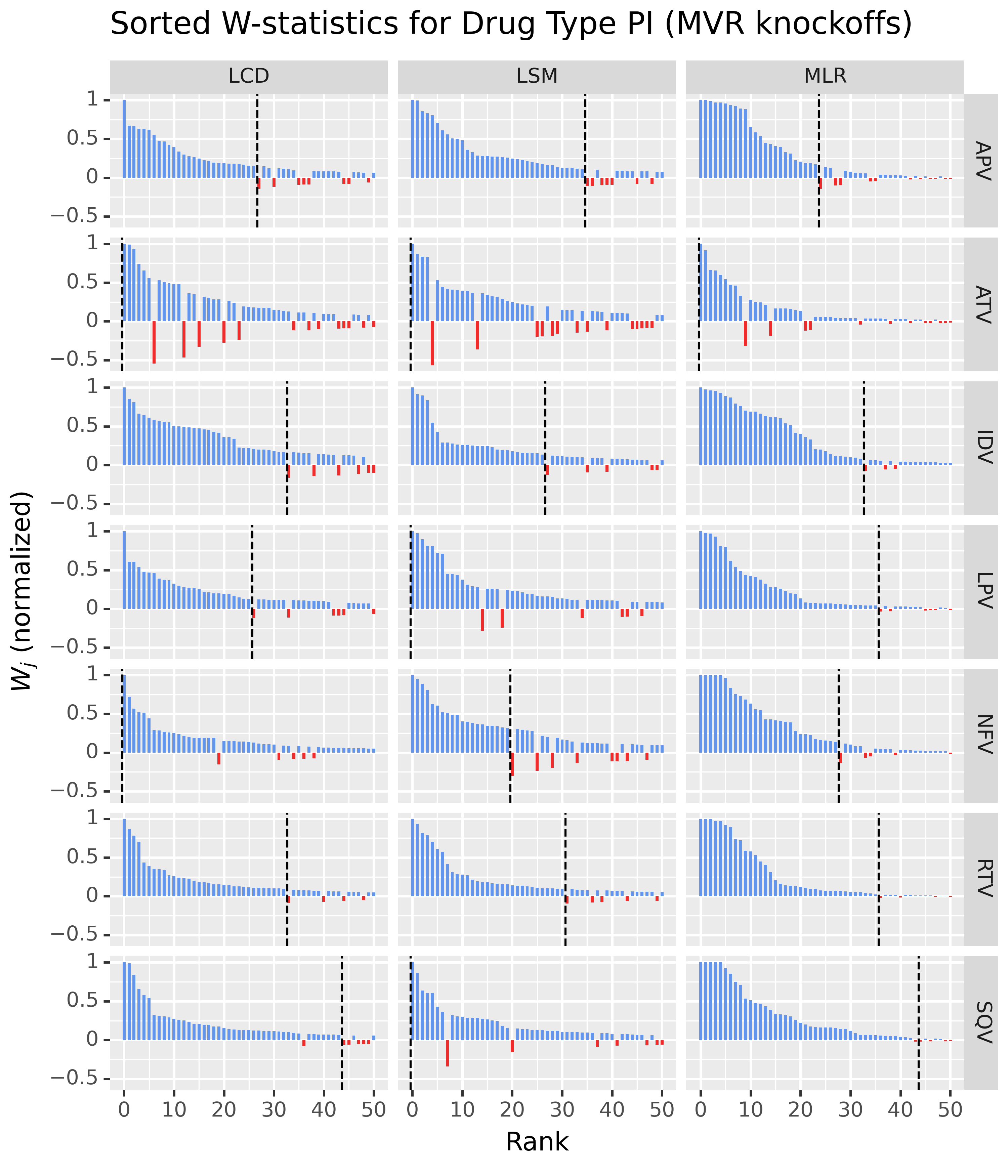}
    \caption{This figure is the same as Figure \ref{fig::wstatplot}, except it shows results for drugs in the PI class using MVR knockoffs.}
    \label{fig::hiv_pi_mvr}
\end{figure}

\begin{figure}
    \centering
    \includegraphics{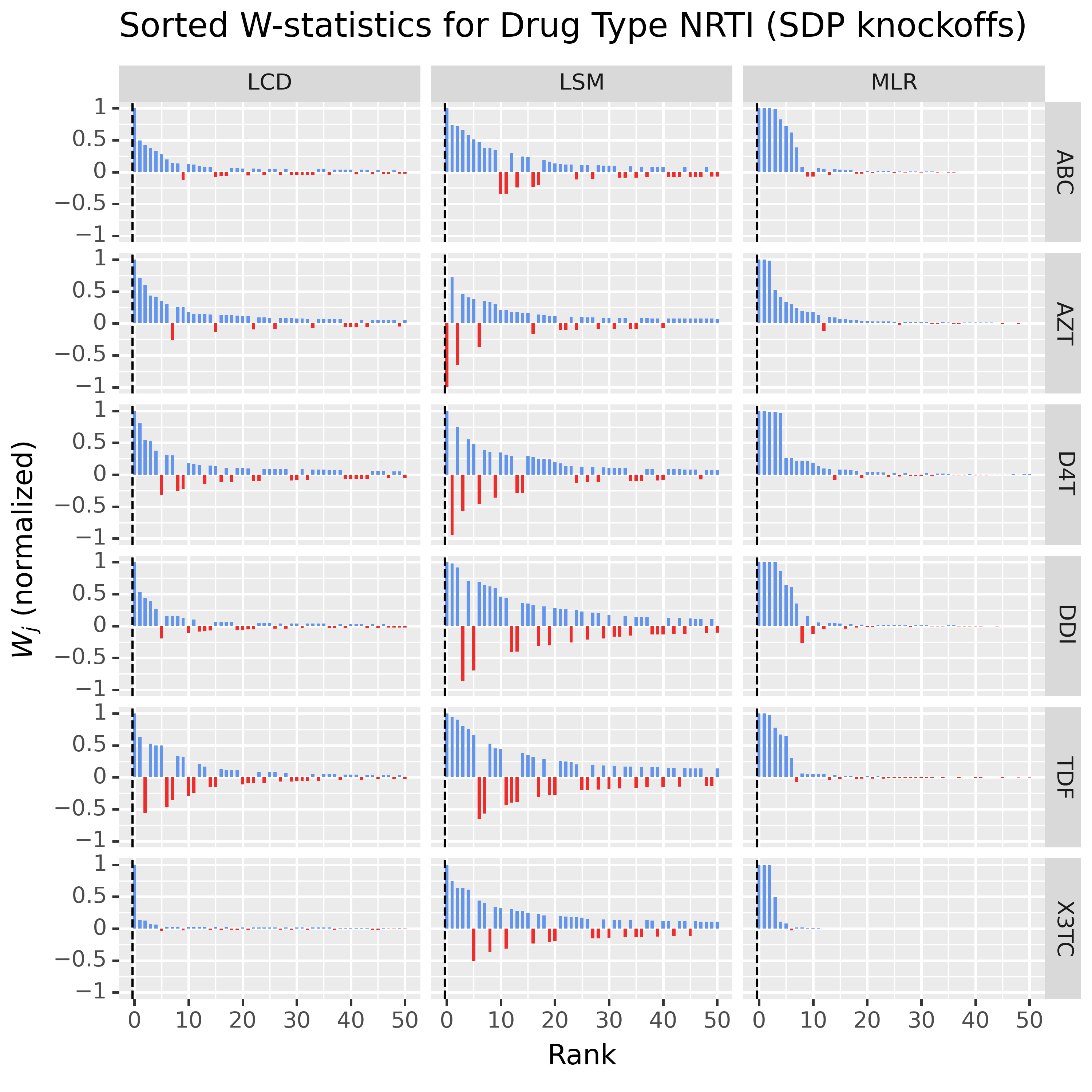}
    \caption{This figure is the same as Figure \ref{fig::wstatplot}, except it shows results for drugs in the NRTI class using SDP knockoffs.}
    \label{fig::hiv_nrti_sdp}
\end{figure}

\begin{figure}
    \centering
    \includegraphics{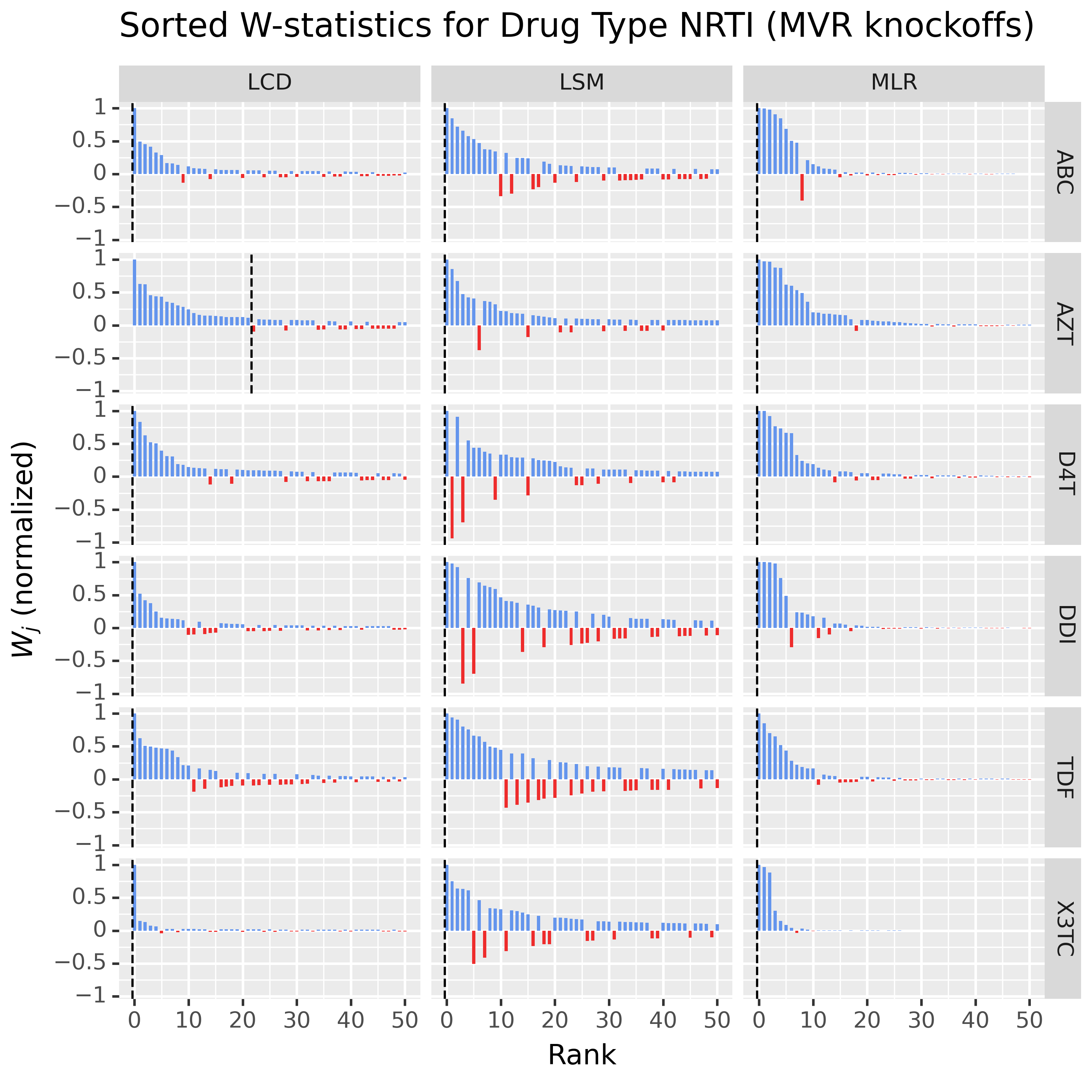}
    \caption{This figure is the same as Figure \ref{fig::wstatplot}, except it shows results for drugs in the NRTI class using MVR knockoffs.}
    \label{fig::hiv_nrti_mvr}
\end{figure}

\begin{figure}
    \centering
    \includegraphics{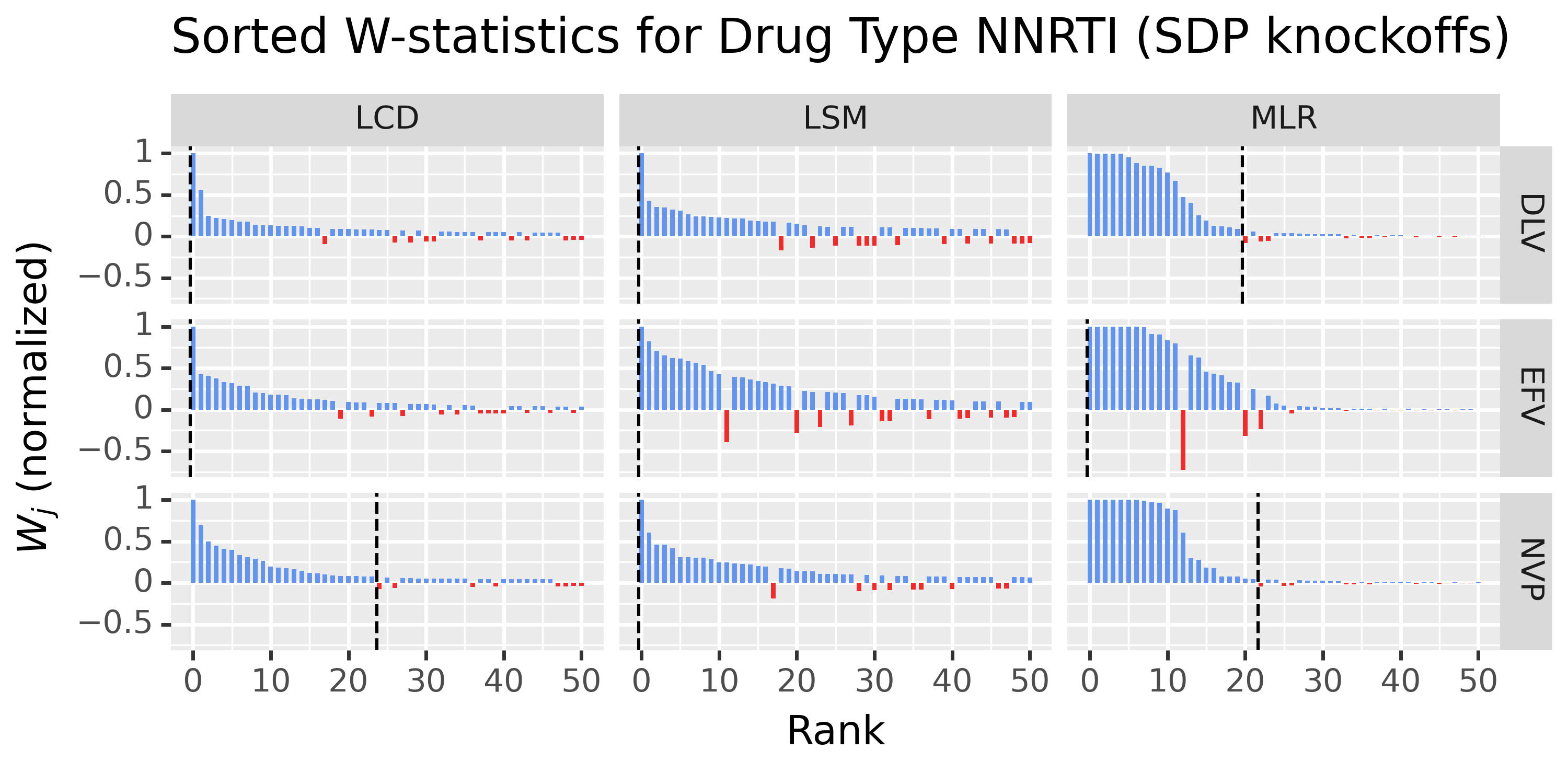}
    \caption{This figure is the same as Figure \ref{fig::wstatplot}, except it shows results for drugs in the NNRTI class using SDP knockoffs.}
    \label{fig::hiv_nnrti_sdp}
\end{figure}

\begin{figure}
    \centering
    \includegraphics{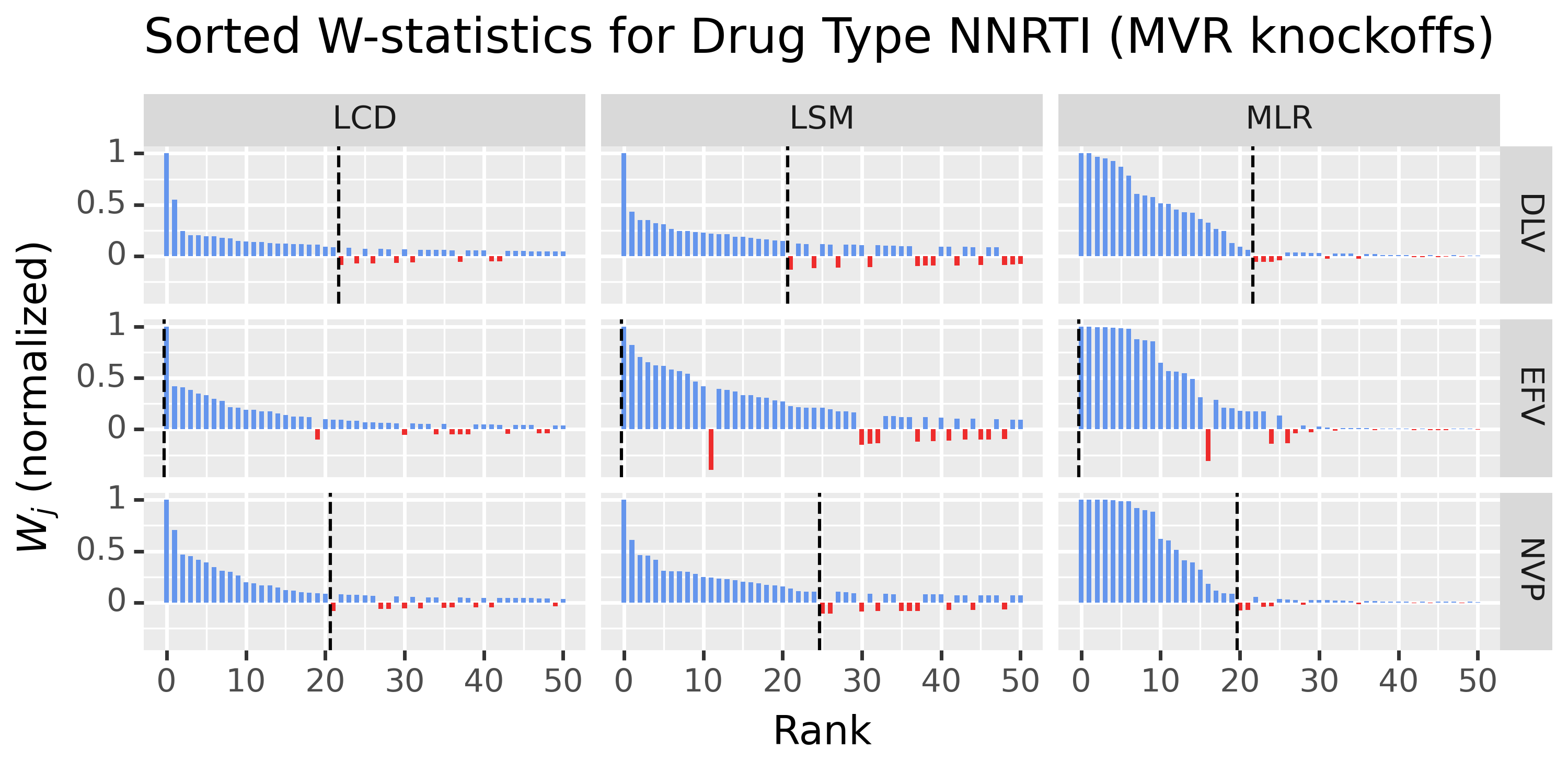}
    \caption{This figure is the same as Figure \ref{fig::wstatplot}, except it shows results for drugs in the NNRTI class using SDP knockoffs.}
    \label{fig::hiv_nnrti_mvr}
\end{figure}

\newpage

\subsection{Financial factor selection}

We now present a few additional details for the financial factor selection analysis from Section \ref{subsec::fundrep}. First, we list the ten index funds we analyze, which are: XLB (materials), XLC (communication services), XLE (energy), XLF (financials), XLK (information technology), XLP (consumer staples), XLRE (real estate), XLU (utilities), XLV (health care), and XLY (consumer discretionary). Second, for each feature statistic, Table \ref{tab::fundfdp} shows the average realized FDP across all ten analyses---as desired, the average FDP for each method is lower than the nominal level of $q=0.05\%$. All code is available at \url{https://github.com/amspector100/mlr_knockoff_paper}.
\begin{table}
    \centering
    \begin{tabular}{|l|l|l|r|}
    \hline
    Knockoff Type & Feature Stat. & Average FDP \\ \hline
    MVR & LCD &  0.013636 \\ \hline
        & LSM &  0.004545 \\ \hline
        & MLR &  0.038571 \\ \hline 
    SDP & LCD &  0.000000 \\ \hline 
        & LSM &  0.035000 \\ \hline 
        & MLR &  0.039002 \\ \hline
    \end{tabular}
    \caption{This table shows the average FDP, defined above, for each method in the financial factor selection analysis from Section \ref{subsec::fundrep}.}
    \label{tab::fundfdp}
\end{table}

\end{document}